\documentclass[11pt]{article}
\usepackage{latexsym}
\usepackage{amssymb,amsmath}
\usepackage{amsthm}
\usepackage{mathtools}
\usepackage{graphicx}
\usepackage{color}
\usepackage{blkarray}
\usepackage{multirow}
\usepackage{hyperref}
\usepackage{float}
\usepackage{subfig}
\usepackage{relsize}

\usepackage{cleveref}

\usepackage[numbers,sort&compress]{natbib}

\hypersetup{
  colorlinks   = true, 
  urlcolor     = blue, 
  linkcolor    = blue, 
  citecolor   = red 
}

\usepackage{breqn}

\usepackage{xcolor}
\usepackage{cancel}

\usepackage{hyperref}
\usepackage{empheq}

\usepackage[utf8]{inputenc}

\usepackage{setspace}

\usepackage{enumerate}

\usepackage[titletoc,toc]{appendix}

\usepackage[font=small,labelfont=bf]{caption}

\newcommand{\ds}{\displaystyle}
\newcommand{\mc}{\mathcal}

\DeclareMathOperator{\argmax}{argmax}
\DeclareMathOperator{\supp}{supp}

\DeclareMathOperator{\argsup}{argsup}

\newcommand{\bbm}{\begin{bmatrix}}
\newcommand{\bpm}{\begin{pmatrix}}
\newcommand{\ebm}{\end{bmatrix}}
\newcommand{\epm}{\end{pmatrix}}

 \newcommand{\del}[2]{\frac{\partial #1}{\partial #2}}
 \newcommand{\dsdel}[2]{\displaystyle\frac{\partial #1}{\partial #2}}

\newcommand{\ddx}[2]{\frac{d #1}{d #2}}
\newcommand{\ddt}[1]{\frac{d #1}{dt}}

\newcommand{\dsddx}[2]{\displaystyle\frac{d #1}{d #2}}
\newcommand{\dsddt}[1]{\displaystyle\frac{d #1}{dt}}

\newcommand{\gronwalls}{Gr\"onwall's  \hspace{0.05mm}}
\newcommand{\holder}{Hölder  \hspace{0.05mm}}

\newcommand*{\defeq}{\mathrel{\vcenter{\baselineskip0.5ex \lineskiplimit0pt
                     \hbox{\scriptsize.}\hbox{\scriptsize.}}}%
                     =}

\renewcommand{\abstractname}{Abstract}
\numberwithin{equation}{section}

\newcommand{\tGamma}{T_{\Gamma}^{\delta}}

\usepackage[T1]{fontenc}
\usepackage[utf8]{inputenc}
\usepackage{authblk}

\title{\large{Analysis of Multilevel Replicator Dynamics for General Two-Strategy Social Dilemmas}}
\author[1]{Daniel B. Cooney}

\affil[1]{Program in Applied and Computational Mathematics, Princeton University}

\usepackage{minitoc}

\newcommand{\myindent}{\hspace{10mm}}

\begin{document}

\renewcommand{\baselinestretch}{1.1}
\newtheorem{definition}{Definition}[section]
\newtheorem{theorem}{Theorem}[section]
\newtheorem{lemma}[theorem]{Lemma}
\newtheorem{corollary}[theorem]{Corollary}
\newtheorem{claim}[theorem]{Claim}
\newtheorem{fact}[theorem]{Fact}
\newtheorem{proposition}{Proposition}[section]
\newtheorem{remark}{Remark}[section]
\newtheorem{example}{Example}[section]
\newtheorem{conjecture}{Conjecture}[section]

\newcommand{\qedsymb}{\mbox{ }~\hfill~{\rule{2mm}{2mm}}}
\newenvironment{proof1}{\begin{trivlist}
\item[\hspace{\labelsep}{\bf\noindent Proof: }]
}{\qedsymb\end{trivlist}}

\newenvironment{hackyproof}{\begin{trivlist}
\item[\hspace{\labelsep}{\bf\noindent Proof: }]
}{
\end{trivlist}}

\maketitle

\begin{abstract}
Here we consider a game-theoretic model of multilevel selection in which individuals compete based on their payoff and groups also compete based on the average payoff of group members. Our focus is on multilevel social dilemmas: games in which individuals are best off cheating, while groups of individuals do best when composed of many cooperators. We analyze the dynamics of the two-level replicator dynamics, a nonlocal hyperbolic PDE describing deterministic birth-death dynamics for both individuals and groups. While past work on such multilevel dynamics has restricted attention to scenarios with exactly solvable within-group dynamics, we use comparison principles and an invariant property of the tail of the population distribution to extend our analysis to all possible two-player, two-strategy social dilemmas.  In the Stag-Hunt and similar games with coordination thresholds, we show that any amount of between-group competition allows for fixation of cooperation in the population. For the Prisoners' Dilemma and Hawk-Dove game, we characterize the threshold level of between-group selection dividing a regime in which the population converges to a delta function at the equilibrium of the within-group dynamics from a regime in which between-group competition facilitates the existence of steady-state densities supporting greater levels of cooperation. In particular, we see that the threshold selection strength and average payoff at steady state depend on a tug-of-war between the individual-level incentive to be a defector in a many-cooperator group and the group-level incentive to have many cooperators over many defectors. %
We also find that lower-level selection casts a long shadow: if groups are best off with a mix of cooperators
and defectors, then there will always be fewer cooperators than optimal at steady state, even in the limit of infinitely strong competition between groups.
\end{abstract}

\singlespacing
{\hypersetup{linkbordercolor=black, linkcolor = black}
\begin{spacing}{0.01}
\renewcommand{\baselinestretch}{0.1}\normalsize
\tableofcontents
\addtocontents{toc}{\protect\setcounter{tocdepth}{2}}
\end{spacing}

\section{Introduction}

\myindent Across a variety of natural systems, a common theme of study is the conflict of selective pressures acting at mutliple oraganizational levels. Multilevel selection has been used as a framework to study problems of plasmid loss and plasmid segregation \cite{paulsson2002multileveled}, the evolution of RNA viruses \cite{turner1999prisoner}, the evolution of virulence \cite{gilchrist2004optimizing,coombs2007evaluating,levin1981selection}, and the cooperative founding of ant colonies \cite{shaffer2016foundress}. Forms of complex life exist as nested hierarchies of self-replicating units, and, throughout the course of evolution, individual entities have joined together to create higher-level self-replicating units corresponding to an evolutionary transition in biological complexity \cite{szathmary1995major}. In the process of that transition, and even in the coexistence of two levels of selection thereafter, there can be a conflict between the incentives of individuals at the lower level and the incentives of groups of individuals at the higher level. 

\myindent As a mathematical framework to compare the incentives of the individual and the incentives of the group, we look to evolutionary game theory \cite{nowak2006evolutionary,nowak2006five}. In the classical setting of the Prisoners' Dilemma, individuals can either cooperate by paying a cost to confer a benefit to their coplayers, or they can defect by paying no cost and conferring no benefit. The ``dilemma'' faced by the two players is that each is individually better off defecting regardless of their opponent's action, but the individuals would receive a greater payoff by both cooperating than they would by both defecting. Scaling up to questions about the interplay of individuals and groups, groups tend to have high average or total payoff when there are many cooperators, so what is benefecial to the group runs against what is beneficial to the individual. Multilevel selection was studied as a mechanism to promote cooperation by Traulsen and Nowak \cite{traulsen2005stochastic,traulsen2006evolution,traulsen2008analytical}, showing that group-level reproduction or fission events could help to promote the fixation probability of a single cooperator in a group-structured population otherwise composed of defectors. Further work in this area has included extensions describibg the role of spatial structure in between-group competition \cite{akdeniz2019cancellation}, allowing for more general payoff matrices and frequency-dependence, and including mechanisms like a group extinction inversely proportional to collective payoff \cite{bottcher2016promotion}.

\myindent Luo introduced a stochastic framework for multilevel selection \cite{luo2014unifying}, in which a group-structured population contained two types of individuals: one that had a selective advantage for individual-level replication and another that conferred an advantage to the group during group-level replication events. Taking the limit as group size and number of groups to infinity, Luo derived a nonlocal PDE describing the changing probability density of the composition of groups. This model was then applied to discuss the debate about individual selection, group selection, and kin selection \cite{van2014simple} and to discuss fixation probabilities in a multilevel model of Hawk-Dove games \cite{mcloone2018stochasticity}. Luo and Mattingly further analyzed the long-time behavior of the frequency-independent multilevel selection model \cite{luo2017scaling}, showing that there was a threshold level of relative selection intensity below which the individually-advantageous type fixed in all groups, and above which a density of all compositions of groups survived in steady state. 

\myindent Simon and coauthors have also studied two-level selection models derived from underlying stochastic descriptions, exploring a variety of group-level evolutionary mechanisms including fission and fusion events, as well as allowing for heterogeneous distributions of group sizes and changing number of groups \cite{simon2010dynamical,simon2012numerical,simon2013towards,simon2016group,puhalskii2017large}. A similar model was introduced to study host-parasite evolution in the microbiome \cite{van2019role}, incorporating within-host replication of microbes, between-host horizontal transmission, and vertical transmission via host reproduction. Velleret has further explored quasi-stationary distributions in an alternate stochastic Fleming-Viot scaling limit of Luo and Mattingly's model \cite{velleret2019two}. Pokalyuk and coauthors have used Luo's ball-and-urn framework to describe a multilevel host-pathogen system, and derived limiting ODE descriptions of the multilevel dynamics which show the long-time persistence of pathogen diversity in the presence of stabilizing selection \cite{pokalyuk2019maintenance} and in a case of neutral competition between pathogens \cite{pokalyuk2019diversity}.

\myindent In a recent paper, Luo's framework was extended by Cooney to describe multilevel selection in which the underlying birth rates at the within-group and between-group levels depends on payoffs from an underlying game \cite{cooney2019replicator}. In the large population limit, within-group dynamics favor defectors in the Prisoners' Dilemma and an equilibrium mix of cooperators and defectors in the Hawk-Dove game. Competition between groups was based on the average payoff of group members, and, depending on the game's payoff matrix, it was possible for a group's payoff to be maximized by a full-cooperator group or by a group featuring a mix of cooperators and defectors. %
For games in which average group payoff was maximized by full-cooperator groups, arbitrary levels of cooperation could be achieved at steady state in the presence of sufficiently strong between-group competition. For games in which average payoff was maximized by a mix of cooperators and defectors, no level of between-group selection intensity could result in an optimal level of cooperation. In such games, within-group competition cast a long shadow: the individual incentives promoting defection remain evident even in the limit in which between-group competition is infinitely stronger than within-group competition. 

\myindent In that paper, the analysis of long-time behavior of multilevel dynamics was restricted to a family of games for which the within-group replicator dynamics were exactly solvable. While this assumption was convenient to allow for direct use of the method of characteristics, it did not provide the possibility of characterizing the long-time behavior of all multilevel two-strategy games and to determine if the shadow of lower-level selection is a generic feature of mulitlevel replicator dynamics. In this paper, we make use of comparison principles to extend our analysis to games without solvable within-group dynamics, which allows for the analysis all two-strategy social dilemmas and introduces an approach for exploring more general deterministic models of multilevel selection across a variety of biological settings. By considering more general two-strategy social dillemas, we are able to analytically characterize how the conflict between individual incentive to defect %
and the group incentive to have many cooperators%
determines the average payoff of steady state population and the threshold level of selection intensity required to sustain cooperation at steady state.

 \myindent The rest of the paper is structured as follows. In Section \ref{sec:Baseline}, we define our evolutionary game, describe the two-level Moran process governing individual and group birth and deaths events, and then discuss the PDE description of the multilevel system in the limit of many groups and large group size. In Section \ref{sec:MultilevelDynamics}, we discuss strategies for studying the multilevel dynamics for cases in which the within-group dynamics are not exactly solvable. In Section \ref{sec:GeneralPD}, we focus on the Prisoners' Dilemma, proving that there is a threshold strength of between-group competition separating a regime in which the population converges to a delta concentration at the equilibrium of the within-group dynamics from a regime in which there exist steady-state densities that can support greater levels of cooperation. In Section \ref{sec:GeneralHD}, we perform a similar analysis for the Hawk-Dove game, and in Section \ref{sec:SHandother}, we consider the multilevel dynamics for the Stag Hunt game and the other two-strategy two-player social dilemmas. We discuss the results and implications for future research in Section \ref{sec:Discussion},  %
and in the Appendix we address well-posedness of solutions to our measure-valued PDE in Section \ref{sec:measureexistence} %
and include detailed calculations of integrals along characteristic curves in Section \ref{sec:integrals}.

\section{Baseline Model} \label{sec:Baseline}

We consider a two-player game with two strategies: cooperate ($C$) and defect ($D$). Individuals receive payoff according to the following payoff matrix 
\begin{equation} \label{eq:generalpayoffmatrix}
\begin{blockarray}{ccc}
& C & D \\
\begin{block}{c(cc)}
C & R & S \\
D & T & P \\
\end{block}
\end{blockarray},
\end{equation}
where a cooperator receives a reward $R$ for cooperating with a cooperator and receives a sucker payoff $S$ for cooperating with a defector, while a defector receives a temptation payoff $T$ for defecting against a cooperation and receives a punishment $P$ for defecting against another defector. 
The Prisoners' Dilemma (PD), Hawk-Dove (HD) game, and Stag-Hunt (SH) are characterized by the following rankings of payoffs 
\begin{subequations} \begin{align} PD  :& \: T > R > P > S  \\ HD :&  \: T > R > S > P \\ SH :& \: R > T > P > S \end{align} \end{subequations} 
We assume that individuals play a game with the payoff matrix of Equation \ref{eq:generalpayoffmatrix} against each other member of their group. 
In a group with a fraction $x$ of cooperators, a cooperator and defector receive average payoffs of \begin{subequations} \begin{align} \pi_C(x) &= Rx + S(1-x) \label{eq:piC} \\  \pi_D(x) &= Tx + P(1-x) \label{eq:piD} \end{align} \end{subequations} The average payoff in a group with fraction $x$ cooperators is \begin{dmath} \label{eq:grouppayoff} G(x) \defeq x \pi_C(x) + (1-x) \pi_D(x) \\ = P +  \left( S + T - 2 P \right) x + \left( R - S - T + P \right) x^2  \end{dmath} 

To describe the simultaneous competition between individuals and competition between groups, one can introduce a nested birth-death process with reproduction events occuring both for individuals and for groups \cite{luo2014unifying,cooney2019replicator}. Individuals birth-death dynamics follow a continuous-time Moran process, in which individual cooperators and defectors in an $x$-cooperator group give birth and replace a randomly chosen individual at rate $1 + w_I \pi_C(x)$ and $1 + w_I \pi_D(x)$,  where $w_I$ is the selection strength of individual-level competition relative to the neutral birth rate $1$. Group-level reproduction events follow a Moran-type process as well, with $x$-cooperator groups reproducing and replacing a randomly chosen group with rate $\Lambda \left( 1 + w_G G(x) \right)$, where  $w_G$ denotes the selection strength of between-group competition  and $\Lambda$ describes a relative timescale of within-group and between-group replication events. %

In the limit of infinitely many groups and of infinite group size, one can use the approach of Luo \cite{luo2014unifying} or Cooney \cite{cooney2019replicator} to derive a determinisitc PDE description of the two-level population dynamics, which is given by
\begin{equation} \label{eq:replicatorpde} \dsdel{f(t,x)}{t} = - \dsdel{}{x} \left[ x(1-x) \left( \pi_C(x) - \pi_D(x) \right) f(t,x) \right] + \lambda f(t,x)  \left[G(x) - \int_0^1 G(y) f(t,y) dy \right], \end{equation}
where $f(t,x)$ is the probability density of groups with fraction $x$ cooperators at time $t$ and $\lambda := \frac{w_G \Lambda}{w_I}$ describes the relative effects of competition between individuals and competition between groups. This equation can be thought of as the two-level version of the replicator equation, and can be used to describe multilevel selection for a broad range of models for the payoffs of cooperators $\pi_C(x)$ and defectors $\pi_D(x)$ and on the average payoff of group members $G(x)$. It is a first-order PDE with a linear advection term $\del{}{x}\left[x(1-x) \left(\pi_C(x) - \pi_D(x) \right) \right]$, describing the effect of individual birth and death due to within-group competition, and a nonlocal and nonlinear term $\lambda f(t,x) \left[G(x) - \int_0^1 G(y) f(t,y) dy \right]$ describing group-level birth and death due to between-group competition. To further illustrate the role of within-group competition, we can rewrite Equation \ref{eq:replicatorpde} in the suggestive form 
\begin{align} \label{eq:replicatorpdediv} \dsdel{f(t,x)}{t} + \left[x(1-x)\left( \pi^C(x) -  \pi^D(x) \right) \right] \dsdel{f(t,x)}{x} &= - f(t,x) \dsdel{}{x} \left[ x(1-x) \left( \pi^C(x) - \pi^D(x) \right) \right]  \nonumber \\ &+  \lambda f(t,x)  \left[G(x) - \int_0^1 G(y) f(t,y) dy \right] \end{align}
In this form, we see that the characteristic curves of Equation \ref{eq:replicatorpde} are given by the replicator dynamics for individual-level competition in a single group 
\begin{equation} \label{eq:replicatorODE} \dsddt{x(t)} = x (1-x) \left( \pi_C(x) - \pi_D(x) \right)
\end{equation}
Introduced by Taylor and Jonker, the replicator equation is a classic tool to describe deterministic strategic dynamics in evolutionary game theory \cite{taylor1978evolutionary,schuster1983replicator,cressman2014replicator}. For two-strategy games with the payoff matrix of Equation \ref{eq:generalpayoffmatrix}, there are four generic behaviors of the within-group replicator equation: dominance of defectors and global convergence to the all-defector equilibrium (Prisoners' Dilemma), stable coexistence of cooperators and defectors at equilibrium $x^{eq}$ such that $\pi_C(x^{eq}) = \pi_D(x^{eq})$ (Hawk-Dove), bistability of all-defector equilibrium and all-cooperator equilibrium (Stag-Hunt), or dominance of cooperators and global convergence to all-cooperator equilibrium (Prisoners' Delight). In the presence of between-group selection, we expect the multilevel dynamics of Equation \ref{eq:replicatorpde} to cause tension between individual-level dynamics favoring stable equilibrium compositions of Equation \ref{eq:replicatorODE} and between-group dynamics favoring compositions with high average payoff. 

In the scenario we study with payoff matrix given by Equation \ref{eq:generalpayoffmatrix}, we can use the shorthand from \cite{cooney2019replicator} $\alpha = R - S -T + P$, $\beta = S-P$, and $\gamma = S + T - 2P$ to describe the key quantities characterizing within-group and between-group competition. For within-group dynamics where it will be useful for us to note that \begin{equation} \label{eq:picminuspid} \pi_C(x) - \pi_D(x) = \beta + \alpha x  \end{equation} represents the relative advantage (or disadvantage) of a cooperator's payoff relative to a defector in a group with a fraction of $x$ cooperators. In this notation, the average payoff of members in a group composed of fraction $x$ cooperators can be written as
\begin{equation} \label{eq:grouppayoffparam} G(x) = P + \gamma x + \alpha x^2, \end{equation}
while the average payoff of the whole population can be found by integrating $G(x)$ against the probability density of group compositions $f(t,x)$ in the population, yielding 
\begin{equation} \label{eq:payoffmomentsparam} \langle G(\cdot) \rangle_{f(t,x)} := \int_0^1  G(y)  f(t,y)  dy = P + \gamma \int_0^1 y f(t,y) dy + \alpha \int_0^1 y^2 f(t,y) dy \end{equation}
 Denoting the $j$th moments of $f(t,x)$ by $M_j^f(t) = \int_0^1 y^j f(t,y) dy$, we can use our expressions for $\pi_C(x) - \pi_D(x)$, $G(x)$, and $\int_0^1 G(y) f(t,y) dy$ from Equations \ref{eq:picminuspid}, \ref{eq:grouppayoffparam}, and \ref{eq:payoffmomentsparam} to describe the dynamics of our multilevel system by the equation 
\begin{equation} \label{eq:replicatorpdeparam} \dsdel{f(t,x)}{t} = -  \dsdel{}{x} \left[ x(1-x) \left( \beta + \alpha x \right) f(t,x) \right] + \lambda f(t,x)  \left[\gamma x + \alpha x^2 - \left( \gamma M_1^f(t) + \alpha M_2^f(t) \right) \right], \end{equation} 
Because the constant term $P$ has no effect on the maximizer $x^*$ of $G(x)$ and cancels out from the expression $G(x) - \int_0^1 G(y) f(t,y) dy$ in the righthand side of Equation \ref{eq:replicatorpdeparam}, we can choose to understand the dynamics of between-group competition through the simplified group average payoff function 
\begin{equation} \label{eq:grouppayoffparamsimplified} G(x) = \gamma x + \alpha x^2. \end{equation}
This will allow us to more succintly characterize the effect of between-group competition along the characteristic curves of Equation \ref{eq:replicatorpdeparam}, so we will use this as our expression for $G(x)$ in all subsequent analysis.

In the absence of between-group competition, the densities are expected to cocentrate at the equilibria of the within-group replicator dynamics of Equation \ref{eq:replicatorODE}. To further study the possibility that the distribution of group types can concetrate as delta-functions at a given composition of cooperators, we can introduce, as in \cite{luo2017scaling,cooney2019replicator}, a weak, measure-valued formulation of the multilevel dynamics. Introducing $\mu_t(dx)$, the measure of group compositions $x$ at time $t$, and a $C^1$ test function $\psi(x)$, we can recharacterize the dynamics of our multilevel system as follows
\begin{equation} \label{eq:replicatormeasurepde} \dsdel{}{t} \ds\int_0^1 \psi(x) \mu_t(dx) =  \ds\int_0^1 \left\{ \dsdel{\psi(x)}{x} \left[ x(1-x) \left(\beta + \alpha x\right) \right] + \lambda \left[\gamma x + \alpha x^2 - \left(\gamma M_1^{\mu}(t) + \alpha M_2^{\mu}(t)  \right) \right] \right\} \mu_t(dx) \end{equation} 
where $M_j^{\mu}(t) = \int_0^1 x^j \mu_t(dx)$ denotes the $j$th moment with respect to $\mu_t(dx)$. As shown in \cite{cooney2019replicator}, delta-functions at the within-group equilibria ($\delta(x)$, $\delta(1-x)$, and $\delta\left(\beta - \alpha x\right)$) are fixed points of Equation \ref{eq:replicatormeasurepde}. By choosing the test function $\psi(x) \equiv 1$, we see that a measure $\mu_t(dx)$ solving Equation \ref{eq:replicatormeasurepde} will satisfy $\ddt{} \int_0^1 \mu_t(dx) = 0$ if $\int_0^1 \mu_t(dx)$, and therefore the measure $\mu_t(dx)$ will remain normalized if the initial measure satisfies $\int_0^1 \mu_0(dx) = 1$. 

To explore solutions of Equation \ref{eq:replicatormeasurepde}, we can use the idea of the method of characteristics to solve between-group dynamics along the solution curves of the within-group dynamics. We can describe the effect of between-group competition on solutions along characteristic curves $\phi_t(x)$ as 

\begin{equation} \label{eq:wtx} w_t(\phi_t(x)) =  \exp\left( \lambda \int_0^t \left[ G(\phi_t(x)) - \langle G(\cdot) \rangle_{\mu_s} \right]  ds \right), \end{equation}
which can also be written more explicitly by writing $G(x) = P +  \gamma x + \alpha x^2$ as 
\begin{equation} \label{eq:wtxmoment} w_t(\phi_t(x)) =  \exp\left( \lambda \int_0^t \left[\gamma \phi_t(x) + \alpha \phi_t(x)^2 - \left( \gamma M_1^{\mu_s} + \alpha M_2^{\mu_s} \right) \right] ds   \right). \end{equation}
We can use the idea of a push-forward measure to describe the effect of within-group dynamics on the change from the initial distribution $\mu_0(dx)$ along characteristic curves $\phi_t(x)$ \cite{luo2017scaling,evers2016mild,canizo2011well}. We characterize the push-forward measure of $\mu_0(dx)$ under the dynamics of $\phi_t(x)$ using the equivalent notations of 
\begin{equation*} \label{eq:pushforward} P_t \mu_0(dx)  = \left[ \mu_0 \circ \phi_t^{-1}  \right](dx) \end{equation*}
In the weak, measure-valued formulation, we can characterize the effect of the push-forward measure by its effect on a test function $\psi(x)$, which takes the form 
\[ \int_0^1 \psi(x) P_t \mu_0(dx) = \int_0^1 \psi(x) \left[ \mu_0 \circ \phi_t^{-1} \right] (dx) = \int_0^1 \psi(\phi_t(x)) \mu_0(dx) \]
Combining the effects of the within-group and between-group dynamics, we can arrive at the following implicit representation formula for the population distribution at time $t$ 
\begin{equation} \label{eq:mutimplicit}  \mu_t(dx) = w_t(x) \left[\mu_0 \circ \phi_t^{-1} \right] (dx). \end{equation}
In the weak formulation, this relation can be expressed as 
\begin{equation} \label{eq:mutimplicitweak} \int_0^1 \psi(x) \mu_t(dx) = \int_0^1  \psi(x) w_t(x) \left[\mu_0 \circ \phi_t^{-1} \right] (dx) = \int_0^1 \psi(\phi_t(x)) w_t(\phi_t(x)) \mu_0(dx)   \end{equation}
The representation formula in Equation \ref{eq:mutimplicitweak} is particularly helpful in describing the long-time behavior of the two-level dynamics. Because it is an implicit expression, we would like to prove that their exists such a $\mu_t(dx)$ satisfying Equation \ref{eq:mutimplicitweak}, and that this $\mu_t(dx)$ is the unique solution of Equation \ref{eq:replicatormeasurepde}. We address the well-posedness of Equation \ref{eq:replicatormeasurepde} in Section \ref{sec:measureexistence}, making use of the method of characteristics and the Banach fixed-point theorem.

\subsection{Possible Group-Level Reproduction Functions} \label{sec:grouppayoff}

A relevant property of the underlying game for solutions of Equation \ref{eq:replicatorpde} is whether the group-level reproduction rate $G(x) =  \gamma x + \alpha x^2$ is maximized by full-cooperator groups or by a group composition with a mix of cooperators and defectors. If $x^* = \argmax_{x \in [0,1]}\left(G(x)\right) = 1$, then all-cooperator groups are most favored at the between-group level and between-group selection pushes for as much cooperation as possible. %
When $x^* < 1$,  between-group selection instead most favors groups composed of $x^*$ fraction cooperators and $1-x^*$ fraction defectors, pushing most strongly for a groups with an intermediate level of cooperation $x^*$. 

For the Prisoners' Dilemma, our classification of the composition $x^*$ maximizing collective payoff depends on the signs of $\gamma$ and $\alpha$. Here we list the four possible cases of $\gamma$ and $\alpha$ and the impact of payoff maximizer $x^*$. 
\begin{itemize}
\item[Case I:] If $\gamma > 0$, $\alpha < 0$, %
then $G''(x) = \alpha < 0$ and $G(x)$ has a local maximum at $- \frac{\gamma}{2 \alpha} > 0$, so either $G(x)$ has an interior maximum or $G(x)$ is increasing for every $x \in [0,1]$. We can characterize the group type $x^*$ with maximal average payoff in terms of either its dependence on $\gamma$ and $\alpha$ or the values of the payoff matrix, yielding
   \begin{displaymath}
   x^* = \left\{
     \begin{array}{cr}
       - \ds\frac{\gamma}{2 \alpha} & :  \gamma < - 2 \alpha \\
       1 & : \gamma \geq- 2 \alpha
     \end{array}
   \right. \Longrightarrow
   x^* = \left\{
     \begin{array}{cr}
        \ds\frac{S + T - 2P}{2 (-R + S + T - P)} & :  2R < T + S \\
       1 & :  2R \geq T  + S
     \end{array}
   \right. 
\end{displaymath} 
In particular, we notice that $x^* = 1$ when $2R \geq T + S$, so between-group competition best favors full-cooperator groups when the total payoff from the interaction of two cooperators contributes more to the collective payoff than the total payoff generated by the interaction of a cooperator and a defector. When $2R < T + S$ and a cooperator-defector pair can outcontribute two cooperators at the group level, we instead have that the most favored group composition $x^*$ feaures a nonzero level of defection. To highlight these two possible optimal compositions, we can further subdivide Case I into Case Ia in which there is intermediate collective payoff optimum $x^* < 1$ and Case Ib in which full cooperator groups $x^* = 1$ maximized collective payoff. 
\item[Case II:] If $\gamma > 0$, $\alpha = 0$, we have $G(x) =  \gamma x$ and we obtain a rescaled version of the Luo-Mattingly model \cite{luo2014unifying,luo2017scaling}. For this case, group payoff is maximized by the full-cooperator group with $x^* = 1$, $\forall \gamma > 0$. 
\item[Case III:] If $\gamma, \alpha > 0$, then $G(x)$ is increasing on $[0,1]$ and is maximized at $x^* = 1$.
\item[Case IV:] If $\gamma < 0$, then $\alpha > 0$, as %
$\alpha = (R - P) - (S + T - 2P) > 0$. We again have $x^* = 1$,  because $G(1) =   \gamma + \alpha  = R - P > 0 = G(0)$ and $G''(x) = 2 \alpha >  0$, so $-\frac{\gamma}{2 \alpha}$ would be a local minimum even if it were in $(0,1)$.
%

%
\end{itemize}
In Figure \ref{fig:pdcases}, we illustrate example group payoff functions for Cases Ia-IV (left) and display the regions in $(\gamma,\alpha$) parameter space which constitute different cases of the PD. The gray region below the dotted line corresponds to games which are not PDs, while the vertical dashed line corresponds to the measure-zero region describe the Case II PD with $\gamma > 0$, $\alpha = 0$. 
\begin{figure}[H]
    \centering
 \hspace{-5mm}\includegraphics[width=0.5\textwidth]{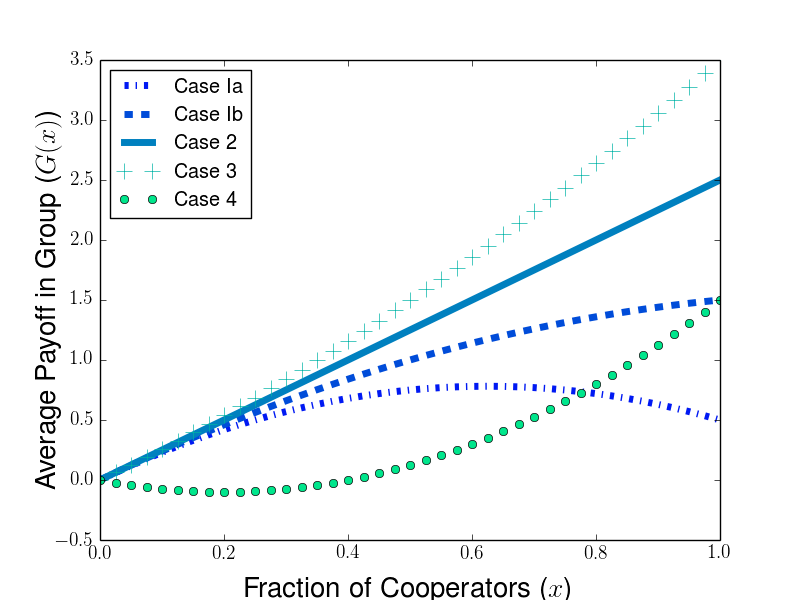} 
\hspace{-5mm}  \includegraphics[width=0.495\textwidth]{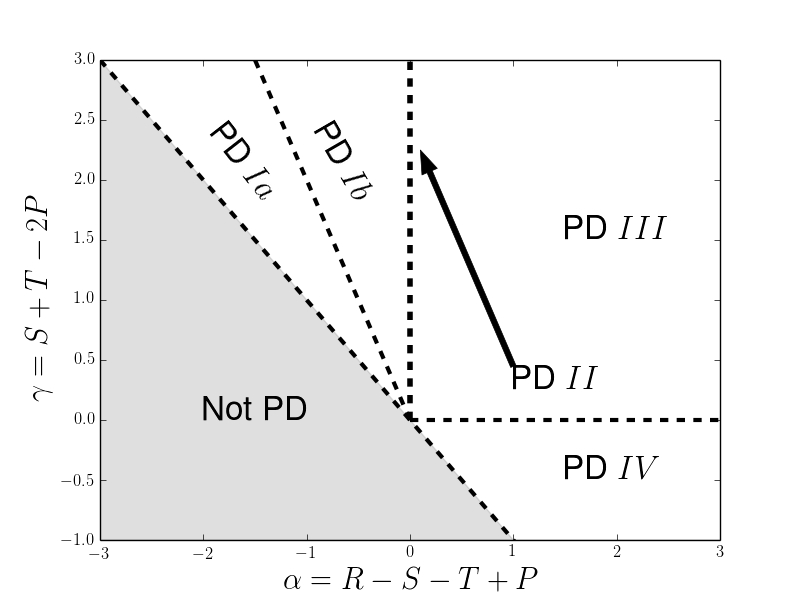} 
    \caption{Illustration of different cases of the Prisoners' Dilemma. (Left) Example group payoff functions $G(x)$ for each case. For Cases Ia-III, we chose $\gamma = 2.5$, and $\gamma = -1$ for Case IV. Case Ia (dash-dotted line, $\alpha = -2$), group payoff is maximized at $x^* = 5/8$ ; Case Ib (dashed line, $\alpha = -1.0$), group payoff increases sublinearily and is maximized at $x^* = 1$; Case II (solid line, $\alpha = 0$), group payoff increases linearly to maximum at $x^* = 1$; Case III (plus signs, $\alpha = 1$), group payoff increases superlinearly to maxomum at $x^*=1$; Case IV (green dots, $\alpha = 2.5$), group payoff is minimized at $x = 0.2$ and maximized at $x^* = 1$. (Right) Different cases of the PD illustrated for different values of $\gamma$ and $\alpha$. Below the dashed line $\gamma = - \alpha$, games are not PDs, and the boundary between Case Ia with interior group payoff optima $x^* \in (0,1)$ and Case Ib with edge optimum $x^* = 1$ is given by the dashed line $\gamma = - 2 \alpha$. The arrow indicates that Case II corresponds to the line segment on which $\gamma > 0$ and $\alpha = 0$.} 
    \label{fig:pdcases}
\end{figure}

For the Hawk-Dove game, %
the possible scenarios are simpler because we know that $\gamma  = (S-P) + (T-P) > 0$ and 
$\alpha = (R-T) + (P-S) < 0$. As for the Case I PD, we find that $G(x)$ has a local maximum  at $x^* = \frac{\gamma}{2 |\alpha|}$ when $G'(x) = \gamma - 2 |\alpha| = 0$, and we can find that the group type maximizing average payoff is given by
 \begin{displaymath}
   x^* = \left\{
     \begin{array}{cr}
       - \ds\frac{\gamma}{2 \alpha} & :  \gamma < - 2 \alpha \\
       1 & : \gamma \geq- 2 \alpha
     \end{array}
   \right.%
\end{displaymath} 
Because $\alpha < 0$ for the HD game, we can rewrite our expression for the most fit group type as $x^* = \min\left(\tfrac{\gamma}{2 |\alpha|},1 \right)$. In addition, we note that $G''(x) = 2 \alpha < 0$, and therefore the critical point $x^*$ is a maximum and we know that $G(x)$ is increasing $\forall x < x^*$. The within-group dynamics of the HD game have a stable fixed point at $x^{eq} = \tfrac{\beta}{|\alpha|}$. Because $T > S$ for the HD game, we have that $\gamma > 2 \beta$ (as $\gamma = S + T - 2P > 2(S-P) = 2 \beta$), and therefore we have that $x^* = \frac{\gamma}{2 |\alpha|} > \frac{2 \beta}{2 |\alpha|} = x^{eq}$. Because the average payof $G(x)$ is increasing for $x < x^*$, we know that $G(x)$ is increasing for all $x \in [0,x^{eq}]$, and therefore average group payoff is better at $x^{eq}$ than for any group composition with fewer cooperators that the within-group equilibrium.

For the SH game, we know from the defining payoff rankings that $\alpha = \left( R  - T\right) + \left( P - S \right) > 0$, that $\beta = S - P < 0$, and that the sign of $\gamma = \underbrace{\left( S - P \right)}_{< 0} + \underbrace{\left( T - P \right)}_{> 0}$ is inconclusive. Therefore we must consider separately the cases in which $\gamma > 0$ and $\gamma < 0$. When $\gamma > 0$, we see that $G(x) = \gamma x + \alpha x^2$ is increasing for all $x \in [0,1]$, and therefore $x^* = 1$. When $\gamma < 0$, we see from computing $G'(x) = -|\gamma| + 2 \alpha x$ that $G(x)$ is decreasing for $x = 0$, and, recalling the within-group equilibrium $x^{eq} = \frac{|\beta|}{\alpha}$, that \[G'(x^{eq}) = G'\left(\tfrac{|\beta|}{\alpha}\right) = -|\gamma| + 2 |\beta| = \left(S + T - 2P \right) + 2 \left( P - S \right) = T - S > 0  \]
so we know that $G(x)$ is increasing for $x \in [x^{eq},1]$. Furthermore, by checking the value of $G(x)$ at $x^{eq}$, we see that 
\[ G(x^{eq}) = -|\gamma| \left( \tfrac{|\beta|}{\alpha} \right) + \alpha \left(\tfrac{|\beta|}{\alpha} \right)^2 = \underbrace{\left(\frac{|\beta|}{\alpha}\right)}_{> 0} \underbrace{\left[|\beta| - |\gamma| \right]}_{= T - P > 0} > 0 = G(0),  \]
and we have for $y > x^{eq}$ that the group payoff at a given level of cooperation $y$ exceeds all of group payoffs for group compositions with fewer cooperators. As a consequence, group average payoff $G(x)$ is maximized by full-cooperator groups ($x^* = 1$) in the $\gamma < 0$ case as well. 

Having studied the possible cases for $x^*$ for the PD, HD, and SH games, we will see in subsequent sections that whether $x^* = 1$ or $x^* \in (0,1)$ gives a qualitative change in steady state behavior in the limit when $\lambda \to \infty$. In the case $x^* \in (0,1)$, the level of cooperation at steady state can never achieve optimal levels of cooperation in this limit, as the multilevel system cannot avoid the shadow of lower-level selection. 

\section{Dynamics of Multilevel System} \label{sec:MultilevelDynamics}

In this section, we characterize properties of the multilevel dynamics of Equation \ref{eq:replicatormeasurepde} that will be useful for analyzing the long-time behavior of both the PD and HD games. In particular, we explore tools which allow us to study the multilevel dynamics for cases in which the within-group replicator dynamics are not necessarily exactly solvable. These include comparison principles relating the exact within-group dynamics with simpler logistic ODEs with explicit solutions (Section \ref{sec:Comparison}), integrating solutions along the simpler solution curves (Section \ref{sec:solutionalong}), and shows the preservation of the tail behavior of the population distribution near the full-cooperator equilibrium (Section \ref{sec:Holderpreserve}). 

\subsection{Comparison Principles for Characteristic Curves} \label{sec:Comparison}

Because we aim to analyze Prisoners' Dilemmas and Hawk-Dove games without solvable within-group replicator dynamics, we will make use of ODE comparision principles to compare solutions of within-group dynamics to solutions of ODEs for which we have exact solutions. For the Prisoners' Dilemma, we denote by $\Psi_t(x_0;k)$ the solution to the logistic equation given by \begin{equation} \label{eq:logisticODE} \dsddt{x(t)} = - k x(1-x) \: \: , \: \: x(0) = x_0 \end{equation} We can solve this ODE to find that \begin{equation} \label{eq:Psit} \Psi_t(x_0;k) = \frac{x_0}{x_0 + (1-x_0) e^{kt}} \end{equation} We further denote by $\Psi^{-1}_t(x;k)$ as the solution of Equation \ref{eq:logisticODE} backwards in time given $t$ and $x(t)$, which is found to be \begin{equation} \label{eq:Psiinvt} \Psi^{-1}_t(x;k) = \frac{x}{x + (1-x) e^{-kt}} \end{equation}

For the Prisoners' Dilemma, it is convenient to highlight the stable fixed point at 0 and unstable fixed point at 1 by rewriting the replicator dynamics as 
\begin{equation} \label{eq:replicatorPDbetaabs} \dsddt{x} = - x ( 1 - x) \left( |\beta| - \alpha x \right) \end{equation} Note that we choose to write $|\beta|$ because $\beta < 0$ for all Prisoners' Dilemmas, while the sign of $\alpha$ varies by case. We can examine the third term $|\beta| - \alpha x$ to understand how solutions of the within-group replicator dynamics can be compared to solutions of logistic ODEs for given $k$. 

\subsubsection{Case I PD}

For Case I, we have that $\alpha < 0$, and we can rewrite Equation \ref{eq:replicatorPDbetaabs} as \begin{equation} \label{eq:replicatorPDbetaalphaabs} \dsddt{x} = - x ( 1 - x) \left( |\beta| + |\alpha| x \right) \end{equation} We now observe that the following inequalities are satisfied for all $x \in [0,1]$, 
\[ - \left( |\beta| + |\alpha| \right) x (1-x) \leq -x(1-x) \left( |\beta| + |\alpha| x \right) \leq - |\beta| x (1-x) \leq 0 \]
This inequality can be rewritten in terms of $\phi_t(\cdot)$ and $\Psi_t(k;\cdot)$ as 
\[\dsddt{} \left[ \Psi_t \left(|\beta| + |\alpha|; \cdot \right) \right] \leq \dsddt{} \left[ \Psi_t(\cdot) \right] \leq \dsddt{} \left[ \Psi_t \left(|\beta|; \cdot \right) \right] \leq 0 \] 
Because all of the within-group dynamics for $\phi_t(\cdot)$ and $\Psi_t(k;\cdot)$ result in decreasing fractions of cooperation, we see that solution for $\Psi_t(|\beta| + |\alpha|;\cdot)$ moves faster than and $\Psi_t(\cdot)$ and $\Psi_t(|\beta| + |\alpha|;\cdot)$. From a shared initial condition $x_0$, it then follows that
\begin{equation} \label{eq:PDcharacteristicsrankingsCaseI} 0 \leq   \Psi_t \left(|\beta| + |\alpha|; x_0 \right) \leq  \Psi_t(x_0)\leq  \Psi_t \left(|\beta|; x_0 \right) \leq 1 \end{equation}
Correspondingly, for a given group located $x$ at time $t$, we see that the backward-in-time solutions of the above equations satisfy
\begin{equation} \label{eq:PDcharacteristicsbackwardrankingsCaseI} 0 \leq   \Psi_t^{-1} \left(|\beta|; x \right) \leq  \Psi_t^{-1}(x)\leq  \Psi_t^{-1} \left(|\beta|+|\alpha|; x \right) \leq 1 \end{equation}

In Figure \ref{fig:pdcomparison}(left), we illustrate the ranking of trajectories from Equation \ref{eq:PDcharacteristicsrankingsCaseI} with an example numerical solution to the forward characteristic curves $\phi_t(x_0)$ for the PD and exactly solvable faster ($\Phi_t(|\beta|+|\alpha|;x_0)$) and slower trajectories ($\Phi_t(|\beta|;x_0)$). In Figure \ref{fig:pdcomparison}(right), we illustrate the equivalent ranking of backwards trajectories from Equation \ref{eq:PDcharacteristicsbackwardrankingsCaseI}.

\begin{figure}[H]
    \centering
   \hspace{-5mm} \includegraphics[width=0.505\textwidth]{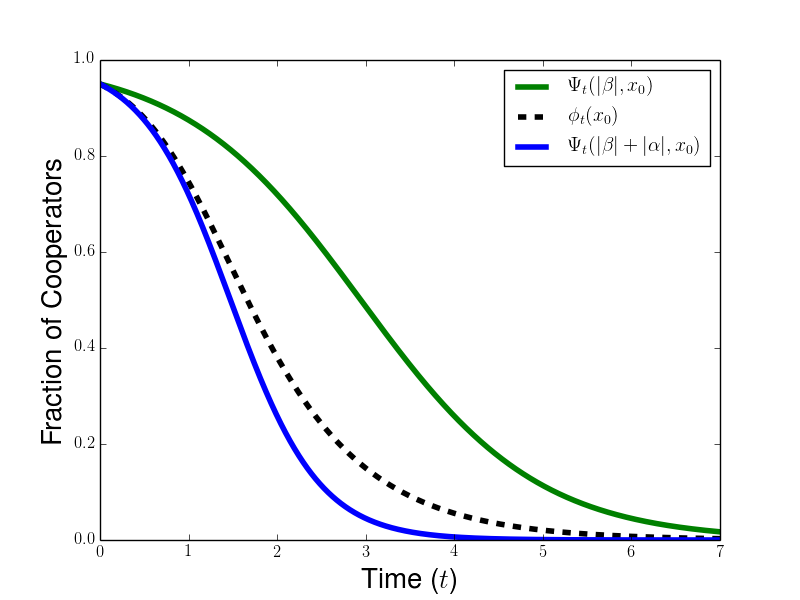} 
    \hspace{-5mm} \includegraphics[width=0.505\textwidth]{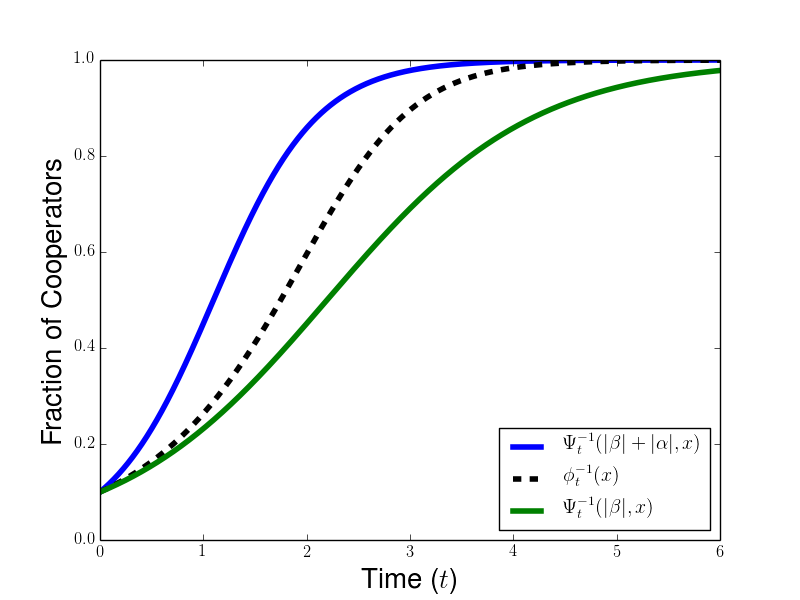}
    \caption{Illustration of comparison principle for solutions of within-group replicator dynamics with logistic ODEs. (Left) Dotted lines correspond to numerical solution for exact characteristic curve $\phi_t(x_0)$ (left) and backwards characteristic $\phi_t^{-1}(x)$ (right), blue lines correspond to faster logistic solutions $\Psi_t(|\beta| + |\alpha|;x_0)$ and $\Psi^{-1}_t(|\beta| + |\alpha|;x_0)$, and green lines correspond to slower logistic solutions $\Psi_t(|\beta|;x)$ and $\Psi^{-1}_t(|\beta|;x)$. Solutions correspond to parameters $\beta = \alpha = -1$ and initial conditions $x_0 = 0.95$ and $x = 0.1$.  }
    \label{fig:pdcomparison}
\end{figure}

We can generalize this comparison principle to within-group dynamics taking the form 
\begin{equation} \label{eq:replicatorF} \dsddt{x(t)} = x \left(1-x\right) F \left[ \pi_C(x) - \pi_D(x) \right]  \end{equation}
for an odd, increasing function $F(\cdot)$. An example of such an equation is the Fermi dynamics \cite{traulsen2006stochastic} used to model decision making based on pairwise payoff comparisons, for which $F(w) = \tanh \left( sw/2 \right)$. Within-group  dynamics often take the form of Equation \ref{eq:replicatorF} because the product $x(1-x)$ corresponds to the probability of finding a cooperator and defector needed for an imitative change of strategy, and the properties of $F(\cdot)$ mean that greater payoff differences result in faster strategy changes and that reversing the payoff differences should reverse the direction of strategy imitation. 

Using Equation \ref{eq:picminuspid} and the oddness of $F(\cdot)$, we see that solutions $\tilde{\phi}_t(x_0)$ to Equation \ref{eq:replicatorF} with initial condition $x_0$ satisfy
\[ \dsddt{} \tilde{\phi}_t(x_0) = x \left( 1 - x\right) F \left(- |\beta| - |\alpha| x \right)  = - x \left(1-x\right) F\left(|\beta| + |\alpha| x \right) \]
Because $F(\cdot)$ is increasing, we further see that 
\[- x \left(1-x\right) F\left( |\beta| + |\alpha| \right) \leq -x \left( 1 - x \right) F\left(|\beta| + |\alpha| x \right)  \leq -x \left( 1 - x \right) F\left(|\beta| \right) .    \]
Then we are able to find the ranking of characteristic curves 
\begin{equation} \label{eq:PDcaseIgeneralFranking} \Psi_t\left(F\left(|\beta| + |\alpha|\right);x_0 \right) \leq \tilde{\phi}_t(x_0) \leq  \Psi_t\left(F\left(|\beta|\right);x_0 \right) \end{equation}
and a similiar ranking of characteristics backwards in time
\begin{equation} \label{eq:PDcaseIgeneralFranking} \Psi_t^{-1}\left(F\left(|\beta|\right);x \right) \leq \tilde{\phi}_t^{-1}(x) \leq  \Psi_t^{-1}\left(F\left(|\beta| + |\alpha|\right);x \right) .  \end{equation}
Because we can prove these comparison principles for more general within-group dynamics, many of the results we prove in the following sections can be extended to a much broader class of two-level replicator equations. This approach can provide a strategy for studying how robust the resulting phenomena like the shadow of lower-level selection are to different assumptions about the rules for birth and death rates for individuals and groups.

\subsubsection{Cases III-IV PD}

For Case III and Case IV PD's, $\alpha > 0$, and we can use Equation \ref{eq:replicatorPDbetaabs} to find the following inequality for $x \in [0,1]$ 
\[ -|\beta| x (1-x) \leq - x(1-x)\left(|\beta| - \alpha x\right) \leq -x(1-x) \left( |\beta| - \alpha \right) \leq  0 \]
Using the definitions of $\phi_t(x_0)$ and $\Psi_t(k;x_0)$, we can deduce the following inequality for the characteristic curves with shared initial condition $x$
\begin{equation} \label{eq:PDcharacteristicsrankingsCaseIII} \Psi_t(|\beta|;x_0) \leq \phi_t(x_0) \leq \Psi_t(|\beta| - \alpha;x_0)  \end{equation} and find the corresponding inequality for backward-in-time characteristic curves
\begin{equation} \label{eq:PDcharacteristicsbackwardrankingsCaseIII} \Psi^{-1}_t(|\beta| - \alpha;x) \leq \phi^{-1}_t(x) \leq \Psi^{-1}_t(|\beta|;x) \end{equation}
In these cases of the PD, we can again generalize these comparison results to the class of within-group dynamics given by Equation \ref{eq:replicatorF}.

\subsubsection{Hawk-Dove Game}

For the Hawk-Dove game, $\beta > 0$ and $\alpha < 0$, so we can rewrite the within-group dynamics as 
\begin{equation} \label{eq:HDrepwithingroup} \dsddx{x(t)}{t} =  x \left(1-x\right) \left( \beta - |\alpha| x \right), \: \: x(0) = x_0, \end{equation} which has a stable interior equilibrium at $x^{eq} = \tfrac{\beta}{|\alpha|}$ and untable equilibria at $0$ and $1$. When $x \in \left[ \tfrac{\beta}{|\alpha|}, 1 \right]$, we can try to describe solutions $\phi_t(x_0)$ to Equation \ref{eq:HDrepwithingroup} by comparison to the function $\Xi_t\left(k;x_0\right)$ which is defined as the solution to 
\begin{equation} \label{eq:HDrightequation}  \dsddx{x(t)}{t} =  k \left(1-x\right) \left( \beta - |\alpha| x \right), \: \: x(0) = x_0,  \end{equation}
 Solving this ODE yields the following formula for $\Xi_t(k;x_0)$
 \begin{equation} \label{eq:Xit} \Xi_t(k;x_0) = \frac{\left(1 - x_0\right) \beta + \left(|\alpha| x_0 - \beta \right) e^{\left( \beta - |\alpha| \right) kt}}{\left(1-x_0 \right)|\alpha| + \left( |\alpha| x_0 - \beta \right) e^{\left( \beta - |\alpha| \right) kt} } \end{equation}
 and we can also solve this ODE backwards in time to see that
  \begin{equation} \label{eq:Xiinvt} \Xi_t^{-1}(k;x) = \frac{\left(1 - x\right) \beta + \left(|\alpha| x - \beta \right) e^{-\left( \beta - |\alpha| \right) kt}}{\left(1-x \right) |\alpha| + \left( |\alpha| x - \beta \right) e^{-\left( \beta - |\alpha| \right) kt} } \end{equation}
 
 Then we note that $\left(1-x \right) \left( \beta - |\alpha| x \right) \leq 0$ for $x \in [\tfrac{\beta}{|\alpha|},1]$, so, for this range of $x$-values, we have that 
 \[ \left(1 - x \right) \left(\beta - |\alpha| x \right) \leq x \left( 1 - x \right) \left( \beta - |\alpha| x \right) \leq \frac{\beta}{|\alpha|} \left( 1 - x \right) \left( \beta - |\alpha| x \right) \leq 0 \] 
 In terms of our named solutions $\phi_t(k;\cdot)$ and $\Xi_t(k;\cdot)$, we can rewrite these inequalities as 
 \[ \dsddx{}{t} \Xi_t(1;\cdot) \leq \dsddx{}{t} \phi_t(\cdot) \leq \dsddx{}{t} \Xi_t\left(\frac{\beta}{|\alpha|};\cdot\right) \leq 0 \]
 so we can deduce that 
 \begin{equation} \label{eq:HDcharranking} \Xi_t(1;x) \leq \phi_t(x) \leq \Xi_t\left(\tfrac{\beta}{|\alpha|};x\right) \end{equation} for $x \in \left[\tfrac{\beta}{|\alpha|},1\right]$. Similarly, we have a corresponding inequality for the characteristic curves backwards in time
  \begin{equation} \label{eq:HDcharbackwardranking} \Xi_t^{-1}(\tfrac{\beta}{|\alpha|};x) \leq \phi_t^{-1}(x) \leq \Xi_t^{-1}\left(1;x\right) \end{equation}
 We can also study simpler solutions for within-group dynamics below the Hawk-Dove equalibrium $\frac{\beta}{|\alpha|}$. We denote by $\Pi_t(k:x_0)$ the solution of the quadratic ODE
 \begin{equation} \label{eq:HDleftequation} \dsddt{x(t)} = k x \left( \beta - |\alpha| x \right), \: \: x(0) = x_0  \end{equation}
 Solving this ODE forward in time, we can write $\Pi_t(k;x_0)$ as 
 \begin{equation} \label{eq:Pit} \Pi_t(k;x_0) = \frac{\beta x_0}{|\alpha| x_0 + \left( \beta - |\alpha| x_0 \right) e^{- \beta k t}} \end{equation}
 and we can solve the ODE backward in time from a given $t$ and $x(t)$ to get
 \begin{equation} \label{eq:Piinvt} \Pi_t^{-1}\left(k;x\right) = \frac{\beta x}{|\alpha| x + \left( \beta - |\alpha| x \right) e^{\beta k t}} \end{equation}
 We know that $x\left( \beta - |\alpha| x \right) \geq 0$ for $x \in \left[0,\tfrac{\beta}{|\alpha|}\right]$, so for $x$ in this range we have the following inequalities
 \[ 0 \leq \left(1 - \frac{\beta}{|\alpha|} \right) x \left( \beta - |\alpha| x \right) \leq \left(1 - x \right) x \left( \beta - |\alpha| x \right)  \leq  x \left( \beta - |\alpha| x \right), \]
 which can be expressed in terms of our named solutions as 
 \[0 \leq \ddt{}\Pi_t\left(1 - \tfrac{\beta}{|\alpha|};\cdot \right) \leq \ddt{} \phi_t(\cdot) \leq \ddt{} \Pi_t\left(1;x_0 \right) \]
This gives us the following ordering of the solutions to the characteristic curves  
\begin{equation} \label{eq:HDPirankings} \Pi_t\left(1 - \tfrac{\beta}{|\alpha|};x_0 \right) \leq  \phi_t(x_0) \leq  \Pi_t\left(1;x_0 \right) \end{equation}
We can get an analogous ranking of characteristic curves solved backwards in time from $x(t)$ as  
\begin{equation} \label{eq:HDPirankingsbackwards}   \Pi^{-1}_t\left(1;x \right)  \leq  \phi^{-1}_t(x) \leq  \Pi^{-1}_t\left(1 - \tfrac{\beta}{|\alpha|};x \right) \end{equation}

Now we illustrate the rankings of characteristics for the HD game. In Figure \ref{fig:hdcomparisonupper}, we represent the dynamics above the within-group HD equilibrium, using the rankings forward in time from Equation \ref{eq:HDcharranking} with the simplified $\Xi_t(k;x)$ curves (left) and the rankings backward in time from Equation \ref{eq:HDcharbackwardranking} with the backward curves $\Xi_t^{-1}(k;x)$ (right). In Figure \ref{fig:hdcomparisonlower}, we illustrate the dynamics below the within-group HD equilibrium, using the rankings forward in time from Equation \ref{eq:HDPirankings} with the $\Pi_t(k;x)$ curves (left) and the rankings backward in time from Equation \ref{eq:HDPirankingsbackwards} with the backward curves $\Xi_t^{-1}(k;x)$ (right).

\begin{figure}[htbp]
    \centering
  \hspace{-5mm} \includegraphics[width=0.505\textwidth]{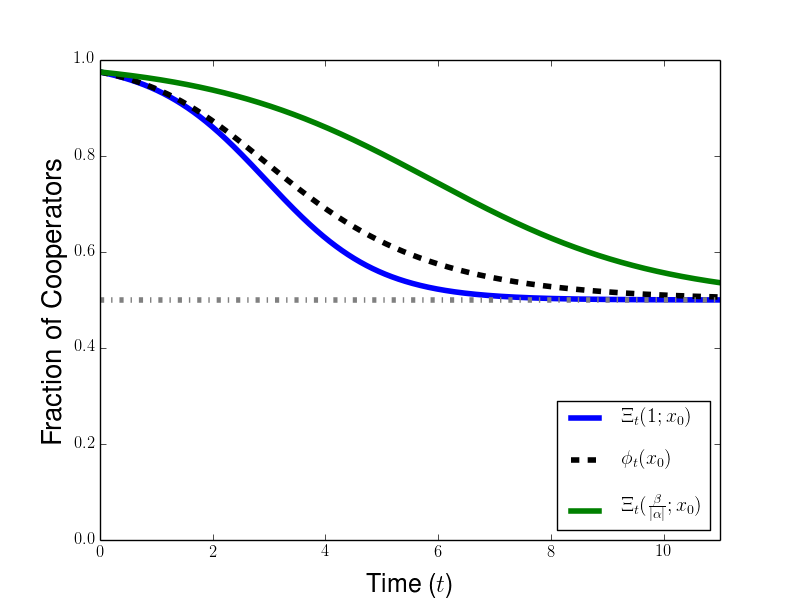} 
   \hspace{-5mm} \includegraphics[width=0.505\textwidth]{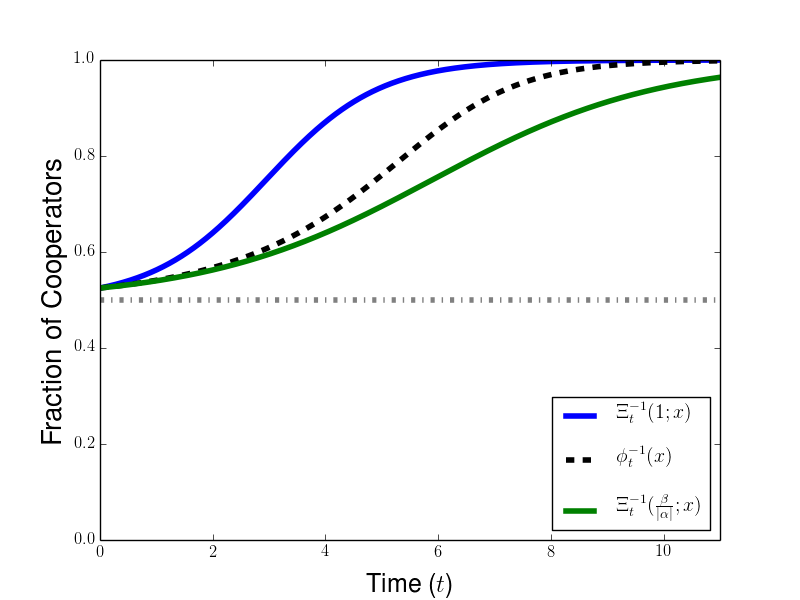} 
    \caption{Illustration of comparison principle for solutions of within-group replicator equation for HD game for initial conditions above the within-group equilibrium at $\frac{\beta}{|\alpha|}$. (Left) Solutions forward in time: dashed black line describes numerical solution of exact characteristic curves $\phi_t(x_0)$, while blue and green lines correspond to faster curve $\Xi_t\left(\tfrac{\beta}{|\alpha|};x_0\right)$ and slower curve $\Xi_t\left(1;x_0 \right)$, respectively. (Right) Solutions backward in time: dashed black line corresponds to numerical solution of characteristic curves $\phi_t^{-1}(x)$, while blue and green lines correspond to faster and slower solutions exactly solvable logistic curves $\Xi_t^{-1}\left(1;x\right)$ and $\Xi_t^{-1} \left(\tfrac{\beta}{|\alpha|};x \right)$, respectively.}
    \label{fig:hdcomparisonupper}
\end{figure}

\begin{figure}[H]
    \centering
  \hspace{-5mm} \includegraphics[width=0.505\textwidth]{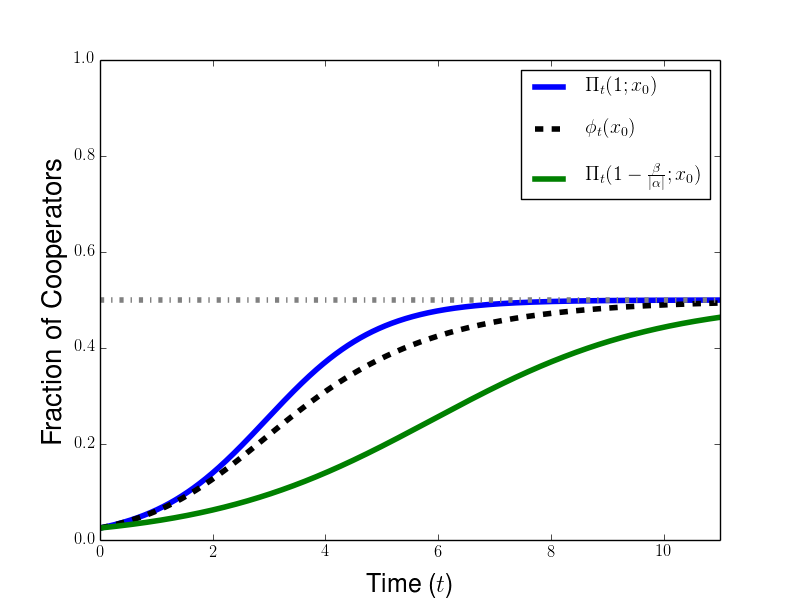} 
   \hspace{-5mm} \includegraphics[width=0.505\textwidth]{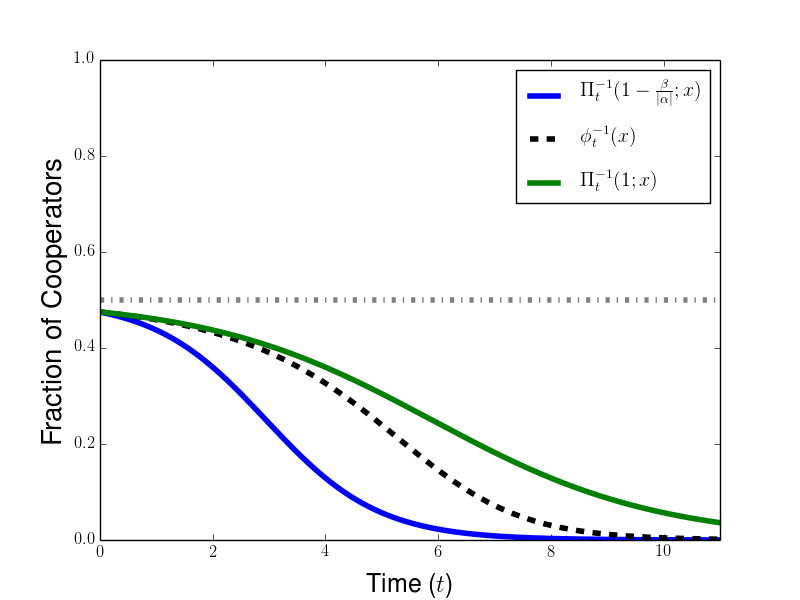} 
    \caption{Illustration of comparison principle for solutions of within-group replicator equation for HD game for initial conditions below the within-group equilibrium at $\frac{\beta}{|\alpha|}$. (Left) Solutions forward in time: dashed black line describes numerical solution of exact characteristic curves $\phi_t(x_0)$, while blue and green lines correspond to faster curve $\Pi_t\left(1;x_0\right)$ and slower curve $\Xi_t\left(1- \tfrac{\beta}{|\alpha|};x_0 \right)$, respectively. (Right) Solutions backward in time: dashed black line corresponds to numerical solution of characteristic curves $\phi_t^{-1}(x)$, while blue and green lines correspond to faster and slower solutions exactly solvable logistic curves $\Xi_t^{-1}\left(1;x\right)$ and $\Xi_t^{-1} \left(1-\tfrac{\beta}{|\alpha|};x \right)$, respectively.}
    \label{fig:hdcomparisonlower}
\end{figure}

\subsection{Solutions Along Characteristics} \label{sec:solutionalong}

An advantage of using our comparison principles is that we can make use of the solvability of solutions along the simiplified characteristics $\Psi_t(k,x)$ for the PD and along $\Xi_t(k;x)$ and $\Pi_t(k;x)$ for the HD game. From the formula describing between-group competiton for solutions along the exact characteristics $\phi_t(x)$ given by Equation \ref{eq:wtx}, we know that we would like to use our comparison principles to estimate $G(\phi_t(x_0))$, the average group payoff along characteristics. For the various cases of the PD, we use the signs of $\gamma$ and $\alpha$ and the rankings of characteristic curves given by Equations \ref{eq:PDcharacteristicsrankingsCaseI} and \ref{eq:PDcharacteristicsrankingsCaseIII} to find the following bounds for $G(\phi_t(x_0))$ along our known simpler solution curves
\begin{equation*} \label{eq:groupcomparison} \begin{aligned} \gamma \Psi_t\left(|\beta| + |\alpha|;x_0 \right) - |\alpha| \Psi_t(|\beta|;x_0)^2 \leq G(\phi_t(x_0)) &%
\leq  \gamma \Psi_t\left(|\beta|;x_0 \right) - |\alpha| \Psi_t(|\beta| + |\alpha|;x_0)^2 \textnormal{: (Case I PD)} 
\\ \gamma \Psi_t\left(|\beta|;x_0 \right) + \alpha \Psi_t(|\beta|;x_0)^2 \leq G(\phi_t(x_0)) &\leq  \gamma \Psi_t\left(|\beta| - \alpha;x_0 \right) + \alpha \Psi_t(|\beta| - \alpha;x_0)^2 \textnormal{: (Case III PD)}
\\ -|\gamma| \Psi_t\left(|\beta|-\alpha;x_0 \right) + \alpha \Psi_t(|\beta|;x_0)^2 \leq G(\phi_t(x_0)) &\leq  -|\gamma| \Psi_t\left(|\beta|;x_0 \right) + \alpha \Psi_t(|\beta|-\alpha ;x_0)^2 \textnormal{: (Case IV PD)}  \end{aligned} \end{equation*}
In Figure \ref{fig:pdgroupcomparison}, we illustrate example trajectories for $G(\phi_t(x_0))$ and corresponding upper and lower bounds for the Case I (left) and Case III (right) PDs. In the example for the Case I PD, we can see that $G(\phi_t(x_0))$ and its bounds of are non-monotonic in time, corresponding to the characteristic curves traversing the intermediate optimum composition for average group payoff.

\begin{figure}[H]
    \centering
   \hspace{-5mm} \includegraphics[width=0.505\textwidth]{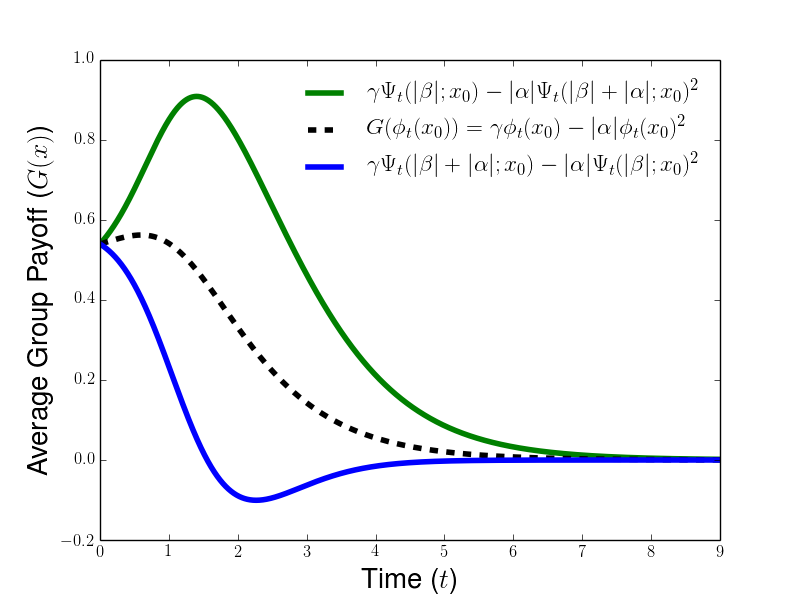} 
    \hspace{-5mm} \includegraphics[width=0.505\textwidth]{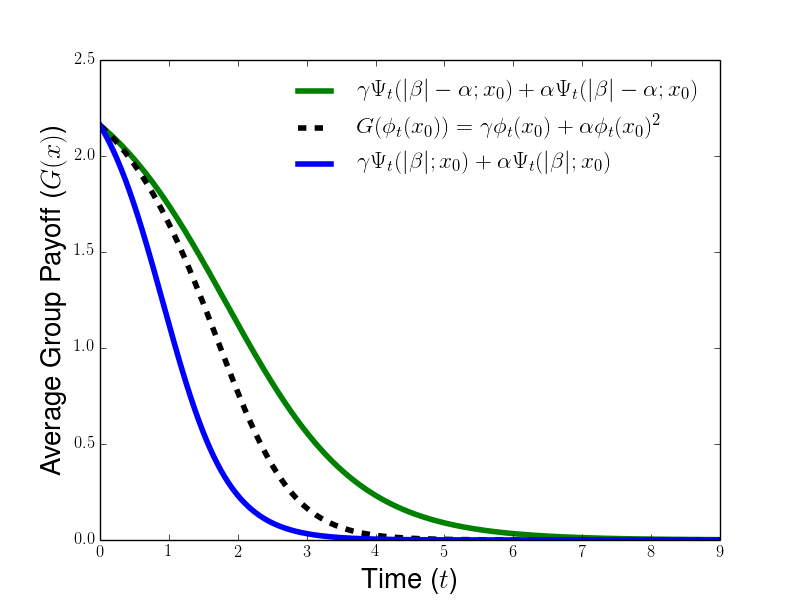}
    \caption{Comparison of group payoff along characteristic curves $G(\phi_t(x_0))$ with expressions along exactly solvable curves for the PD game. Parameters correspond to Case I (left) and Case III (right).}
    \label{fig:pdgroupcomparison}
\end{figure}

We can similarly bound $G(\phi_t(x_0))$ for the HD game using the ranking of characteristic curves from Equations \ref{eq:HDcharranking} and \ref{eq:HDPirankings} and the fact that $\gamma > 0$ and $\alpha < 0$ for all HD games. We find that 
\begin{align*} \gamma \Pi_t\left( 1 - \tfrac{\beta}{|\alpha|};x_0\right) -  |\alpha| \Pi_t(1;x_0)^2 &\leq G(\phi_t(x_0)) \leq  \gamma \Pi_t\left(1;x_0 \right) - |\alpha| \Pi_t\left( 1 - \tfrac{\beta}{|\alpha|};x_0\right)^2  : \: \: x_0 < \tfrac{\beta}{|\alpha|} 
\\  \gamma \Xi_t\left( 1;x_0\right) -  |\alpha| \Xi_t\left(1- \tfrac{\beta}{|\alpha|};x_0\right)^2 &\leq G(\phi_t(x_0)) \leq  \gamma \Pi_t\left(1- \tfrac{\beta}{|\alpha|};x_0 \right) - |\alpha| \Pi_t\left( 1 ;x_0\right)^2  : \: \: x_0 > \tfrac{\beta}{|\alpha|} \end{align*} 
In Figure \ref{fig:hdgroupcomparison}, we illustrate example trajectories for $G(\phi_t(x_0))$ and corresponding upper and lower bounds for the group payoff both above the within-group equilibrium (left) and and below the within-group equilibrium (right).

\begin{figure}[H]
    \centering
   \hspace{-5mm} \includegraphics[width=0.505\textwidth]{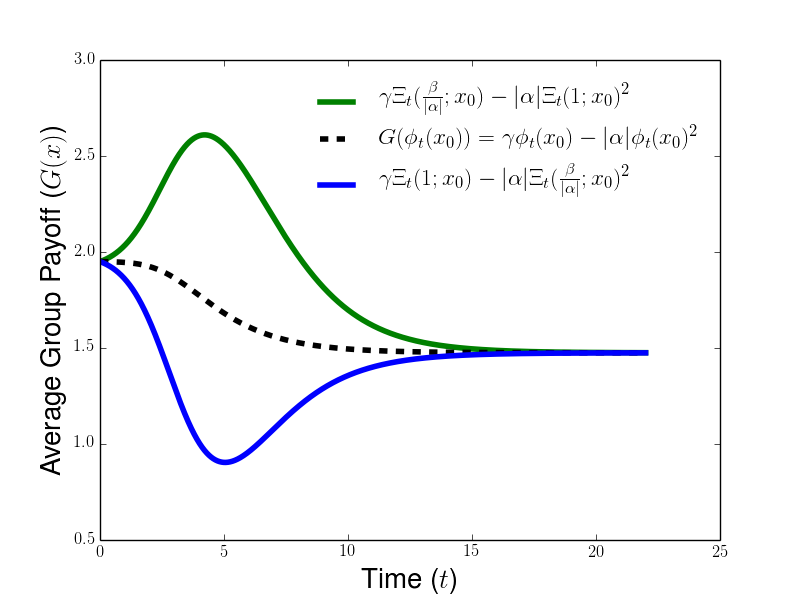} 
    \hspace{-5mm} \includegraphics[width=0.505\textwidth]{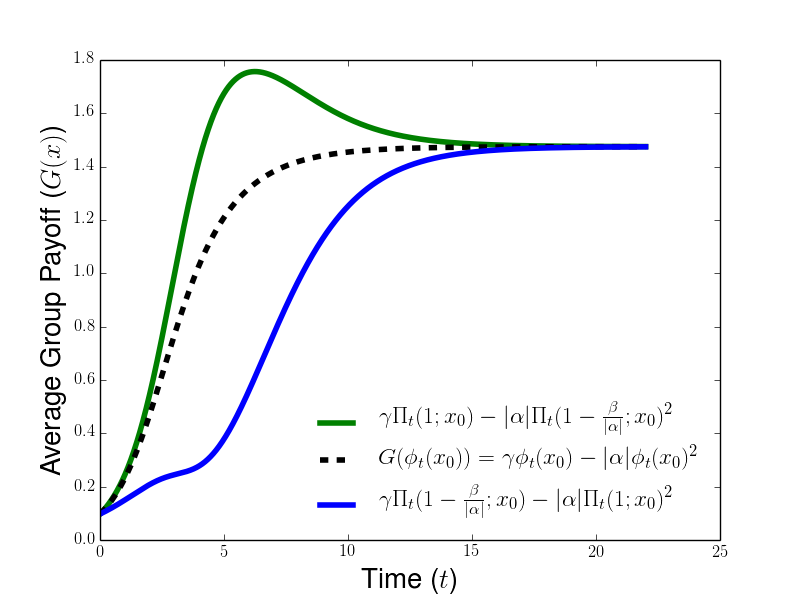}
    \caption{Comparison of group payoff along characteristic curves $G(\phi_t(x_0))$ with expressions for upper and lower bounds of group payoff in terms of exactly solvable curves. for the HD game. Dynamics described above within-group equlibrium (left) and below the within-group equilibrium (right). Long-time behavior for $G(\phi_t(x_0))$ and the bounds agree upon group payoff $G(x^{eq})$ at within-group equilibrium for HD game.}
    \label{fig:hdgroupcomparison}
\end{figure}

To study solutions along characteristics from Equation \ref{eq:wtx}, we also know the must integrate $G(\phi_t(x_0))$ in time. From Equation \ref{eq:wtxmoment}, we know that this requires computing \[\int_0^t G(\phi_s(x_0)) ds = \gamma \int_0^t \phi_s(x_0) ds + \alpha \int_0^t \phi_s(x_0)^2 ds\]
Using the bounds for $G(\phi_t(x_0))$ found above, this means that we need to be able to compute both the integrals in time of our simplified characteristic curves and of the squares of our simplified characterstic curves. For convenience, we include a derivation of these formulas in Section \ref{sec:integrals} of the appendix.

For the Prisoners' Dilemma, we see that the integrals of the simplified curves $\Psi_t(k;x)$ and their squares $\Psi_t(k;x)^2$ are given by
\begin{align} \label{eq:psiintegral} \int_0^t \Psi_s(k,x) ds &= t - \frac{1}{k} \log\left( x + (1-x) e^{kt} \right), \\
 \label{eq:psisquaredintegral} \int_0^t \Psi_s(k;x)^2 ds &= t + \frac{1}{k} \log \left( x + (1-x) e^{-kt} \right) + \frac{x}{k} - \frac{1}{k} \frac{x}{x + (1-x) e^{-kt}}. \end{align}

 For within-group dynamics above the Hawk-Dove equilibrium, we find the integrals along the simpler characteristic curves $\Xi_t(k;x_0)$ are given by 
 \begin{align} \label{eq:Xiintegral} \int_0^t \Xi_s(k,x) ds &= t - \frac{1}{|\alpha| k} \log\left(\frac{ |\alpha| x - \beta + |\alpha| \left( 1 - x \right) e^{\left(-|\alpha| - \beta\right) kt}}{|\alpha| - \beta}  \right), \\
 \label{eq:Xisquaredintegral} \int_0^t \Xi_s(k;x)^2 ds &=  t - \left(\frac{|\alpha| + \beta}{|\alpha|^2 k}\right) \log \left( \frac{ |\alpha| x - \beta + |\alpha| \left( 1 - x \right) e^{\left(|\alpha| - \beta\right) kt}}{|\alpha| - \beta}  \right)  \\ &-  \frac{1}{|\alpha| k} \left[ \frac{\left(1-x\right)\left( |\alpha| x - \beta \right) \left( 1 - e^{- \left( |\alpha| - \beta \right) k t} \right)}{|\alpha| x - \beta + |\alpha| \left( 1 - x \right) e^{- \left( |\alpha| - \beta \right) k t}} \right] \nonumber
  \end{align}
 For within-group dynamics in the Hawk-Dove game below the within-group equilibrium, we want to integrate the simpler solution curves $\Pi_t(k;x)$ in time. With the same method, we find that 
 \begin{subequations} \begin{align} \label{eq:Piintegral} \int_0^t \Pi_s(k;x_0) ds &= - \frac{1}{k} \log\left(\frac{1}{\beta} \left[\left( \beta - |\alpha| x \right) + |\alpha| x e^{- \beta k t} \right] \right)  \\ \label{eq:Pisquaredintegral} \int_0^t \Pi_s(k;x_0)^2 ds &=  - \frac{\beta}{|\alpha|^2 k} \log\left(\frac{1}{\beta} \left\{\left( \beta - |\alpha| x \right) + |\alpha| x e^{- \beta k t} \right\} \right) \\  &-  \frac{1}{|\alpha| k} \left[ \frac{x \left( \beta - |\alpha| x \right) \left(e^{\beta k t} - 1 \right)}{|\alpha| x + \left( \beta - |\alpha| x \right) e^{\beta k t}} \right] \nonumber %
 \end{align} \end{subequations}
We note that both of these integrals are bounded for all $x$ and $t$, unlike the integrals for dynamics above the within-group equilibria that scale linearly with $t$. For the other integrals, the linear $t$ term will help us to determine the levels of $\lambda$ for which all groups will convergence to a delta-concentration at the equilibrium level of cooperation for the within-group dynamics.

\subsection{\holder Exponent Near Full-Cooperator Group} \label{sec:Holderpreserve}

To describe the long-time behavior of our multilevel dynamics, it is useful to characterize the tail of the measure $\mu_t(dx)$ near the endpoint $x=1$ corresponding to full-cooperator groups. We quantify this tail behavior using the \holder exponent near 1, which is defined as follows.
\begin{definition} \label{def:Holderx1} The \holder exponent $\theta$ of measure $\mu(dx)$ near $x=1$ satisfies the relation \begin{equation} \label{eq:holderformula} \theta = \ds\inf_{\Theta \geq 0} \left\{ \lim_{y \to 0} \frac{\mu\left([1-y,1] \right)}{y^{\Theta}} > 0  \right\} \end{equation} \end{definition}

The \holder exponent near $x=1$ was employed by Luo and Mattingly \cite{luo2017scaling} and by Cooney \cite{cooney2019replicator} to show when multilevel dynamics converge to a delta-function or a steady-state density for cases in which the within-group dynamics were exactly solvable. They arise when tracing solutions along characteristics backwards in time, as the within-group dynamics drive groups away from compositions of cooperators near $x=1$ as time proceeds. We can better understand the intuitive meaning of the \holder exponent near $x=1$ from the following example. 
\begin{example} \label{ex:holdertheta}
The measures with power law densities of the form $\mu(dx) = \theta \left( 1 - x \right)^{\theta -1} dx$ have \holder exponent $\theta$ near $x=1$. We can see this by noting that $\mu\left([1-y,1] \right) = \theta \int_{1-y}^1 \left( 1 - z \right)^{\theta -1} dz = y^{\theta}$, and then finding the limit 
\[ \ds\lim_{y \to 0} \frac{\mu\left([1-y,1] \right)}{y^{\Theta}} = \ds\lim_{y \to 0} y^{\theta - \Theta} =  \left\{ \begin{array}{cr}  0 & :  \Theta < \theta \\
      1 &: \Theta = \theta \\ \infty &: \Theta > \theta
     \end{array}
   \right. . \]
Then Definition \ref{def:Holderx1} tells us that $\theta$ is the \holder exponent near $x=1$.
\end{example}
Intuitively, we can think of a measure with \holder exponent $\theta$ near $x=1$ as acting like $\mu(dx) = \theta \left( 1 - x \right)^{\theta - 1}dx$ near $x=1$, i.e. having a similar level of concentration or lack of probability near the full-cooperator groups. In particular, higher $\theta$ corresponds to sparser concentration of groups with composition close to full cooperation.

In the previously studied cases, the \holder exponent of the initial distribution turned out to also be the \holder exponent of the eventual steady state of the dynamics \cite{luo2017scaling,cooney2019replicator}. For a special family of the Prisoners' Dilemma for which Equation \ref{eq:replicatorpde} has an explicit solution, it was shown directly using the solution formula that the \holder exponents of $\mu_t(dx)$ were preserved in time \cite{cooney2019replicator}. Making use of comparison principles, we now shown in the following propositions that the \holder exponent is preserved in time for all PD and HD games, highlighting the importance of the \holder exponent in describing the long-time behavior of our multilevel dynamics. Because we use comparisons with both $\Psi_t(k;x)$ and $\Xi_t(k;x)$, we separately address the PD game in Proposition \ref{prop:PDholder} and the HD game Proposition \ref{prop:HDholder}.

\begin{proposition} \label{prop:PDholder} For the multilevel PD dynamics, the \holder exponent near $x=1$ is preserved in time. Given an initial measure $\mu_0(dx)$ with \holder exponent near $x=1$ of $\theta_0 = \theta$, the \holder exponent $\theta_t$ of the measure $\mu_t(dx)$ solving Equation \ref{eq:replicatormeasurepde} satisfies $\theta_t = \theta$.   \end{proposition}

\begin{proof}
Using the push-forward reseprentation of $\mu_t(dx) = w_t(x) (\mu_0 \circ \phi_t^{-1})(dx)$ and noting that $x=1$ is a fixed point of the replicator dynamics (so $\phi_t(1) = 1 $ and $\phi_t^{-1}(1) = 1$), we see that \[ \frac{\mu_t[1-x,1]}{x^{\Theta}} = \frac{\int_{1-x}^1 \mu_t(dy)}{x^{\Theta}} = \frac{\ds\int_{\phi_t^{-1}(1-x)}^1 w_t(\phi_t(y)) \mu_0(dy)}{x^{\Theta}} \]

We note, for $x \in [0,1]$, that $\phi_t^{-1}(1-x) \geq 1-x \geq \phi_t(1-x)$.
We recall from Equation \ref{eq:wtx} that $w_t(\phi_t(x)) = \exp\left(\int_0^t \left[G(\phi_s(x)) - \langle G(\cdot) \rangle_{\mu^s} \right] ds \right)$, which allows us to deduce the estimate \[ \exp\left( - G^* t \right) \leq w_t(\phi_t(x)) \leq \exp\left(G^*t \right) \: \: \mathrm{where} \: \: G^* = \max_{x \in [0,1]}{G(x)} - \min_{x \in [0,1]}{G(x)} \]
Using the upper bound for $w_t(\phi_t(x))$ and the fact that $1 -x \leq \phi_t^{-1}(1-x)$, we now see that 
\[ \frac{\mu_t[1-x,1]}{x^{\Theta}} \leq e^{G^* t} \frac{\mu_0[\phi_t^{-1}(1-x),1]]}{x^{\Theta}} \leq e^{G^* t} \frac{\mu_0[1-x,1]}{x^{\Theta}}. \] 
Because this is true for $x \in [0,1]$, we further see that 
\[ \ds\limsup_{x \to 0} \frac{\mu_t[1-x,1]}{x^{\Theta}} \leq e^{G^* t} \ds\limsup_{x \to 0} \frac{\mu_0[1-x,1]}{x^{\Theta}} 
\]
If we consider $\Theta < \theta_0$, we know from our assumption that $\mu_0(dx)$ has \holder exponent $\Theta$ near $1$ that \[ \ds\limsup_{x \to 0} \frac{\mu_0[1-x,1]}{x^{\Theta}}  = \ds\lim_{x \to 0} \frac{\mu_0[1-x,1]}{x^{\Theta}} = 0. \]
This allows us to deduce that 
\begin{equation*} \ds\limsup_{x \to 0} \frac{\mu_t[1-x,1]}{x^{\Theta}} \leq 0 \textnormal{ and therefore }  \ds\lim_{x \to 0} \frac{\mu_t[1-x,1]}{x^{\Theta}} =0 \textnormal{ when $\Theta < \theta_0$}.   \end{equation*}
As a result, the smallest $\Theta$ for which $\lim_{x \to 0} \frac{\mu_t[1-x,1]}{x^{\Theta}} > 0$ can be no smaller than $\theta_0$, and we conclude that $\theta_t \geq \theta_0$, the Hölder exponent near $x=1$ is non-decreasing in time. 

Using the corresponding lower bound for $w_t(\phi_t(x))$ and the fact that $x \leq 1 - \phi_t(1-x)$, we have that 
\[ \frac{\mu_t[1-x,1]}{x^{\Theta}} \geq e^{-G^* t} \frac{\mu_0[\phi_t^{-1}(1-x),1]}{x^{\Theta}} \geq  e^{-G^* t} \frac{\mu_0[\phi_t^{-1}(1-x),1]}{\left(1 - \phi_t(1-x) \right)^{\Theta}} \]
Because this inequality holds for all $x \in [0,1]$, we know further that 
\begin{equation} \label{eq:pdphiliminf} \ds\liminf_{x \to 0} \frac{\mu_t[1-x,1]}{x^{\Theta}} \geq e^{-G^* t} \ds\liminf_{x \to 0}  \frac{\mu_0[\phi_t^{-1}(1-x),1]}{\left(1 - \phi_t(1-x) \right)^{\Theta}}  \end{equation}
Now we explore the righthand side of this inequality, noticing that
\begin{align}%
\ds\lim_{x \to 0}  \frac{\mu_0[\phi_t^{-1}(1-x),1]}{\left(1 - \phi_t(1-x) \right)^{\Theta}}  &= \lim_{x \to 0} \left(  \left[\frac{\mu_0[\phi_t^{-1}(1-x),1]}{\left(1 - \phi_t^{-1}(1-x) \right)^{\Theta}} \right] \left[  \frac{\left(1 - \phi_t^{-1}(1-x) \right)^{\Theta}}{\left(1 - \phi_t(1-x) \right)^{\Theta}} \right] \right) \nonumber \\ &= \left[ \ds\lim_{x \to 0} \left(\frac{\mu_0[\phi_t^{-1}(1-x),1]}{\left(1 - \phi_t^{-1}(1-x) \right)^{\Theta}} \right)\right] \left[ \ds\lim_{x \to 0} \left(  \frac{\left(1 - \phi_t^{-1}(1-x) \right)}{\left(1 - \phi_t(1-x) \right)} \right) \right]^{\Theta} \label{eq:pdrightlimit} \end{align}
First we note by the substitution $y = 1 - \phi_t^{-1}(1-x)$ and the continuity of $\phi_t^{-1}(1-x)$ that %
\[ \ds\lim_{x \to 0} \frac{\mu_0[\phi_t^{-1}(1-x),1]}{\left(1 - \phi_t^{-1}(1-x) \right)^{\Theta}} = \ds\lim_{ y \to 0} \frac{\mu_0[1-y,1]}{y^{\Theta}}  \] 
We next see that we can describe the rightmost limit of Equation \ref{eq:pdrightlimit} through comparison principles on $\phi_t(\cdot)$ and $\phi_t^{-1}(\cdot)$ for the relevant game. For the Prisoners' Dilemma, we recall that there exists $k_f$ such that there are faster characteristic curves satisfying $\phi_t^{-1}(x) \leq \Psi_t^{-1}(k_f;x)$ and $\phi_t(x) \geq \Psi_t(k_f;x)$, and therefore we know that \[  \frac{1 - \phi_t^{-1}(1-x) }{1 - \phi_t(1-x) }  \geq   \frac{1 - \Psi_t^{-1}(k_f;1-x) }{1 - \Psi_t(k_f;1-x)}  \]
Using Equations \ref{eq:Psit} and \ref{eq:Psiinvt}, we are able to find the expressions $\Psi_t(k_f,1-x) = \ds\frac{1-x}{1-x + xe^{k_f t}}$ and $\Psi_t^{-1}(k_f,1-x) = \ds\frac{1-x}{1-x + xe^{-k_f t}}$. This allows us to see that 
\[  \frac{1 - \phi_t^{-1}(1-x) }{1 - \phi_t(1-x) } \geq \left(\frac{x e^{k_f t}}{1-x + xe^{k_f t}} \right) \bigg/ \left(\frac{xe^{-k_f t}}{1-x + xe^{-k_f t}} \right) = \frac{e^{2 k_ft} \left(1 - x + x e^{-k_f t}\right)}{ \left(1 - x + xe^{k_f t} \right) }.  \]
Then we compute the limit
\[ \ds\lim_{x \to 0} \frac{1 - \phi_t^{-1}(1-x) }{1 - \phi_t(1-x) }  \geq \ds\lim_{x \to 0} \frac{e^{2 k_ft} \left(1 - x + x e^{-k_f t}\right)}{ \left(1 - x + xe^{k_f t} \right)} = e^{2 k_f t}   \]
This allows us to write Equation \ref{eq:pdrightlimit} as %
\[\ds\lim_{x \to 0}  \frac{\mu_0[\phi_t^{-1}(1-x),1]}{\left(1 - \phi_t(1-x) \right)^{\Theta}} \geq e^{2 k_f t} \left[\ds\lim_{y \to 0} \frac{\mu_0[1-y,1]}{y^{\Theta}} \right]. \]
Considering $\Theta \geq \theta_0$, we know from our assumption on the \holder exponent of $\mu_0(dx)$ that 
\[\ds\liminf_{x \to 0}  \frac{\mu_0[\phi_t^{-1}(1-x),1]}{\left(1 - \phi_t(1-x) \right)^{\Theta}} = \ds\lim_{x \to 0}  \frac{\mu_0[\phi_t^{-1}(1-x),1]}{\left(1 - \phi_t(1-x) \right)^{\Theta}} \geq e^{2 k_f t} \left[\ds\lim_{y \to 0} \frac{\mu_0[1-y,1]}{y^{\Theta}} \right] = \infty. \]
Then we are able to see from Equation \ref{eq:pdphiliminf} that 
\[  \ds\liminf_{x \to 0} \frac{\mu_t[1-x,1]}{x^{\Theta}} \geq \infty \textnormal{ and therefore } \ds\lim_{x \to 0} \frac{\mu_t[1-x,1]}{x^{\Theta}} = \infty \textnormal{ when $\Theta \geq \theta_0$}. \]

As a result, the smallest $\Theta$ for which $\lim_{x \to 0} \frac{\mu_t[1-x,1]}{x^{\Theta}} > 0$ cannot be larger than $\theta_0$, so we deduce that $\theta_t \leq \theta_0$, the \holder exponent near $x=1$ is non-increasing in time for the Prisoners' Dilemma. Combining this with our previous observation that $\theta_t \geq \theta_0$ lets us conclude that $\theta_t = \theta_0$ for all $t \geq 0$, so the Hölder exponent near $x=1$ is preserved in time under the evolution of our multilevel system for PD games.
\end{proof}

\begin{proposition} \label{prop:HDholder} For the multilevel Hawk-Dove game, the \holder exponent near $x=1$ is preserved in time. Given an initial measure $\mu_0(dx)$ with \holder exponent near $x=1$ of $\theta_0 = \theta$, the \holder exponent $\theta_t$ of the measure $\mu_t(dx)$ solving Equation \ref{eq:replicatormeasurepde} satisfies $\theta_t = \theta$.   \end{proposition}

\begin{proof}
Because the within-group dynamics for the HD have an intermediate equilibrium $x^{eq}$, we restrict our attention to sufficiently small values of $x$ such that $1-x$ is located above the HD equilibrium so that we are still able to apply the inequality $\phi_t^{-1}(1-x) \geq 1-x \geq \phi_t(1-x)$. With this restriction, the proof that $\theta_t \leq \theta_0$ carries over to the HD game because it relies only on properties of a general solution $\phi_t(x_0)$ and $\phi_t^{-1}(x)$ to the replicator dynamics with an unstable fixed point at $1$. Similarly, we can use the argument from Proposition \ref{prop:PDholder} to show that Equations \ref{eq:pdphiliminf} and \ref{eq:pdrightlimit} hold for the HD as well. 
To further analyze Equation \ref{eq:pdrightlimit} for the Hawk-Dove game, we can use the inequalities provided by Equations \ref{eq:HDcharranking} and \ref{eq:HDcharbackwardranking} to see that 
\begin{equation} \label{eq:HDcharcompare} \frac{1 - \phi_t^{-1}(1-x) }{1 - \phi_t(1-x) }  \geq   \frac{1 - \Xi_t^{-1}(1;1-x) }{1 - \Xi_t(1;1-x)}   \end{equation}
Using our expressions from Equations \ref{eq:Xit} and \ref{eq:Xiinvt}, we know that \begin{subequations}\begin{align} \Xi_t\left(k;1-x\right) &=  \frac{\beta x  + \left[|\alpha| \left(1 -x\right) - \beta \right] e^{\left( \beta - |\alpha| \right) kt}}{|\alpha| x + \left[|\alpha| \left(1-x\right)  - \beta \right] e^{\left( \beta - |\alpha| \right) kt}}  \\  \Xi_t^{-1}\left(k;1-x\right) &=  \frac{\beta x  + \left[|\alpha| \left(1 -x\right) - \beta \right] e^{\left(  |\alpha| - \beta \right) kt}}{|\alpha| x + \left[|\alpha| \left(1-x\right)  - \beta \right] e^{-\left( |\alpha| - \beta \right) kt}}   \end{align} \end{subequations}
Plugging these expressions into the inequality in Equation \ref{eq:HDcharcompare}, we see that
\begin{align*} \frac{1 - \phi_t^{-1}(1-x) }{1 - \phi_t(1-x) }  &\geq \left( \frac{\left(|\alpha| - \beta \right)x}{|\alpha| x + \left[|\alpha| \left(1-x\right)  - \beta \right] e^{\left( \beta - |\alpha| \right) kt}} \right) \bigg/ \left( \frac{\left( |\alpha| - \beta\right) x }{|\alpha| x + \left[|\alpha| \left(1-x\right)  - \beta \right] e^{\left( |\alpha| - \beta \right) kt}} \right) \\ &= \frac{|\alpha| x + \left[|\alpha| \left(1-x\right)  - \beta \right] e^{\left( |\alpha| - \beta \right) kt}}{|\alpha| x + \left[|\alpha| \left(1-x\right)  - \beta \right] e^{\left( \beta - |\alpha| \right) kt}} \end{align*}
Then, taking the limit as $x \to 0$ on both sides, we see that 
\[ \ds\lim_{x \to 0} \frac{1 - \phi_t^{-1}(1-x) }{1 - \phi_t(1-x) }  \geq e^{2 \left( |\alpha| - \beta \right) kt} \geq 0 \]
Looking back to Equation \ref{eq:pdrightlimit}, we can now use above inequality and the continuity of $\phi_t^{-1}(1-x)$ to find that 
\[\ds\lim_{x \to 0}  \frac{\mu_0[\phi_t^{-1}(1-x),1]}{\left(1 - \phi_t(1-x) \right)^{\Theta}}\geq e^{\left(2 \left(|\alpha| - \beta\right) - G^*\right) t} \left[\ds\lim_{y \to 0} \frac{\mu_0[1-y,1]}{y^{\Theta}} \right]. \]
From these calculations, we can show as in the proof of Proposition \ref{prop:PDholder} that if $\lim_{y \to 0} \frac{\mu_0[1-y,1]}{y^{\Theta}} > 0$, then so is $\lim_{y \to 0} \frac{\mu_t[1-y,1]}{x^{\Theta}}$. Therefore $\theta_t \leq \theta_0$ for $t > 0$, and then we can further deduce that $\phi_t = \phi_0$, the \holder exponent near $x=1$ is conserved under the evolution of the multilevel dynamics for HD games.
\end{proof}

 \section{General PD Games} \label{sec:GeneralPD}

In this section, we consider PD games with general payoff matrices. In Section \ref{sec:PDlongtime}, we characterize the conditions under which the multilevel PD dynamics converge to a concentration at all-defector groups. In Section \ref{sec:PDsteady}, we find density steady states of the PD dynamics and characterize the most abundant composition of cooperators and average payoff of the population at steady state.

\subsection{PD Long-Time Behavior} \label{sec:PDlongtime}

Here, we use the comparison principles from Section \ref{sec:Comparison} to characterize the conditions under which the within-group dynamics dominate and the population converges to a delta-function at the all-defector equilibrium. Due to a subtlety in the group payoff functions, we first handle Case I-III PDs in Proposition \ref{prop:PD13delta} and then study the Case IV PD in Proposition \ref{prop:PD4delta}. In Proposition \ref{prop:PD13delta}, we show that there is a critical relative level of selection $\lambda^*$ such that the multilevel dynamics for the Case I-III PDs converge to $\delta(x)$ as $t \to \infty$ when $\lambda < \lambda^*$. 

\begin{proposition} \label{prop:PD13delta} Consider Cases I-III of the PD (in which $\gamma > 0$) and supose an initial condition $\mu_0(dx)$ with \holder coefficient $\theta$ near $x=1$. If $\lambda (\gamma + \alpha) < \left(|\beta| - \alpha\right) \theta$, then $\mu_t(dx) \rightharpoonup \delta(x)$.  

\end{proposition}

Here we use the symbol ``$\rightharpoonup$'' to denote weak convergence of probability measures, so we say that $\mu_n(dx) \rightharpoonup \mu(dx)$ if $\int_0^1 \psi(x) \mu_n(dx) \to \int_0^1 \psi(x) \mu(dx)$ for each admissible test function $\psi(x)$. 

\begin{proof}
We wish to show,for any continuous test function $\psi(x)$, that $\int_0^1 \psi(x) \mu_t(dx) \to \int_0^1 \psi(x) \delta(x) dx = \psi(0)$. Because $\mu_t(dx)$ is a probability distribution, we have that 
\[ \bigg| \ds\int_0^1 \psi(x) \mu_t(dx) - \psi(0) \bigg| = \bigg| \ds\int_0^1   \psi(x) - \psi(0)\mu_t(dx) \bigg| \leq \int_0^1 \big| \psi(x) - \psi(0) \big| \mu_t(dx)  \] Because $\psi(\cdot)$ is continuous, we know that $\forall \epsilon > 0$, $\exists \delta$ such that $|\psi(x) - \psi(0)| < \epsilon$ when $x \in [0,\delta]$. Using this and our pushforward representation $\mu_t(dx) = w_t(x)\left[ \mu_0 \circ \phi_t^{-1} \right](dx)$, we obtain \begin{align*}  \bigg| \ds\int_0^1 \psi(x) \mu_t(dx) - \psi(0) \bigg| &\leq \int_0^{\delta} \big| \psi(x) - \psi(0) \big| \mu_t(dx) + \int_{\delta}^{1} \big| \psi(x) - \psi(0) \big| \mu_t(dx)  \\ & \leq \epsilon + 2 ||\psi||_{\infty} \ds\int_{\phi^{-1}(\delta)}^1 w_t(\phi_t(x)) \mu_0(dx)  \end{align*}
We now recall from Equation \ref{eq:wtx} that 
\[w_t(\phi_t(x)) = \exp \left( \int_0^t \lambda \left[ G\left(\phi_s(x)\right) - \langle G(\cdot) \rangle_{\mu_s} \right] ds\right) =  \exp \left( \int_0^t \lambda \left[ \gamma \phi_s(x) + \alpha \phi_s(x)^2 - \langle G(\cdot) \rangle_{\mu_s} \right] ds  \right) \] Recalling that $G(x) \geq 0$ for $x \in [0,1]$ for Cases I-III of the PD, we know $\langle G(\cdot) \rangle_{\mu_s} \geq 0$, so we have that $\exp \left( - \lambda \int_0^t \langle G(\cdot) \rangle_{\mu_s} ds \right) \leq 1$. 
For the Case I PD (with $\alpha < 0$), we know from Equation \ref{eq:PDcharacteristicsrankingsCaseI} that $\Psi_s(|\beta|+|\alpha|;x) \leq \phi_s(x)$ and $\phi_s(x) \leq \Psi_s(|\beta|;x)$. Denoting $k_s = |\beta|$ and $k_f = |\beta| + |\alpha|$, we use these inequalities and the expressions for $\int_0^t \Psi_s(k;x)ds$ and $\int_0^t \Psi_s(k;x)^2 ds$ from Equations \ref{eq:psiintegral} and \ref{eq:psisquaredintegral} to see that there is a constant $M_1(\delta)$ such that for $t \geq 0$ and $x \in [\delta,1]$,   
\begin{align*} w_t(\phi_t(x)) &= \exp\left(\lambda \gamma \int_0^t \phi_s(x) ds + \lambda \alpha \int_0^t \phi_s(x)^2 ds - \lambda  \int_0^t \langle G(\cdot) \rangle_{\mu_s} ds \right)  \\ &\leq \exp\left( \lambda \gamma \int_0^t  \Psi_s(k_s;x) ds + \lambda \alpha \int_0^t \Psi_s(k_f;x) ^2 ds   \right) \\ &= e^{\lambda \left(\gamma + \alpha\right) t} \frac{\left(x + (1-x) e^{-k_s t}\right)^{\frac{\lambda \gamma}{k_s}}} {\left(x + (1-x) e^{-k_f t}\right)^{\frac{|\alpha| \gamma}{k_f}}} \exp\left(\frac{x}{k_f} \left(\frac{(1-x)e^{-k_ft} + (x-1)}{x + (1-x)e^{-k_f t}} \right) \right) \\ &\leq M_1(\delta) e^{\lambda\left(\gamma + \alpha \right)t} \end{align*}
For the Case III PD, we know from Equation \ref{eq:PDcharacteristicsrankingsCaseIII} that $\phi_s(x) \leq \Psi_s(|\beta| - \alpha;x)$. Writing $k_s = |\beta| - \alpha$, we use this inequality and the integrals along the simplified characteristic curve $\Psi_s(k;x)$ to see that there is a constant $M_2(\delta)$ such that for $t \geq 0$ and $x \in [\delta,1]$,  
\begin{align*} w_t(\phi_t(x)) &= \exp\left(\lambda \gamma \int_0^t \phi_s(x) ds + \lambda \alpha \int_0^t \phi_s(x)^2 ds \right)  \\ &\leq \exp\left( \lambda \gamma \int_0^t  \Psi_s(k_s;x) ds + \lambda \alpha \int_0^t \Psi_s(k_s;x) ^2 ds   \right) \\ &= e^{\lambda \left(\gamma + \alpha\right) t} \left(x + (1-x) e^{-k_s t} \right)^{\frac{\lambda (\gamma + \alpha)}{k_s}} \exp\left(\frac{x}{k_s} \left(\frac{(1-x)e^{-k_st} + (x-1)}{x + (1-x)e^{-k_s t}} \right) \right) \\ & \leq M_2(\delta) e^{\lambda (\gamma + \alpha)t} \end{align*}
Because $\alpha = 0$ for the Case II PD, a similar estimate follows from either of the other. For $M(\delta) = \max\left(M_1(\delta),M_2(\delta)\right)$, we can put this together to say that  $w_t(\phi_t(x)) \leq M e^{\lambda \left( \gamma + \alpha\right) t}$. This now tells us that 
\begin{align*} \bigg| \ds\int_0^1 \psi(x) \mu_t(dx) - \psi(0) \bigg| &\leq \epsilon +  2 ||\psi||_{\infty} M(\delta) e^{\lambda \left(\gamma + \alpha\right) t} \int_{\psi^{-1}(x)}^1 \mu_0(dx) \\ &= \epsilon +  2 ||\psi||_{\infty} M(\delta) e^{\lambda \left(\gamma + \alpha\right) t} \mu_0\left([\phi^{-1}(x),1] \right)  \end{align*}

Now we show how to complete the proof using comparison results that hold for the Case I PD, and then subsequently address Case III. From Equation \ref{eq:PDcharacteristicsbackwardrankingsCaseI}, we know that that within-group trajectories for the Case I PD satisfty $\Psi^{-1}_t(|\beta|;x) \leq \phi_t^{-1}(x)$, which allows us to deduce that \begin{align*}  \mu_0\left([\phi_t^{-1}(x),1]\right) &\leq \mu_0\left([\Psi_t^{-1}(|\beta|;x),1]\right) = \mu_0\left(\left[\frac{x}{x + (1-x) e^{-|\beta| t}} , 1 \right]\right) = \mu_0\left(\left[1 - \frac{(1-x) e^{-|\beta| t}}{x + (1-x) e^{-|\beta| t}},1\right]\right) \end{align*}
Then we see, for $x \in [\delta,1]$ that $\exists D > 0$ (namely $\frac{1}{\delta} - 1$) and sufficiently large $t$ such that we can use the fact that $\frac{1-x}{x} \leq \frac{1-\delta}{\delta} = D$ and our assumption about the Hölder exponent of $\mu_0(dx)$ near $x=1$ to deduce that  \[\mu_0\left([\phi_t^{-1}(x),1]\right) \leq \mu_0\left(\left[1 - \left(\frac{1-x}{x} \right) e^{- |\beta| t},1 \right]\right) \leq \mu_0\left([1 - De^{-|\beta| t},1]\right) \approx C D^{\theta} e^{- |\beta| \theta t} \] 

If $\lambda (\gamma + \alpha) < (|\beta| + \alpha) \theta$, then either $\lambda (\gamma + \alpha) < |\beta| \theta < (|\beta| + \alpha) \theta$ or $ |\beta| \theta < \lambda (\gamma + \alpha)  < (|\beta| + \alpha) \theta$. In the former case, we can use the above estimates and the fact that $\lambda(\gamma + \alpha) < |\beta| \theta$ to see that 
\begin{align*} \bigg| \ds\int_0^1 \psi(x) \mu_t(dx) - \psi(0) \bigg| &\leq \epsilon + 2 ||\psi||_{\infty} C M D^{\theta} e^{\left[\lambda \left(\gamma + \alpha\right) - |\beta| \theta \right] t} < 2 \epsilon \: \: \mathrm{as} \: \: t \to \infty \end{align*} 
and we can deduce that $\mu_t(dx) \rightharpoonup \delta(x)$ as $t \to \infty$ when $\lambda \left( \gamma + \alpha \right) < |\beta| \theta$ for the Case I PD. 

When $|\beta| \theta < \lambda \left(\gamma + \alpha\right) < \left( |\beta| + \alpha \right) \theta$, we must take a more refined choice of slower characteristic curve $\Psi_t(|\beta| + \alpha k;x)$ to compare with $\phi_t(x)$. 
Because $ \lambda \left(\gamma + \alpha\right) < \left( |\beta| + \alpha \right) \theta$, for each choice of $\lambda$,  there exists $k \in [0,1)$ such that $\lambda (\gamma + \alpha) = (|\beta| + \alpha k) \theta$. 
If $\lambda (\gamma +\alpha) = \left(|\beta| + \alpha k\right) \theta$ for $k < 1$, then we know that $\lambda(\gamma + \alpha) < \left( |\beta| + \alpha \Gamma \right) \theta$ for $\Gamma := \frac{k+1}{2} \in (k,1)$. Now we look to make use of a comparison principle to study the solutions of $\phi_t(x)$ using solutions of $\Psi_t(|\beta| + |\alpha| \Gamma;x)$.

Solving backwards from $\delta$, we know for any $\Gamma \in [0,1)$ that $\exists \tGamma$ such that $\forall t \geq T$, $\phi_t^{-1}(\delta) \geq \Gamma$. For $x \geq \Gamma$, we further see that 
$|\beta| + |\alpha| x \geq |\beta| + |\alpha| \Gamma$ and therefore we see, from looking at solutions backward in time, that \[ x(1-x)\left(|\beta| + |\alpha| \Gamma\right) \leq  x (1-x) \left( |\beta| + |\alpha| x \right) \: \: \mathrm{and} \: \:  \dsddt{} \Psi_t^{-1}\left(|\beta| + |\alpha| \Gamma ; x\right) \leq  \dsddt{} \phi_t^{-1}(x)  \]
Therefore, for $t \geq \tGamma$ and $x \geq \delta$, characteristic curves $\phi_t^{-1}(x)$ travel faster than solutions to $\Psi_t^{-1}\left(|\beta| + |\alpha| \Gamma ; x\right)$ and we have that $\Psi_t^{-1}\left(|\beta| + |\alpha| \Gamma ; x\right) \leq \phi_t^{-1}(x)$. %

In fact, tracking the solution to $\dsddt{} \Psi_t(|\beta| + |\alpha| \Gamma;\cdot)$ which starts at $\Gamma$ at time $T_{\delta}$ (coinciding with the solution of $\phi^{-1}_t(\delta)$ at this point), we see the subsequent trajectory will be given at times $t > T$ by \begin{align*} \Psi_{t-T}^{-1}\left( |\beta| + |\alpha| \Gamma ; \Gamma \right) = \frac{\Gamma}{\Gamma + (1-\Gamma) e^{-\left( |\beta| + |\alpha| \Gamma\right)(t-T)}} &= 1 - \frac{(1-\Gamma) e^{-\left( |\beta| + |\alpha| \Gamma\right)(t-T)}}{\Gamma + (1-\Gamma) e^{-\left( |\beta| + |\alpha| \Gamma\right)(t-T)}} \\ & \geq 1 - \left( \frac{1-\Gamma}{\Gamma} \right) e^{\left(|\beta| + |\alpha| \Gamma \right) T} e^{-\left(|\beta| + |\alpha| \Gamma \right) t}   \end{align*}
Then recalling that $\phi_t(x) \geq  \Psi_{t-T}^{-1}\left( |\beta| + |\alpha| \Gamma ;\Gamma \right)$ for $x \in [\Gamma,1]$, we have that 
\[ \mu_0\left([\phi_t^{-1}(\delta),1]\right) \leq \mu_0\left(\left[\Psi_{t-T}^{-1}\left( |\beta| + |\alpha| \Gamma ;\Gamma \right),1 \right]\right) \leq \mu_0 \left(\left[ 1 - \left( \frac{1-\Gamma}{\Gamma} \right) e^{\left(|\beta| + |\alpha| \Gamma \right) T} e^{-\left(|\beta| + |\alpha| \Gamma \right) t} ,1 \right]\right) \] Denoting $K = \left( \frac{1-\Gamma}{\Gamma} \right) e^{\left(|\beta| + |\alpha| \Gamma \right) T}$ and using our assumption about the Hölder exponent near $x=1$ for $\mu_0(dx)$, we can say, for sufficiently large $t$, that \[\mu_0\left([\phi_t^{-1}(x),1]\right) \leq C \left[\left( \frac{1-\Gamma}{\Gamma} \right) e^{\left(|\beta| + |\alpha| \Gamma \right) T} \right]^{\theta} e^{-\left(|\beta| + |\alpha| \Gamma \right) \theta t} \approx C  K^{\theta} e^{-\left(|\beta| + |\alpha| \Gamma \right) \theta t} \]
Combining this with our estimates from above and using that $\lambda \left(\gamma + \alpha\right) < \left(|\beta| + |\alpha| k \right) \theta <  \left(|\beta| + |\alpha| \Gamma \right) \theta$ for our given $\lambda$, we can deduce that 
\begin{align*} \bigg| \ds\int_0^1 \psi(x) \mu_t(dx) - \psi(0) \bigg| &\leq \epsilon + 2 ||\psi||_{\infty} C M  K^{\theta} e^{\left[\lambda \left(\gamma + \alpha\right) - \left(|\beta| + |\alpha| \Gamma \right) \theta \right] t} < 2 \epsilon \: \: \mathrm{as} \: \: t \to \infty \end{align*} 
and therefore $\mu_t(dx) \rightharpoonup \delta(x)$ as $t \to \infty$ when $\lambda \left(\gamma + \alpha\right) < \left(|\beta| + |\alpha| k \right) \theta$. Because such a $k$ can be found whenever $\lambda \left(\gamma + \alpha\right) < \left(|\beta| + |\alpha| \right) \theta$, we now know that $\mu_t(dx) \rightharpoonup \delta(x)$ as $t \to \infty$ whenever this inequality is satisfied for the Case I PD.

Now we turn our attention to Case III. Because $\alpha > 0$ for Case III, so our hypothesis on $\lambda$ can be written as $\lambda \left( \gamma + \alpha \right) < \left( |\beta| - \alpha \right) \theta$. We know from Equation \ref{eq:PDcharacteristicsbackwardrankingsCaseIII} that the trajectories for this case satisfy $ \Psi_t^{-1}(|\beta|-\alpha;x) \leq \phi_t^{-1})(x)$, which allows us to deduce that 
\begin{align*} \mu_0\left( [\phi_t^{-1}(x),1] \right) \leq \mu_0\left([\Psi_t^{-1}\left(|\beta|-\alpha ; x \right) ,1] \right) &= \mu_0\left(\left[ \frac{x}{x + \left(1 - x\right) e^{\left(|\beta| - \alpha \right) t}}, 1 \right] \right) \\ &=  \mu_0\left(\left[ 1 - \frac{\left(1 - x\right) e^{\left(|\beta| - \alpha \right) t}}{x + \left(1 - x\right) e^{\left(|\beta| - \alpha\right) t}}, 1 \right] \right)  \end{align*}
In an analogous manner to Case I, we know that for $D = \frac{1}{\delta} - 1$, that for $x \in [\delta,1]$ and sufficiently large $t$, we can use the fact that $\frac{1-x}{x} \leq \frac{1-\delta}{\delta} = D$ and our assumption about the \holder exponent of the initial distribution $\mu_0(dx)$ to see that
\[  \mu_0\left( [\phi_t^{-1}(x),1] \right) \leq \mu_0\left( \left[1 - \left( \frac{1-x}{x} \right) e^{-{|\beta| - \alpha} t}, 1\right] \right) \leq \mu_0\left(\left[ 1 - D e^{-{|\beta| - \alpha} t}, 1\right] \right) \approx C D^{\theta}  e^{-{|\beta| - \alpha} \theta t}\] 
Using the above estimate that $w_t(\phi_t(x)) \leq M(\delta) e^{\lambda \left(\gamma + \alpha\right) t}$, we further see, when $\lambda \left( \gamma + \alpha \right) < \left( |\beta| - \alpha \right) \theta$, that
\[ \bigg| \int_0^1 \psi(x) \mu_t(dx) - \psi(0) \bigg| \leq \epsilon + 2 ||\psi||_{\infty} C M(\delta) D^{\theta} e^{\left[\lambda \left( \gamma + \alpha \right) - \left(|\beta| - \alpha\right) \theta\right] t} < \epsilon \: \: \mathrm{as} \: \: t \to \infty, \]
so we can conclude that $\mu_t(dx) \rightharpoonup \delta(x)$ at $t \to \infty$ for the Case III PD when $\lambda \left( \gamma + \alpha \right) < \left( |\beta| - \alpha \right) \theta$.

\end{proof}

\begin{proposition} \label{prop:PD4delta} Consider Case IV of the PD ($\gamma < 0$ and $\alpha > 0$) and an initial distribution $\mu_0(dx)$ with \holder exponent $\theta$ near $x=1$. If $\lambda (\gamma + \alpha) + \frac{\gamma^2}{4 \alpha} < (|\beta| + \alpha) \theta$, then $\mu_t(dx) \rightharpoonup \delta(x)$. 
\end{proposition}

\begin{remark} The distinction with case IV is that $G(x)$ is not minimized at $x=0$, but rather has an interior minimizer at $x_{min} = -\frac{\gamma}{2 \alpha}$ with value $G(x_{\min}) =  \frac{\gamma^2}{4 \alpha}$. Unlike Cases I-III, we don't know that $\int_0^t \langle G \rangle_{\mu_s} ds \leq 0$, and can only say that $\int_0^t \langle G \rangle_{\mu_s} ds \geq - \frac{\gamma^2}{4 \alpha} t$. This allows us to say that
\[ \bigg| \ds\int_0^1 \psi(x) \mu_t(dx) - \psi(0) \bigg| \leq \epsilon +  2 ||\psi||_{\infty} M e^{\left[\lambda \left(\gamma + \alpha\right) + \frac{\gamma^2}{4 \alpha} \right] t} \int_{\psi^{-1}(x)}^1 \mu_0(dx) \]
and then we require $\lambda (\gamma + \alpha) + \frac{\gamma^2}{4 \alpha} < (|\beta| + \alpha) \theta$ to show that the exponential term converges to $0$ in the long-time limit. However, we also show in Section \ref{sec:PDsteady} that Case IV PDs have density steady states precisely when $\lambda (\gamma + \alpha)  > (|\beta| + \alpha) \theta$, so we conjecture that a more refined argument can show that $\mu_t(dx) \rightharpoonup \delta(x)$ when $\lambda (\gamma + \alpha)  < (|\beta| + \alpha) \theta <  \lambda (\gamma + \alpha) + \frac{\gamma^2}{4 \alpha} $.
\end{remark}

\subsection{Steady States of General Multilevel PD} \label{sec:PDsteady}

In this section, we explore the steady state solutions for the multilevel PD dynamics. 
Steady state solutions $f(x)$ of Equation \ref{eq:replicatorpdeparam} satisfy 
\[ \dsdel{}{x}\left[ x (1-x) \left( \beta + \alpha x \right) f(x) \right] = \lambda f(x) \left[ \gamma x + \alpha x^2 - \left( \gamma M_1^f + \alpha M_2^f \right) \right]  \]
Because $\beta < 0$ for the Prisoners' Dilemma, we hereafter write $\beta$ as $-|\beta|$. We can solve the steady state ODE by separating variables and a partial fraction expansion, yielding steady state densities of the form 
 \[ f(x) = Z_f^{-1} x^{\frac{\lambda}{|\beta|} \left( \gamma M_1^f + \alpha M_2^f  \right) - 1} \left(1-x\right)^{\left(\frac{\lambda}{ |\beta| - |\alpha|}\right) \left(\gamma + \alpha - \left( \gamma M_1^f + \alpha M_2^f \right) \right) - 1} \left( |\beta| - \alpha x\right)^{\frac{\lambda \alpha}{|\beta| \left( |\beta| - \alpha\right)} \left( \left( \gamma M_1^f + \alpha M_2^f \right) - \alpha \left(\gamma + |\beta| \right) \right)-1} \]
From Proposition \ref{prop:PDholder}, we know that the \holder exponent of $\mu_t(dx)$ near $x=1$ is preserved in time, so it sensible to describe our steady state densities in terms of their \holder exponents as well. %
Using Definition \ref{def:Holderx1} for $\mu(dx) := f(x) dx$, we can compute the \holder exponent of the steady state densities. Denoting the exponents of $x$ and $|\beta| - \alpha x$ by $A$ and $B$, respectively, we find that 
\begin{dmath*} \ds\lim_{y \to 0} \frac{\mu\left(\left[1-y,1\right] \right)}{y^{\Theta}} =  \ds\lim_{y \to 0} \frac{\int_{1-y}^1 f(z) dz}{y^{\theta}} = \ds\lim_{y \to 0} \frac{f(1-y)}{\Theta y^{\Theta - 1}} \\ 
= \ds\lim_{y \to 0} \left[\Theta^{-1} Z_f^{-1} \left(1 - y \right)^A \left( |\beta| - \alpha + \alpha y \right)^B y^{\left(\frac{\lambda}{|\beta|-\alpha}\right) \left(\gamma + \alpha - \left( \gamma M_1^f + \alpha M_2^f \right) \right) - \Theta}  \right] \\ =   \left\{
     \begin{array}{cr}
    0 & :  \Theta < \lambda \left( |\beta| - \alpha \right)^{-1}  \left(\gamma + \alpha - \left( \gamma M_1^f + \alpha M_2^f \right) \right) \\
      Z_f^{-1}  \left( |\beta| - \alpha\right)^B & : \Theta =  \lambda \left( |\beta| - \alpha \right)^{-1}  \left(\gamma + \alpha - \left( \gamma M_1^f + \alpha M_2^f \right) \right)  \\ \infty &: \Theta >  \lambda \left(  |\beta| - \alpha \right)^{-1}  \left(\gamma + \alpha - \left( \gamma M_1^f + \alpha M_2^f \right) \right) 
     \end{array}
   \right. \end{dmath*}
Therefore, for a given average payoff $\langle G(\cdot) \rangle_{f} = \gamma M_1^f + \alpha M_2^f$, the characterization from Definition \ref{def:Holderx1} tell us that the Hölder exponent near the endpoint $x=1$ of our steady-state density $f(x)$ is equal to
 \begin{equation} \label{eq:steadyHolder} \theta = \lambda \left(|\beta| - \alpha\right)^{-1} \left( \gamma +  \alpha - \left(\gamma M_1^f + \alpha M_2^f  \right) \right) \end{equation} 
We can also rearrange this to express the average payoff in the whole population ($\int_0^1 G(y) dy = \gamma M_1^f + \alpha M_2^f$) in terms of the Hölder exponent $\theta$, for given $\lambda$ and $\theta$, as follows
\begin{equation} \label{eq:averagepopulationpayoff}\int_0^1 G(y) f(y) dy =  \gamma M_1^f + \alpha M_2^f =  \left( \gamma + \alpha \right) -\frac{ \left( |\beta| - \alpha \right)}{\lambda} \theta \end{equation} %
Rewriting our steady states in terms of $\theta$, we have \begin{equation} \label{eq:pdsteadybetatheta} f^{\lambda}_{\theta}(x) = Z_f^{-1}  x^{|\beta|^{-1} \left(\lambda \left(\gamma + \alpha \right) - \left( |\beta| - \alpha \right) \theta \right) - 1} \left( 1 - x \right)^{\theta - 1} \left( |\beta| - \alpha x \right)^{- \frac{\lambda}{|\beta|} \left( \gamma + |\beta| + \alpha \right) - \frac{\alpha}{|\beta|}\theta -1}  \end{equation}

For $\theta > 0$, we see that there is a threshold level of between-group selection strength $\lambda^*$ such that the steady states $f_{\theta}(x)$ are integrable if and only if
\begin{equation} \label{eq:pdlambdastar} \lambda > \lambda^* := \left( \frac{|\beta| - \alpha}{\gamma + \alpha} \right) \theta \end{equation}
and that it is not possible to have an integrable steady state for any positive $\lambda$ when $\gamma \to -\alpha$.

We can further understand the existence of steady state densities by recharacterizing the threshold selection strength needed for integrability of our steady states $\lambda^* := %
\left(\frac{T-R}{R - P}\right) \theta$, and make the following observations 

\begin{itemize}
\item  $\lambda^*$ is increasing in $T -R$:  increasing the incentive to defect against a cooperator increases the relative strength of between-group selection need to sustain cooperation.
\item $\lambda^*$ is decreasing in $R - P$: increasing the relative advantage of mutual cooperation over mutual defection decreases the relative strength of between-group selection needed to sustain cooperation. 
\item $\lambda^* \to \infty$ as $R - P \to 0$: as mutual cooperation loses its advantage over mutual defection, then between-group selection cannot maintain cooperation at any selection strength.
\end{itemize}

 In terms of our expressions for $G(x)$ and $\pi_C(x) - \pi_D(x)$ from Equations \ref{eq:grouppayoffparamsimplified} and \ref{eq:picminuspid}, we can rewrite our threshold selection strength as 
\begin{equation} \label{eq:pdlambdapayoff} \lambda^* = \left( \frac{\pi_D(1) - \pi_C(1)}{G(1) - G(0)} \right) \theta , \end{equation}
where we have included in our expression $G(0) = 0$ because the expression would have held even if we had not shifted $G(x)$ by $P$, and because the difference between the group payoff at the unstable and stable equilibria for the within-group dynamics also shows up in an analogous way for the threshold in the multilevel HD dynamics studied in Section \ref{sec:HDsteady}.
This form of the threshold $\lambda^*$ highlights the struggle of trying to achieve cooperation by multilevel selection as a tug-of-war between the individual-level advantage of a defector outperforming cooperators in a many-cooperator group and the group-level advantage of many-cooperator groups achieving higher average payoff than the all-defector groups that arise from the effects of individual-level selection alone.
We can also rewrite the average payoff at steady state from Equation \ref{eq:averagepopulationpayoff} using the expressions from Equations  \ref{eq:grouppayoffparamsimplified} and \ref{eq:picminuspid} to see that 
\[ \langle G(\cdot) \rangle_{f^{\lambda}_{\theta}} = G(1) - \frac{\left(\pi_D(1) - \pi_C(1)\right) \theta}{\lambda}   \]
Using the expression for $\lambda^*$ from Equation \ref{eq:pdlambdapayoff}, we can further find that
\begin{equation} \label{eq:PDsteadyfitnesspayoff} \langle G(\cdot) \rangle_{f^{\lambda}_{\theta}} = G(1) - \left(\frac{\lambda^*}{\lambda}\right) \left(G(1) - G(0) \right)  \end{equation}
where we see that $\langle G(\cdot) \rangle_{f^{\lambda}_{\theta}} \to G(0) = 0$ as $\lambda \to \lambda^*$ and that $\langle G(\cdot) \rangle_{f^{\lambda}_{\theta}} \to G(1)$ as $\lambda \to \infty$. In particular, we see from Equation \ref{eq:PDsteadyfitnesspayoff} that the average payoff of the population at steady state is always limited by the payoff of the full-cooperator group, so a PD with intermediate average payoff optimum will always see suboptimal levels of cooperation, even in the limit of arbitrarily strong between-group population.

So far, we have shown in Section \ref{sec:PDlongtime} that a population with initial \holder exponent $\theta$ near $x=1$ will converge to a delta-concentration at the full-defector group when $\lambda (\gamma + \alpha) < \left(\beta + \alpha\right) \theta$. Under the alternate condition  $\lambda (\gamma + \alpha) > \left(\beta + \alpha\right) \theta$, we have shown in this section that there is a unique and integrable steady state for each \holder exponent $\theta$, and we showed in Section \ref{sec:Holderpreserve} that the \holder exponent near $x=1$ is preserved in time for solutions of Equation \ref{eq:replicatormeasurepde} for the PD. It is then natural to suspect that populations should converge to the steady state with the same \holder exponent near $x=1$ as the initial distribution, so we make the following conjecture about the long-time behavior of solutions to Equation \ref{eq:replicatormeasurepde}.
\begin{conjecture} \label{con:pddensity} Suppose we have an initial distribution $\mu_0(dx)$ has a \holder exponent of $\theta$ near $x=1$. If $\lambda \left( \gamma + \alpha \right) > (\beta + \alpha) \theta$, then \[\mu_t(dx) \rightharpoonup \mu_{\infty}(dx) = x^{|\beta|^{-1} \left(\lambda \left(\gamma + \alpha \right) - \left( |\beta| - \alpha \right) \theta \right) - 1} \left( 1 - x \right)^{\theta - 1} \left( |\beta| - \alpha x \right)^{- \frac{\lambda}{|\beta|} \left( \gamma + |\beta| + \alpha \right) - \frac{\alpha}{|\beta|}\theta -1} dx,\]  where $Z_f^{-1}$ is a normalizing constant such that $\int_0^1 \mu_{\infty}(dx) = 1$. \end{conjecture}
 This long-term behavior has already been shown to hold for special families of Case I PDs with $\alpha = \beta = -1$ \cite{cooney2019replicator} and for the Case II PD \cite{luo2017scaling}. The main impediment to the general proof of the conjecture is that we have not yet shown that solutions $\mu_t(dx)$ of Equation \ref{eq:replicatormeasurepde} necessarily converge to a steady state. However, for the remainder of this section, we will study the properties of the density steady states given by Equation \ref{eq:pdsteadybetatheta}, knowing that these are time-independent solutions of Equation \ref{eq:replicatormeasurepde}, with the potential additional relevance that these steady states could be the long-time outcome for initial data with \holder exponent $\theta$ near $x=1$ when $\lambda \left( \gamma + \alpha \right) > (\beta + \alpha) \theta$. In Figure \ref{fig:pdsteadydensities}, we illustrate sample steady state densities for PDs in which average group payoff is maximized by full-cooperator groups (left) and by groups with an intermediate level of cooperation (right). In the fomer case, we see that values of $\lambda$ as small as 9 can produces steady states with high levels of cooperation and many groups close to full cooperation. In the latter case in which group payoff is maximized by groups with 80 percent cooperator, we see in Figure \ref{fig:pdsteadydensities} that most groups feature at most 60 percent cooperation even for values $\lambda$ as large as 300, and we will further address the behavior of the limit $\lambda \to \infty$ in Proposition \ref{prop:PDpeak}.

\begin{figure}[H]
    \centering
   \hspace{-5mm} \includegraphics[width=0.505\textwidth]{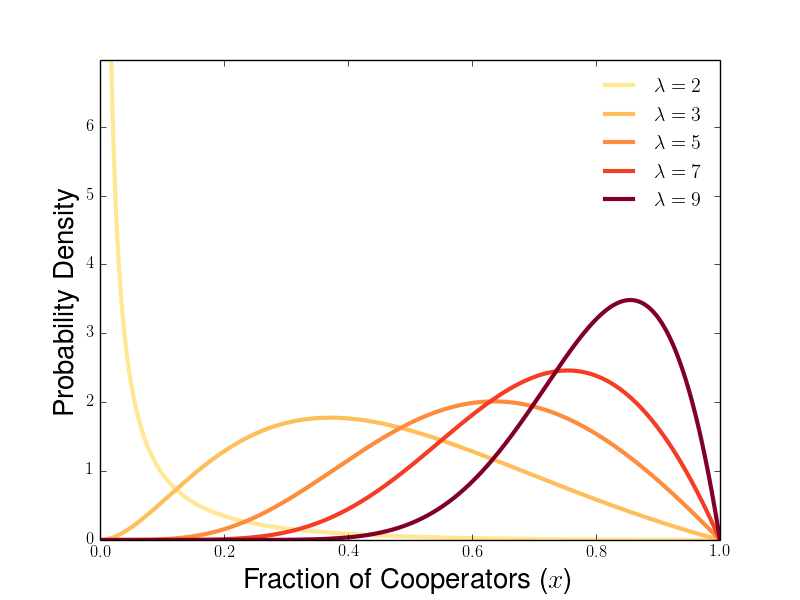} 
    \hspace{-5mm} \includegraphics[width=0.505\textwidth]{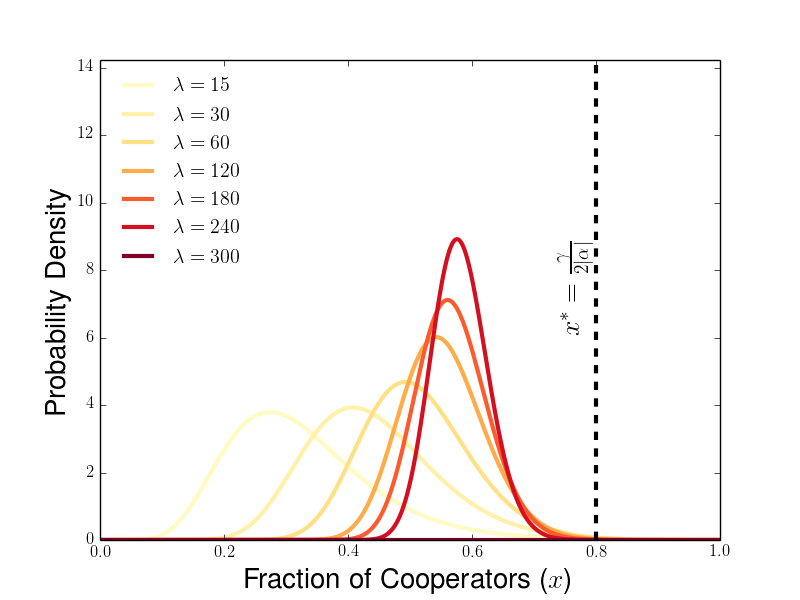}
    \caption{Steady state densities for the PD for various values of $\lambda$. Parameters shown are $\gamma = 2.5$ and $\gamma = -0.5$ (Left) and $\gamma = 3.2$ and $\alpha = -2.0$ (Right), with $\beta = -1$ and $\theta = 2$ for both panels. Dotted line in right panel corresponds to the group type $x^*$ with maximal average payoff.}
    \label{fig:pdsteadydensities}
\end{figure}

Now we will examine which type of group cooperator composition $x$ is most abundant in steady state $f^{\lambda}_{\theta}(x)$. To do this, we denote peak abundance by $\hat{x}_{\lambda}(f^{\lambda}_{\theta}(x)) = \ds\argsup_{x \in [0,1]} f^{\lambda}_{\theta} (x)$.  Because our density-valued steady states become unbounded as $x\to1$ when $\theta < 1$, we will now focus only on steady states when $\theta \geq 1$. We show that if $\gamma \geq 2 \alpha$ (when full-cooperator groups are optimal),  $\hat{x}_{\lambda} \to 1$ as $\lambda \to \infty$. If $\gamma < 2 \alpha$ (the optimal group has both cooperators and defectors), we see that the most abundant group type at steady states has more defectors than the type of group maximizing average payoff, even in the limit at $\lambda \to \infty$.  

\begin{proposition} \label{prop:PDpeak} Suppose $\theta  \geq 1$. If  $\gamma + 2 \alpha < 0$, then $\ds\lim_{\lambda \to \infty} \hat{x}_{\lambda} (f^{\lambda}_{\theta}) = \frac{\gamma + \alpha}{-\alpha} \in (0,1)$ and $\frac{\gamma + \alpha}{-\alpha} < \frac{\gamma}{-2 \alpha} : x^*$. If  $\gamma + 2 \alpha > 0$, $ \ds\lim_{\lambda \to \infty} \hat{x}_{\lambda} (f^{\lambda}_{\theta} ) = 1 = x^*$. 

\end{proposition}

\begin{proof}

We start by differentiating $f^{\lambda}_{\theta}(x)$ and see that $$\dsddx{f^{\lambda}_{\theta}(x)}{x} = g(x) \left[ Z_f^{-1} x^{|\beta|^{-1} \left(\lambda \left(\gamma + \alpha \right) - \left( |\beta| - \alpha \right) \theta \right) - 2} \left( 1 - x \right)^{\theta - 2} \left( |\beta| - \alpha x \right)^{- \frac{\lambda}{|\beta|} \left( \gamma + |\beta| + \alpha \right) - \frac{\alpha}{|\beta|}\theta -2}  \right] $$ where $g(x)$ is a quadratic in $x$ given by \begin{equation} \label{eq:gxdef} g(x) = \left[ \lambda \left(\gamma + \alpha \right) - \left(|\beta| - \alpha \right) \theta - |\beta| \right]  + \left[ - \lambda \gamma + 2 \left(\alpha + |\beta| \right) \right]  x - \left[ \lambda + 3\right] \alpha x^2  \end{equation}
We note that $g(x)$ vanishes at the points \begin{equation} \label{eq:xhatlambdapm} x^{\lambda}_{\pm} = \frac{\lambda \gamma - 2 (\alpha +|\beta|) \pm \ds\sqrt{\left(\lambda \gamma - 2 (\alpha + |\beta|) \right)^2 + 4 \left(\lambda (\gamma + \alpha) - (|\beta| - \alpha) \theta -|\beta| \right) ( 3 + \lambda) \alpha}}{-2 ( 3 + \lambda)\alpha} \end{equation} Using $o(\lambda^2)$ to mean any function $h(\lambda)$ such that $\lim_{\lambda \to \infty} \left(h(\lambda)/\lambda^2\right) = 0$, we can rewrite the above expression as 
\begin{equation} \label{eq:xhatlambdapmolambdasq} x^{\lambda}_{\pm} = \frac{\lambda \gamma - 2 (\alpha +|\beta|)}{-2(3 + \lambda)\alpha}  \mp \frac{1}{2\alpha} \ds\sqrt{ \frac{\left[\gamma^2 + 4\alpha \left(\gamma + \alpha\right) \right] \lambda^2 + o(\lambda^2)}{(3+\lambda)^2 }} \end{equation}

 In the limit of large $\lambda$, that this simplifies to \begin{equation} \label{eq:xhatlambdapminfinity} x^{\infty}_{\pm} := \ds\lim_{\lambda \to \infty} x^{\lambda}_{\pm} = - \frac{\gamma}{2 \alpha} \mp \frac{1}{2\alpha}  \ds\sqrt{\gamma^2 + 4 \alpha \gamma + 4 \alpha^2}  = \frac{- \gamma \mp \sqrt{\left(\gamma + 2 \alpha \right)^2}}{2\alpha} \end{equation} Therefore the critical points of $g(x)$ depend on the sign of $\gamma + 2\alpha$, which corresponds to whether group average payoff $G(x)$ has an interior maximum.

\begin{itemize}
 
\item For $\gamma + 2 \alpha < 0$, 
then $ x^{\infty}_{\pm} = \frac{-\gamma \pm \left(\gamma + 2 \alpha\right)}{2 \alpha}$, and therefore $x^{\infty}_{+} = 1$ and $ x^{\infty}_{-} = - \frac{\left(\gamma + \alpha\right)}{\alpha}$. %
Because $\gamma + \alpha = R - P > 0$ for the PD, we see that $x^{\infty}_{-} < 0$ if $\alpha > 0$, and therefore the only feasible critical point is ${x}^{\infty}_{+} = 1$ in the limit $\lambda \to \infty$ when $\gamma + 2 \alpha < 0$ and $\alpha > 0$.
When $\alpha < 0$, \[ x^{\infty}_{-} = - \frac{\left(\gamma + \alpha\right)}{\alpha} = -\frac{\gamma}{2\alpha} - \left[ \frac{\gamma}{2\alpha} + 1 \right] = x^* - \underbrace{ \left[ \frac{\gamma + 2 \alpha}{2\alpha} \right]}_{>0} <  x^*  \]
So we see that $g(x)$ has a unique interior critical point $x^{\infty}_{-} \in [0,1]$ and that this interior critical point satisfies $x^{\infty}_{-} < x^* = \frac{-\gamma}{2\alpha}$, having fewer cooperators than the type of group with maximum average payoff.

\item For $\gamma + 2 \alpha > 0$, %
$x^{\infty}_{\pm} = \frac{-\gamma \mp \left(\gamma + 2 \alpha \right)}{2 \alpha}$ and then $x^{\infty}_{-} = 1$ and $x^{\infty}_{-} = - \frac{\left(\gamma + \alpha \right)}{\alpha}$. Again, if $\alpha > 0$, then we see that $x^{\infty}_{+} < 0$, and the only feasible critical point for $g(x)$ is $x^{\infty}_{-} = 1$. When $\alpha < 0$, we see that   \[x^{\infty}_{+} = - \frac{\left(\gamma + \alpha\right)}{\alpha} = 1 - \underbrace{\left( \frac{\gamma + 2 \alpha}{\alpha} \right)}_{< 0} > 1, \] so the unique feasible critical point for $g(x)$ is also $x^{\infty}_{-} = 1$ in this case.

\end{itemize}

For $\lambda > \frac{(|\beta| - \alpha) \theta}{\gamma + \alpha}$, we have from Equation \ref{eq:pdsteadybetatheta} that $f^{\lambda}_{\theta} (0) = f^{\lambda}_{\theta}(1) = 0$ when $\theta \geq 1$ (when $\theta < 1$, the steady state density blows up near $x=1$). Because $0$ and $1$ are the only possible critical points of $f^{\lambda}_{\theta}(x)$ other than those of $g^{\lambda}(x)$, we can deduce that $\ds\lim_{\lambda \to \infty} \hat{x}_{\lambda}(f^{\lambda}_{\theta}) = - \frac{\left(\gamma + \alpha\right)}{\alpha}$ when and $\alpha < 0$ and $\gamma + 2 \alpha < 0$ and that $\ds\lim_{\lambda \to \infty} \hat{x}_{\lambda}(f^{\lambda}_{\theta}) = 1$ for $\gamma \geq 1$ for other Prisoners' Dilemmas. 
\end{proof}

Based on the conditions derived above, we see for Case I PDs (in which $\alpha < 0$) that $\lim_{\lambda \to \infty} \hat{x}^{\lambda} < x^*$ if an only if $\alpha + 2 \gamma < 0$. Only when $\alpha < 0$ and $\gamma + 2 \alpha$ is it the case that the group type with maximal fitness features a mix of cooperators and defectors. 
Notably, in cases II-IV (in which $\alpha \geq 0$), in which average group payoff $G(x)$ is maximized by full cooperator groups, peak abundance at steady state $\hat{x}_{\lambda} \to 1$ as $\lambda \to \infty$. Unlike Case I, there is no discrepency between peak abundance at steady state and group type with highest average payoff ($x^* = 1$) in the limit of strong between-group selection. %

In the process of proving Proposition \ref{prop:PDpeak}, we also see that the most abundant group compostion at steady state is given by  
\begin{equation} \label{eq:xhatlambdapd} \hat{x}^{\lambda} = \left\{
     \begin{array}{cl}
       0 & :  0 \leq \lambda < \lambda^* +|\beta| \\
       \hat{x}^{\lambda}_{-} & :\lambda >  \lambda^* + |\beta|
     \end{array}
   \right. \end{equation}
   
   Because we know from Proposition \ref{prop:PD13delta} that $\mu_t(dx) \rightharpoonup \delta\left(x\right)$ when $\lambda < \lambda^*$, we can use Equation \ref{eq:PDsteadyfitnesspayoff} to characterize the average payoff at steady state for all values of $\lambda$ with the piecewise description
 \begin{equation} \label{eq:Glambdahd} \langle G \rangle_{f^{\lambda}_{\theta}} = \left\{
     \begin{array}{cl}
       G\left(0  \right) & :  0 \leq \lambda < \lambda^* \\
   \left(\frac{\lambda^*}{\lambda}\right)  G(0) + \left( 1 -  \left(\frac{\lambda^*}{\lambda}\right)  \right) G(1) & :\lambda >  \lambda^* 
     \end{array}
   \right. \end{equation}

In Figure \ref{fig:PDghostfigure}(left), we plot the average payoff at steady state $\langle G (\cdot) \rangle_{f^{\lambda}_{\theta}}$ and the average payoff of the most abundant group type at steady state $G(\hat{x}_{\lambda})$ as functions of $\lambda$, showing that both tend to $G(1)$ as $\lambda \to \infty$ (lower-dashed line) rather than the maximal possible group payoff (upper-dashed line). In Figure \ref{fig:PDghostfigure}(right), we plot the maximal group payoff $G(x^*)$ and the average payoff at steady state in the limit as $\lambda \to \infty$ ($\ds\lim_{\lambda \to \infty} \langle G (\cdot) \rangle_{f^{\lambda}_{\theta}} = G(1)$). In Figure \ref{fig:PDpeakalphagamma}, we present heat maps of the most abundant group type in the limit of large $\lambda$, $\lim_{\lambda \to \infty} \hat{x}_{\lambda}$, as a function of $\alpha$ and $\gamma$ for all cases of the PD (left) and focusing on Case I (right). In Figure \ref{fig:PDfitnessalphagamma}, we present average payoffs at steady state in the limit of large $\lambda$as a function of $\alpha$ and $\gamma$ for all cases of the PD (left) and for Case I (right). In Figure \ref{fig:pdpeakone}, we plot the group payoff function $G(x)$ for two sets of parameter values, showing how the group payoff achieved by the most abundant group composition $G\left(\lim_{\lambda \to \infty} \hat{x}_{\lambda}\right) = G(1)$.

\begin{figure}[H]
    \centering
   \hspace{-5mm} \includegraphics[width=0.505\textwidth]{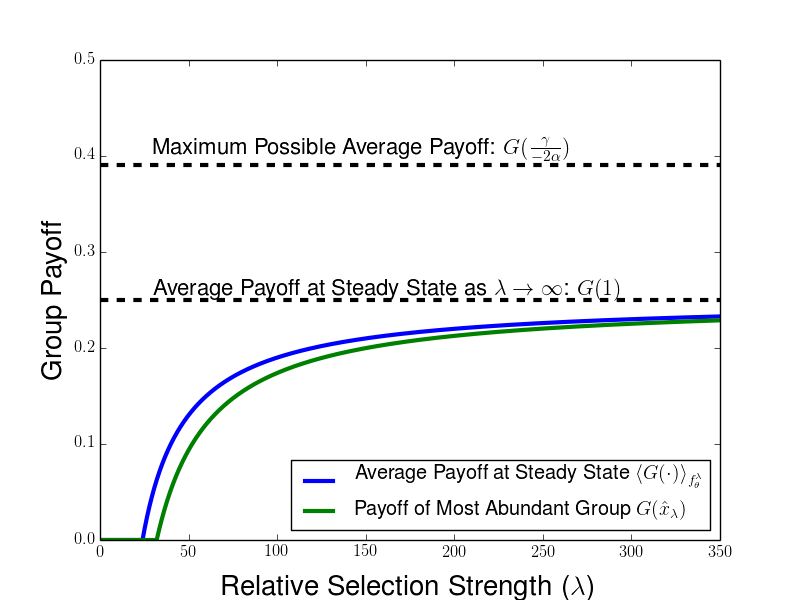} 
    \hspace{-5mm} \includegraphics[width=0.505\textwidth]{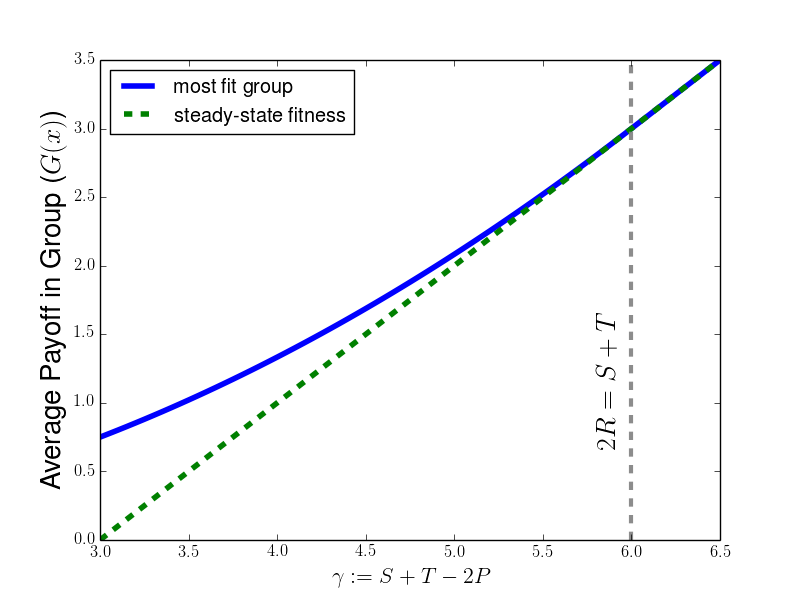}
    \caption{Comparison between average payoff at steady state and optimal average payoff for a group. (Left) Plots of average payoff at steady state $\langle G (\cdot) \rangle_{f^{\lambda}_{\theta}}$ (blue line) and payoff of most abundant group type in steady state distribution $G(\hat{x}_{\lambda})$ (green line) as a function of $\lambda$ in a sample game with optimal group composition $x^* = \frac{3}{4}$. Lower dashed line corresponds to $G(1)$, the limit of average steady state payoff as $\lambda \to \infty)$, and the upper dashed line corresponds to maximal possible group payoff $G(x^*)$. (RIght) Plots of the group type with maximal average payoff $G(x^*)$ (blue solid line) and the average payoff at the population at steady state $\lim_{\lambda \to \infty} \langle G (\cdot) \rangle_{f^{\lambda}_{\theta}}$ in the large $\lambda$ limit (green dashed line), each described as a function of $\gamma$ with a fixed choice of $\alpha = -3$. The gray dotted verticle line corresponds to the value $\gamma =6$ at which the group compostion maximizing payoff $x^* = -\frac{\gamma}{-2\alpha} \big|_{\alpha = -3} = 1$. We notice that the two lines coincide for $\gamma > 6$, when the full-cooperator groups are optimal for between-group competition, while the green dotted line falls below the blue solid line when $\gamma < 6$, as the average payoff of the population falls belows the interior optimal group payoff. In particular, when $\gamma \to 3$, we see that $\langle G (\cdot) \rangle_{f^{\lambda}_{\theta}} \to 0$ and the full-defection outcome is achieved, even though the for this game the group average payoff is $G(x) = 3x (1- x)$, and groups of a fifty-fitfy mix of cooperators and defectors are most favored by between-group competition.  }
    \label{fig:PDghostfigure}
\end{figure}

\begin{figure}[H]
    \centering
   \hspace{-5mm} \includegraphics[width=0.505\textwidth]{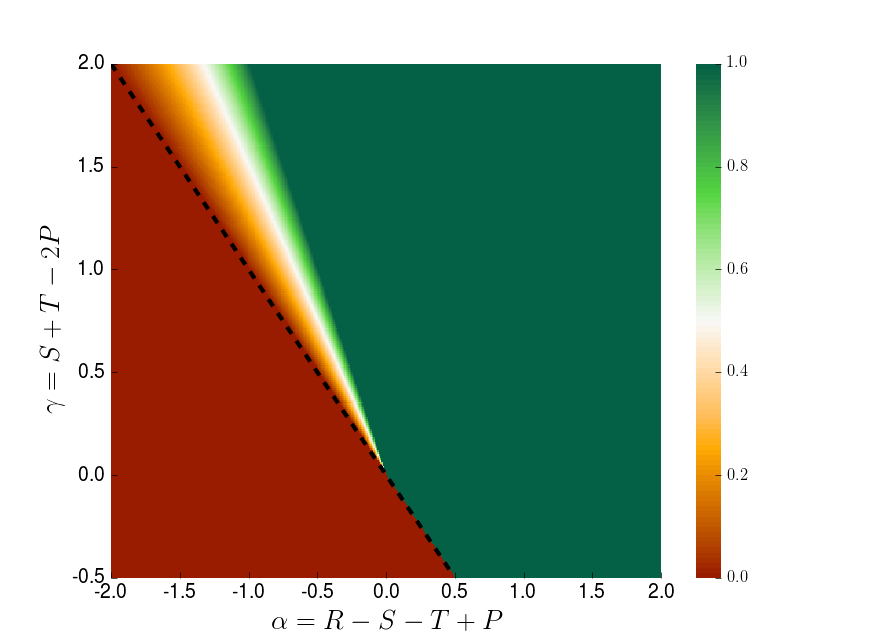} 
    \hspace{-5mm} \includegraphics[width=0.505\textwidth]{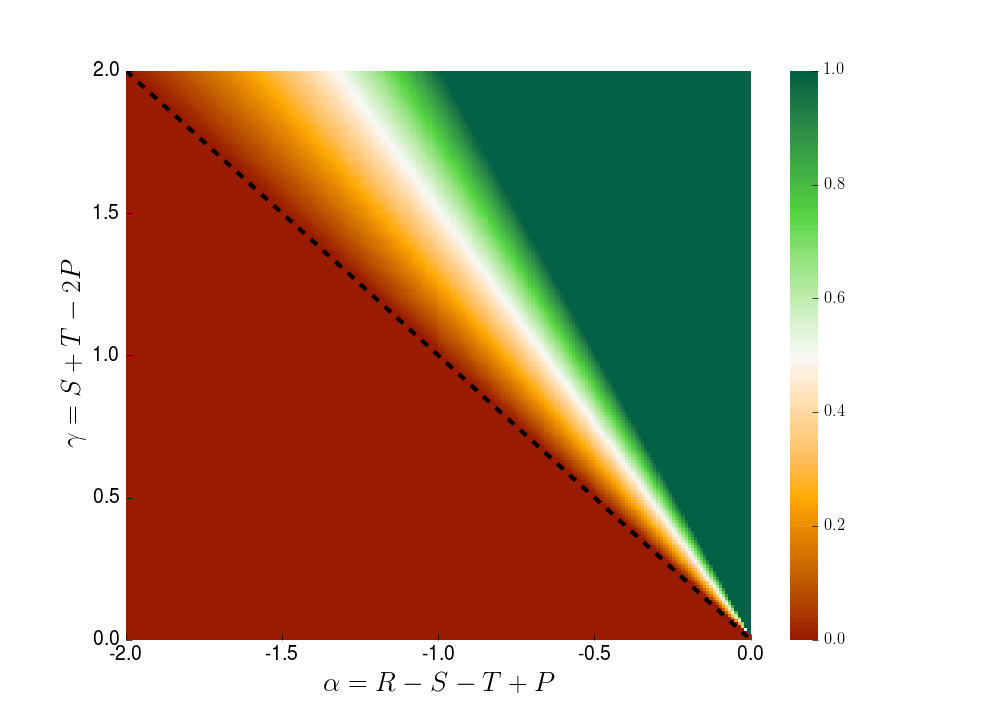}
    \caption{Most abundant group type as $\lambda \to \infty$ for various values of $\gamma$ and $\alpha$. Games in region below dashed line are not Prionsers' Dilemmas.  (Left) Region in $\gamma$ and $\alpha$ parameter space featuring all cases of Prisoners' Dilemmas. (Right) Focus on Case I PDs for which intermediate fitness optima are possible, and the only region in which $\hat{x}_{\lambda} \not\to 1$ as $\lambda \to \infty$ for PD games.}
    \label{fig:PDpeakalphagamma}
\end{figure}

\begin{figure}[H]
    \centering
   \hspace{-5mm} \includegraphics[width=0.535\textwidth]{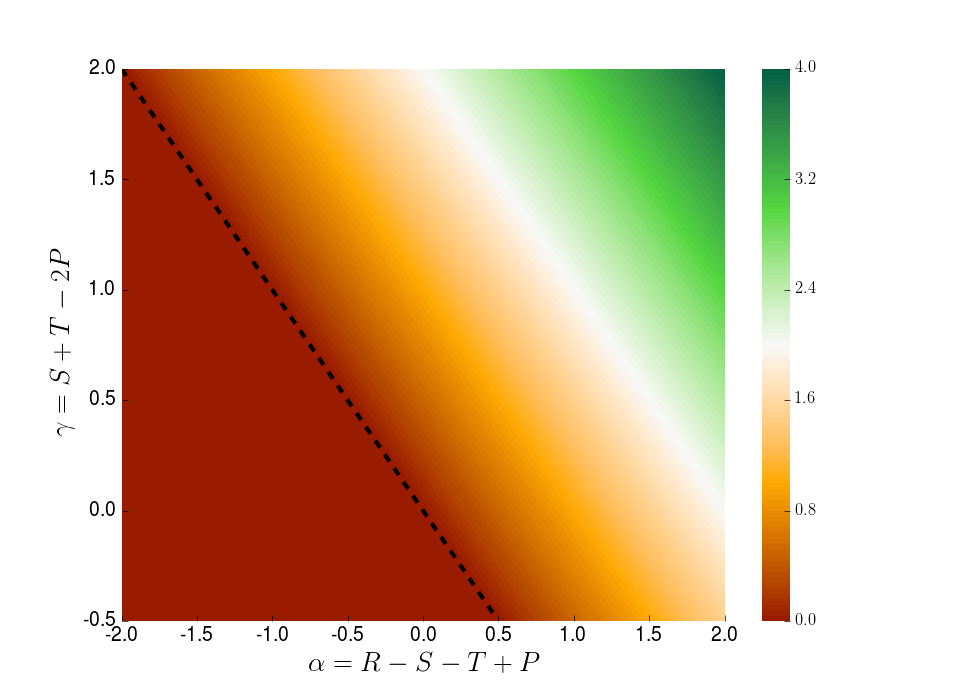} 
    \hspace{-7mm} \includegraphics[width=0.505\textwidth]{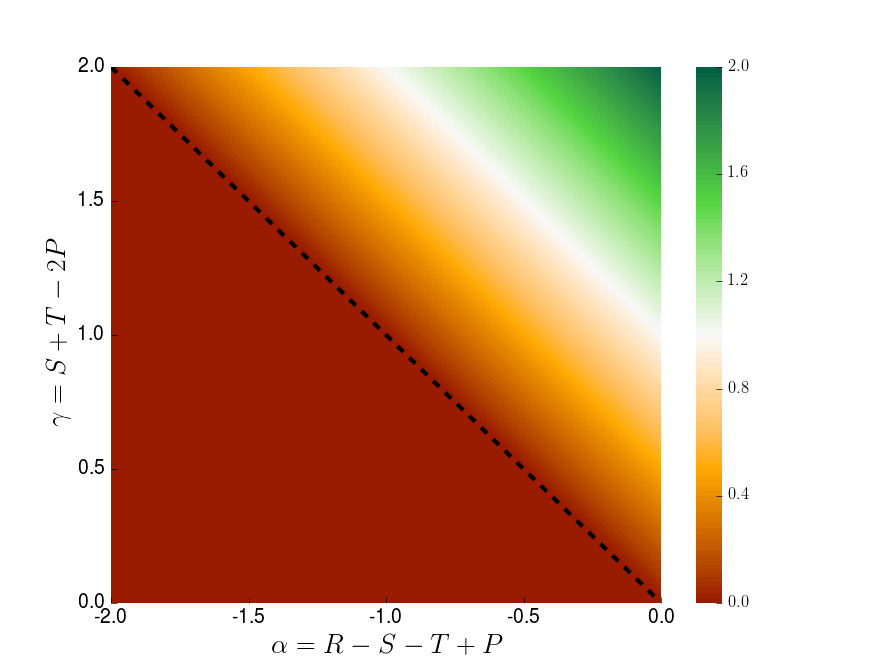}
    \caption{Average payoff of population at steady state as $\lambda \to \infty$ for various values of $\gamma$ and $\alpha$. Games in region below dashed line are not Prionsers' Dilemmas.  (Left) Region in $\gamma$ and $\alpha$ parameter space featuring all cases of Prisoners' Dilemmas. (Right) Focus on Case II PDs for which intermediate fitness optima are possible. Notably, the large $\lambda$ steady state payoff $\lim_{\lambda \to \infty} \langle G \rangle_{f^{\lambda}_{\theta}}  = \gamma + \alpha$ is an increasing function of both $\gamma$ and $\alpha$. }
    \label{fig:PDfitnessalphagamma}
\end{figure}

\begin{figure}[H]
    \centering
   \hspace{-5mm} \includegraphics[width=0.455\textwidth]{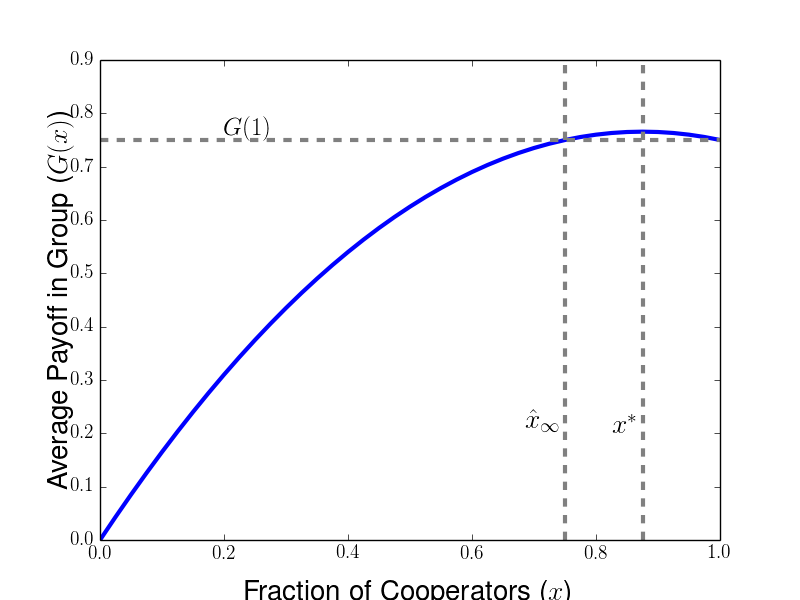} 
      \hspace{-5mm} \includegraphics[width=0.455\textwidth]{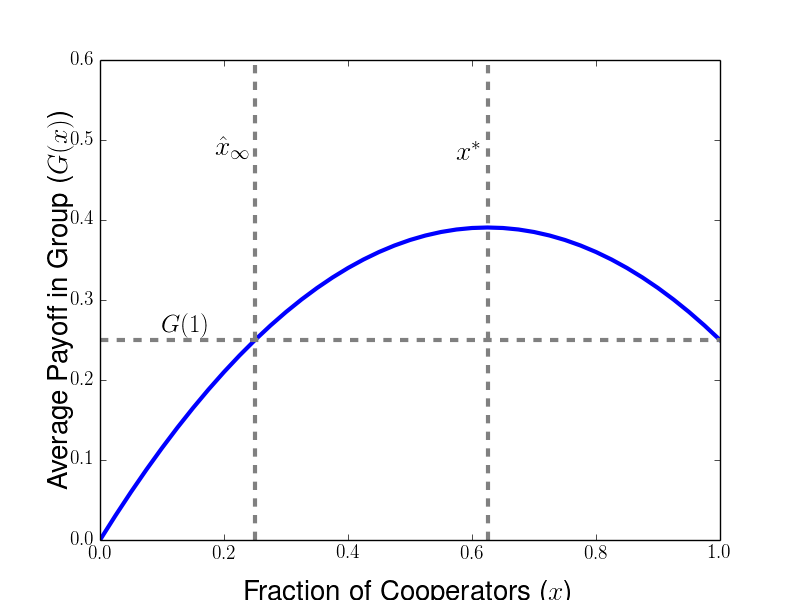}
    \caption{Illustration of peak abudance at steady state as $\lambda \to \infty$ and group type with maximal payoff $x^*$ for Case Ia PDs with $\alpha = -1$ (both) and $\gamma = 1.75$ (left) and $\gamma = 1.25$ (right). In both cases, we see how the modal level of cooperation at steady state $\hat{x}_{\lambda}$ in the large $\lambda$ limit has the same collective payoff as the full cooperator group, $G(\hat{x}_{\infty}) = G(1)$. We also observe that the gap between the maximal possible collective payoff $G(x^*)$ and the maximal payoff achieved at steady state$G(\hat{x}_{\infty})$ increases as $x^*$ decreases from $0.875$ (left) to $x^* = 0.625$.   }
    \label{fig:pdpeakone}
\end{figure}

\section{General HD Games} \label{sec:GeneralHD}

In this section, we describe the dynamics and steady-states for the multilevel dynamics of the Hawk-Dove game. In Section \ref{sec:HDlongtime}, we show that the probability of having groups below the within-group equilibrium vanishes in the long-time limit, and we characterize conditions under which the population converges to a delta-concentration at the within-group equilibrium mix of cooperators and defectors. In Section \ref{sec:HDsteady}, we study the steady state densities which exist under the opposite conditions, supporting all levels of cooperation between the within-group equilibrum and the full-cooperator group.

\subsection{HD Long-Time Behavior} \label{sec:HDlongtime}

Because $\alpha < 0$ for the Hawk-Dove game, we write the expression for group payoff along characteristic curves as $G(\phi_t(x) = \gamma \phi_t(x) - |\alpha| \phi_t(x)^2$. This allows us to describe the change in the population distribution $\mu_t(dx)$ due solely to between-group competition as
\begin{equation} \label{eq:wtphixHD} w_t(\phi_t(x)) = \exp\left( \lambda \left[ \int_0^t \left(\gamma \phi_s(x) - |\alpha| \phi_s(x)^2 \right) ds - \int_0^t \left(\gamma M_1^{\mu_s} - |\alpha| M_2^{\mu_s}  \right) ds \right]  \right) \end{equation}
As in Reference \citep{cooney2019replicator}, it can be helpful to rearrange the above expression by picking $a \in [0,\tfrac{\beta}{|\alpha|}]$ and writing
\begin{dmath} \label{eq:wtphixHDa} w_t(\phi_t(x)) = \exp\left( \left[\lambda \int_0^t \left(\gamma \phi_s(x) - |\alpha| \phi_s(x)^2 - \left( \gamma a - |\alpha| a^2\right)  \right) ds \\ - \int_0^t \left[\gamma  \left(M_1^{\mu_s} - a\right) - \left( |\alpha| M_2^{\mu_s}  - a^2\right) \right] ds \right] \right) \end{dmath}
 where the integral $\int_0^t \left[ \left( \gamma M_1^{\mu_s} - a \right) +  \left(M_2^{\mu_s} - a^2 \right) \right] ds$ provides a measurement of the deviation of the flow of measures $\{\mu_s(dx)\}_{s \in [0,t]}$ from the delta-function $\delta(x-a)$ concentrated at composition $a$ up to time $t$. We will use this integral in a same way that we used that $\int_0^t \langle G(\cdot) \rangle_{\mu_s} ds \geq 0$ in the proof of Proposition \ref{prop:PD13delta}. For convenience, we will denote its integrand by 
 \begin{equation} \label{eq:jas} j_a(s) := \gamma \left(M_1^{\mu_s} - a \right) - |\alpha| \left(M_2^{\mu_s} - a^2 \right)  \end{equation}

We now prove two lemmas which are useful for characterizing the long-time behavior of multilevel HD dynamics. In Lemma \ref{lem:HDleftprob}, we show that the probability of being below the within-group HD equilibrium should vanish in the long-run. This Lemma is a direct extension of Lemma 4.1 from \citep{cooney2019replicator}, making use of the comparison principles of Equations \ref{eq:HDPirankings} and \ref{eq:HDPirankingsbackwards} to extend the argument to games without solvable within-group dynamics. In Lemma \ref{lem:HDleftint}, we show that $\int_0^{\infty} j_a(s) ds$ is bounded below, aiding in the proof of convergence of the population to a delta-function at the within-group equilibrium in the appropriate parameter regime. For these Lemmas, we can characterize the dynamics of dynamics below $x^{eq}$ in terms of the \holder exponent of the initial measure near $x=0$, which is given by 
\begin{equation} \label{eq:holderzero} \zeta = \ds\inf_{\Theta \geq 0} \left\{ \ds\lim_{y \to 0} \frac{\mu_0([0,y])}{y^{\Theta}} > 0 \right\}  \end{equation}
We will assume for convenience that our initial measures $\mu_0(dx)$ will have \holder exponent $\zeta > 0$, although the results should hold as well in the case in which some groups are concentrated at the all-defector group in the initial population.

\begin{lemma} \label{lem:HDleftprob} Suppose $\mu_0(dx)$ has \holder exponent $\zeta > 0$ near the all-defector group $x=0$. Then, for any $\delta > 0$, $\mu_t\left([0,\tfrac{\beta}{|\alpha|} -\delta] \right) \to 0$ as $t \to \infty$.  \end{lemma}

\begin{proof} 
Using the measure-valued formulation, we have for test function $\psi(x)$ that
\[ \int_0^{\tfrac{\beta}{|\alpha|} - \delta} \psi(x) \mu_t(dx) = \int_0^{\phi_t^{-1}\left(\tfrac{\beta}{|\alpha|} - \delta \right)} \psi\left(\phi_t(x) \right) w_t\left( \phi_t(x)\right) \mu_0(dx) \leq  ||\psi||_{\infty} \int_0^{\phi_t^{-1}\left(\tfrac{\beta}{|\alpha|} - \delta \right)}  w_t\left( \phi_t(x)\right) \mu_0(dx) \]
Using the characterization of $w_t(\phi_t(x))$ from equation \ref{eq:wtx}, we have that
\begin{align*} \int_0^{\phi_t^{-1}\left(\tfrac{\beta}{|\alpha|} - \delta\right)} \psi(x) \mu_t(dx)  &\leq \int_0^{\phi_t^{-1}\left(\tfrac{\beta}{|\alpha|} - \delta\right)} \exp\left(\lambda  \int_0^t  \left[ G(\phi_s(x)) - \langle G(\cdot) \rangle_{\mu_s} \right] ds\right) \mu_0(dx)  \\ &\leq  \int_0^{\phi_t^{-1}\left(\tfrac{\beta}{|\alpha|} - \delta\right)} \exp\left(\lambda  \int_0^t  G(\phi_s(x))  ds\right) \mu_0(dx)  \\ &\leq  \int_0^{\phi_t^{-1}\left(\tfrac{\beta}{|\alpha|} - \delta\right)} \exp\left(\lambda  \int_0^t  G(\Pi_s(1;x))  ds\right) \mu_0(dx) \end{align*}
The second inequality holds above because $\langle G (\cdot) \rangle_{\mu_s} = \gamma M_1^{\mu} - |\alpha| M_2^{\mu} \geq |\alpha| \left( M_1^{\mu} - M_2^{\mu} \right) \geq 0$ (as $\gamma > |\alpha|$ for the HD game and $x \geq x^2$ for $x \in [0,1]$). The third inequality holds because $G(\phi_t(x)) \leq G\left(\Pi_t(1:x)\right)$ (as $G(x)$ is increasing for $x \in [0,\beta/|\alpha|]$ and $\phi_t(x) \leq \Pi_t(1;x)$ from Equation \ref{eq:HDPirankings}). %
Because we see from Equations \ref{eq:Piintegral} and \ref{eq:Pisquaredintegral} that $\int_0^t G(\Pi_t(1;x)) ds = \gamma \int_0^t \Pi_s(1;x) ds - |\alpha| \int_0^t \Pi_s(1;x)^2 ds$ is bounded in $t$, we can deduce that $\exists M > 0$ such that
\[ \int_0^{\tfrac{\beta}{|\alpha|} - \delta} \psi(x) \mu_t(dx)   \leq M ||\psi||_{\infty} \int_0^{\phi_t^{-1}\left(\tfrac{\beta}{|\alpha|} - \delta\right)} \mu_0(dx) \leq M  ||\psi||_{\infty} \mu_0\left(\left[0,\phi_t^{-1}\left(\frac{\beta}{|\alpha|} - \delta \right) \right] \right)  %
   \]
Then we can use the ranking of backward characteristics from Equation \ref{eq:HDPirankingsbackwards} to see that 
\[  \int_0^{\tfrac{\beta}{|\alpha|} - \delta} \psi(x) \mu_t(dx)   \leq  M  ||\psi||_{\infty} \mu_0\left( \left[0, \Pi_t^{-1}\left(1 - \frac{\beta}{|\alpha|};x \right) \right] \right)\]
After a large time $t$, we see from Equation \ref{eq:Piinvt} and our assumption about the initial distribution near $x=0$ that $\exists C > 0$ such that
\begin{align*} \mu_t \left([0,\Pi_t^{-1}\left(1 - \tfrac{\beta}{|\alpha|};x \right) \right) &\approx C \left(  \frac{\beta x}{|\alpha| x + \left( \beta - |\alpha| x \right) e^{\frac{\beta}{|\alpha|} \left(|\alpha| - \beta\right) t}} \right)^{ \zeta} \\ &= C\left( \frac{\beta x}{|\alpha| x e^{-\frac{\beta}{|\alpha|} \left(|\alpha| - \beta\right) t} + \beta - |\alpha| x } \right)^{\zeta}e^{-\frac{\beta}{|\alpha|} \left(|\alpha| - \beta\right) t}  \end{align*}
Because $x \in [0, \frac{\beta}{|\alpha|} - \delta)$ and $\beta, |\alpha|x e^{-\frac{\beta}{|\alpha|} \left(|\alpha| - \beta\right) t} > 0$, we further see that 
\[ \mu_t \left([0,\Pi_t^{-1}\left(1 - \tfrac{\beta}{|\alpha|};x \right) \right)  \leq C  \left( \frac{\beta}{\beta - |\alpha| x} \right)^{\zeta}  e^{-\frac{\beta}{|\alpha|} \left(|\alpha| - \beta\right) t} \leq C \left(\frac{\beta}{|\alpha| \delta} \right)^{\zeta} e^{-\frac{\beta}{|\alpha|} \left(|\alpha| - \beta\right) t} \]
and, after choosing the test function $\psi(x) \equiv 1$, we can use the fact that $\beta < |\alpha|$ for the HD game to deduce that 
\[ \mu_t\left(\left[0,\tfrac{\beta}{|\alpha|} - \delta\right] \right) = \int_0^{\tfrac{\beta}{|\alpha|} - \delta} \mu_t(dx) \leq \left[\frac{C M \beta^{\zeta}}{\left(|\alpha| \delta\right)^{\zeta}}\right] %
e^{-\frac{\beta}{|\alpha|} \left(|\alpha| - \beta\right) t}\to 0 \: \: \mathrm{as} \: \: t \to \infty \] \end{proof}

\begin{lemma} \label{lem:HDleftint} Suppose $\mu_0(dx)$ has \holder exponent $\zeta > 0$ near the all-defector group $x=0$. For any $a \in \left[0,\frac{\beta}{|\alpha|}\right)$, there exists an $A_a >  -\infty$ such that $\int_0^{\infty} j_a(s) ds > A_a$. 
  \end{lemma}
  
  \begin{proof}  We can decompose the integrand $j_s(a)$ as 
  \begin{align*} j_a(s) = \int_0^a \left[ \gamma \left( x - a\right) - |\alpha| \left( x^2 - a^2 \right) \right] \mu_t(dx) + \int_a^1 \left[ \gamma \left( x - a\right) - |\alpha| \left( x^2 - a^2 \right) \right] \mu_t(dx) := j_a^1(s) + j_a^2(s)  \end{align*} 
  Following the proof of Lemma 4.2 in \cite{cooney2019replicator}, we are able to show that $j_a^2(s) \geq 0$ using the change of variable $z = \frac{x-a}{1-a}$ mapping $[a,1] \to [0,1]$ and the measure $\nu_t(dz) := \frac{1-a}{p_a} \mu_t(dz)$ with $p_a = \mu_t\left([a,1] \right)$. As calculated in \cite{cooney2019replicator}, this approach allows us express the integrals in $j_a^2(s)$ in terms of the moments of $\nu_t(dx)$, denoted by $M_j^{\nu} := \int_0^1 z^j \nu_t(dz)$, yielding 
  \begin{subequations} \begin{align*} \int_a^1 \left(x - a \right) \mu_t(dx) &=  p_a \left(1-a\right) M_1^{\nu} \\ \int_a^1 \left(x^2 - a^2 \right) \mu_t(dx) &=  p_a \left(1 - a \right)\left[ 2 a M_1^{\nu} + \left( 1 - a \right) M_2^{\nu} \right] \end{align*} \end{subequations} 
  Using these expressions, we can compute $j_a^2(s)$ as 
  \begin{align*} j_a^2(s) &= \gamma \int_a^1 \left(x - a \right) \mu_t(dx) - |\alpha| \int_a^1 \left( x^2 - a^2 \right) \mu_t(dx) \\ &= p_a \left( 1 - a \right) \left[ \left( \gamma - 2 |\alpha| a \right) M_1^{\nu} - |\alpha| \left( 1 - a \right) M_2^{\nu} \right] \\ &= p_a \left( 1 - a \right) \left\{ \left[\gamma - |\alpha| \left(1+ a\right) \right] M_1^{\nu} + |\alpha| \left(1-a\right) \left(M_1^{\nu} - M_2^{\nu}\right) \right\} \\ &\geq p_a \left( 1 - a \right)  \left[\gamma - |\alpha| \left(1+ a\right) \right] M_1^{\nu}   \end{align*}
  where we noted that $M_1^{\nu} \geq M_2^{\nu}$ because $\supp(\nu_t(dz)) \subseteq [0,1]$. Because $a \leq \frac{\beta}{|\alpha|}$ by assumption, we further see that \[j_a^2(s) \geq p_a \left(1-a\right) \left[ \gamma - \left( |\alpha| + \beta\right)\right] = p_a \left(1 - a \right) \left( R - S \right) \geq 0 \]
 Now turning our attention to $j_a^1(s)$, we can adapt the approach from \cite{cooney2019replicator} to get a lower bound for $j_a^1(s)$ and $\int_0^t j_a^1(s) ds$. We can observe that 
 \begin{align*} j_a^1(s) &= \gamma \int_0^a \left( x - a \right) \mu_t(dx) - |\alpha| \int_0^a \underbrace{\left(x^2 - a^2\right)}_{\leq 0}  \mu_t(dx) \geq - \gamma a \int_0^1 \mu_t(dx) = - \gamma a \mu_t\left([0,a] \right) \end{align*}
 Because $a < \frac{\beta}{|\alpha|}$, there is a $\delta_a > 0$ such that $a = \frac{\beta}{|\alpha|} - \delta_a$. This allows us to use the proof of Lemma \ref{lem:HDleftprob} to see that $\exists C, M < \infty$ such that $j_a^1(s)$ satisfies the bound
 \[j_a^1(s) \geq - \gamma a \mu_t\left(\left[0,\tfrac{\beta}{|\alpha|} - \delta_a \right] \right) \geq - \frac{CM\gamma a \beta^{\zeta}}{|\alpha|^{\zeta} \delta_a^{\zeta}} e^{- \tfrac{\beta}{|\alpha|} \left(|\alpha| - \beta\right) t} \]
 Using this bound and that $j_a^2(s)$, we can bound the integral of $j_a(s)$ as
 \begin{align*} \int_0^t j_a(s) ds = \int_0^t \left(j_a^1(s) + j_a^2(s) \right) ds & \geq - \int_0^t  \frac{CM\gamma a \beta^{\zeta}}{|\alpha|^{\zeta} \delta_a^{\zeta}} e^{- \tfrac{\beta}{|\alpha|} \left(|\alpha| - \beta\right) s} ds \\ &= \frac{CM\gamma a \beta^{\zeta-1}}{|\alpha|^{\zeta-1} \delta_a^{\zeta} \left(|\alpha| - \beta\right)} \left(e^{- \tfrac{\beta}{|\alpha|} \left(|\alpha| - \beta\right) t} - 1   \right)  \end{align*}
and then we see that in the long-time limit that \[\int_0^{\infty} j_a(s) ds \geq - \frac{CM\gamma a \beta^{\zeta-1}}{|\alpha|^{\zeta-1} \delta_a^{\zeta} \left(|\alpha| - \beta\right)} \geq -\infty \] \end{proof}

Now we are ready to characterize the conditions under which the population converges to a delta-function at the within-group equilibrium of the HD game. In Proposition \ref{prop:HDdelta}, we use Lemmas \ref{lem:HDleftprob} and \ref{lem:HDleftint} and the comparison principles for HD characteristic curves to show that there exists a critical level of between-group selection $\lambda^*$ such that the population converges to $\delta\left(x - \frac{\beta}{|\alpha|}\right)$ as $t \to \infty$ when $\lambda < \lambda^*$.
\begin{proposition} \label{prop:HDdelta} Suppose $\mu_0(dx)$ has \holder exponents of $\zeta$ near $x = 0$ and $\theta$ near $x=1$. If $\lambda \left(\gamma - \left(\beta +  |\alpha| \right) \right) <  |\alpha| \theta$, then $\mu_t(dx) \rightharpoonup \delta\left(x - \tfrac{\beta}{|\alpha|}\right)$ as $t \to \infty$.   \end{proposition}

\begin{proof} Again, we consider a continuous test function $\psi(x)$ and will show that $\int_0^1 \psi(x) \mu_t(dx) \to \int_0^1 \psi(x) \delta\left(x - \tfrac{\beta}{|\alpha|}\right) dx = \psi\left(\tfrac{\beta}{|\alpha|}\right)$. Using the continuity of $\psi(x)$, we know that $\forall \epsilon> 0$, $\exists \delta$ such that $|\psi(x) - \psi\left( \frac{\beta}{|\alpha|}\right)| < \epsilon$ when $x \in \left[\tfrac{\beta}{|\alpha|} - \delta, \tfrac{\beta}{|\alpha|} + \delta \right]$. Because $\mu_t(dx)$ is a probability distribution, we see that
\begin{align*} \bigg| \int_0^1 \psi(x) \mu_t(dx) - \psi\left(\tfrac{\beta}{|\alpha|}\right) \bigg| 
& \leq \int_0^{\tfrac{\beta}{|\alpha|}-\delta} \bigg| \psi(x) - \psi\left(\tfrac{\beta}{|\alpha|} \right) \bigg| \mu_t(dx)  +  \int_{\tfrac{\beta}{|\alpha|} + \delta} ^{1} \bigg| \psi(x) - \psi\left(\tfrac{\beta}{|\alpha|} \right) \bigg| \mu_t(dx)  \\ &+ \int_{\tfrac{\beta}{|\alpha|} - \delta}^{\tfrac{\beta}{|\alpha|} + \delta}  \bigg| \psi(x) - \psi\left(\tfrac{\beta}{|\alpha|} \right) \bigg| \mu_t(dx)  \\ & \leq \epsilon  + 2||\psi||_{\infty}  \int_0^{\tfrac{\beta}{|\alpha|}-\delta} \mu_t(dx)  +  2||\psi||_{\infty}   \int_{\phi_t^{-1}\left(\tfrac{\beta}{|\alpha|} + \delta\right)} ^{1} w_t(\phi_t(x)) \mu_0(dx) \\ & \leq \epsilon + 2 ||\psi||_{\infty} \mu_t\left(\left[0,\tfrac{\beta}{|\alpha|} - \delta\right]\right)  +  2||\psi||_{\infty}   \int_{\phi_t^{-1}\left(\tfrac{\beta}{|\alpha|} + \delta\right)} ^{1} w_t(\phi_t(x)) \mu_0(dx)
\end{align*}
From Lemma \ref{lem:HDleftprob}, we know that for $\delta > 0$ that $ \mu_t\left(\left[0,\tfrac{\beta}{|\alpha|} - \delta\right]\right)   \to 0$ as $t \to \infty$, so there exists a time $T_{1}$ such that $\mu_t\left(\left[0,\tfrac{\beta}{|\alpha|} - \delta\right]\right) \leq \frac{\epsilon}{2 ||\phi||_{\infty}}$ for $t > T_1$. For such $t$, we have that 
\[ \bigg| \int_0^1 \psi(x) \mu_t(dx) - \psi\left(\tfrac{\beta}{|\alpha|}\right) \bigg| \leq 2 \epsilon + 2||\psi||_{\infty}   \int_{\phi_t^{-1}\left(\tfrac{\beta}{|\alpha|} + \delta\right)} ^{1} w_t(\phi_t(x)) \mu_0(dx) \]
Using Equation \ref{eq:wtphixHDa}, we can see for $a \in [0,\tfrac{\beta}{|\alpha|})$ that 
\[w_t(\phi_t(x)) = \exp\left(\lambda\left[ \int_0^t \left( \gamma \phi_s(x) - |\alpha| \psi_s(x)^2 \right) ds - \left( \gamma a - |\alpha|a^2\right) t - \int_0^t j_s(a) ds  \right] \right) \]
From Lemma \ref{lem:HDleftint}, we know that there exists an $A(a) > -\infty$ such that $\int_0^t j_s(a) ds \geq A(a)$ for all $t$, and therefore we know that  
\[w_t(\phi_t(x)) \leq e^{- \lambda \left(\gamma a - |\alpha| a^2 \right) t-\lambda A(a)}\exp\left(\lambda \int_0^t \left( \gamma \phi_s(x) - |\alpha| \psi_s(x)^2 \right) ds \right) \]
Using the ranking of characteristic curves from Equation \ref{eq:HDcharranking}, we know that $\phi_s(x) \leq \Xi_s\left(\tfrac{\beta}{|\alpha|} ;x\right)$ and that $-\phi_s(x)^2 \leq - \Xi_s(1;x)^2$. This allows us to say that 
\[w_t(\phi_t(x)) \leq e^{- \lambda \left(\gamma a - |\alpha| a^2 \right) t-\lambda A_a}\exp\left(\lambda\left[ \int_0^t \left( \gamma \Xi_s(\tfrac{\beta}{|\alpha|};x) - |\alpha| \Xi_s(1;x)^2 \right) ds   \right] \right) \]
  Using the integrals along simplified characteristic curves $\Xi_t(k;x_0)$ from Equations \ref{eq:Xiintegral} and \ref{eq:Xisquaredintegral}, there is a bound $M_1(\delta)$ such that 
  \[w_t(\phi_t(x)) \leq   e^{- \lambda \left(\gamma a - |\alpha| a^2 \right) t-\lambda A_a} \left(M_1(\delta) e^{\lambda \left(\gamma - |\alpha| \right) t} \right) =  M_1(\delta) e^{\lambda \left(1 - a \right) \left(\gamma  - \left( 1 + a \right) |\alpha|  \right) t} e^{-\lambda A(a)} \]
  The exponential rate \begin{equation} \label{eq:Hofa} H(a) := \left(1 - a \right) \left( \gamma - \left(1 + a \right) |\alpha| \right) \end{equation} is a continuous and decreasing function of $a$, which ranges between \begin{equation} \label{eq:Hinequality}  \left(\frac{|\alpha| - \beta}{|\alpha|} \right) \left( \gamma - \left(|\alpha| + \beta\right)\right) = H\left(\frac{\beta}{|\alpha|}\right)  \leq H(a) \leq H(0) = \gamma - |\alpha| \end{equation}  for $a \in [0,\tfrac{\beta}{|\alpha|}]$. 
We can rewrite the hypothesis $\lambda  \left( \gamma - \left(|\alpha| + \beta\right)\right) <  |\alpha| \theta$ in terms of our use our expression for $H\left(\tfrac{\beta}{|\alpha|}\right)$. Because $|\alpha| > \beta > 0$ for the HD game, we can see that 
\begin{equation}  \label{eq:HDcondition} \lambda  \left( \gamma - \left(|\alpha| + \beta\right)\right) <  |\alpha| \theta \Longleftrightarrow \lambda H\left( \frac{\beta}{|\alpha|} \right) :=  \left(\frac{|\alpha| - \beta}{|\alpha|} \right) \left( \gamma - \left(|\alpha| + \beta\right)\right) < \left(|\alpha| - \beta\right) \theta  \end{equation}
so we will use the form $\lambda H(\left(\tfrac{\beta}{|\alpha}\right) < \left( |\alpha| - \beta \right) \theta$ in order to best make use of the estimate
\[w_t(\phi_t(x)) \leq  M_1(\delta) e^{\lambda H(a) t} e^{-\lambda A(a)}.  \]
If $\lambda H\left(\tfrac{\beta}{|\alpha|}\right) < \left( |\alpha| - \beta \right) \theta$, we can use the fact that $H(a)$ is continuous and decreasing to deduce that either there exists an $a^* \in [0,\tfrac{\beta}{|\alpha|}]$ such that $\lambda H(a^*) =   \left( |\alpha| - \beta \right) \theta$ or that $\lambda H(0) < \left(  |\alpha| - \beta \right)  \theta$. In either case, this tells us that we can choose $a^{**} = \tfrac{1}{2} \left(\max(a^*,0)+\tfrac{\beta}{|\alpha|} \right)$ which satisfies $\lambda H(a^{**}) <   \left(|\alpha| - \beta \right) \theta$. We will stick this choice of $a^{**}$ for the rest of the proof, which means that we can write
\[w_t(\phi_t(x)) \leq M_1(\delta) e^{\lambda H(a^{**})t} e^{-\lambda A(a)}  \]
 Now that we have bounded $w_t(\phi_t(x))$, we can say for $t > T_1$ that
 \begin{equation} \label{eq:HD2eps} \bigg| \int_0^1 \psi(x) \mu_t(dx) - \psi\left(\tfrac{\beta}{|\alpha|}\right) \bigg| \leq 2 \epsilon + 2||\psi||_{\infty} M_1(\delta) e^{\lambda H(a^{**})t} e^{-\lambda A(a)} \mu_0\left( \left[\phi_t^{-1}\left( \tfrac{\beta}{|\alpha|} + \delta \right),1 \right] \right)  \end{equation}
To consider the effect of $ \mu_0\left( \left[\phi_t^{-1}\left( \tfrac{\beta}{|\alpha|} + \delta \right),1 \right] \right)$ on the long-time behavior of $\mu_t(dx)$, we consider separately the cases in which $\lambda H(a^{**}) < \frac{\beta}{|\alpha|} \left( |\alpha| - \beta\right) \theta < \left( |\alpha| - \beta\right) \theta$ and in which $\frac{\beta}{|\alpha|} \left( |\alpha| - \beta\right) \theta < \lambda H(a^{**}) < \left( |\alpha| - \beta \right) \theta$. When $\lambda H(a^{**}) < \tfrac{\beta}{| \alpha|} \left( |\alpha| - \beta \right) \theta$, we use the ranking of backward characteristics from Equation \ref{eq:HDcharbackwardranking} to recall that $ \Xi_t^{-1}\left(\tfrac{\beta}{|\alpha|};x\right) \leq \phi_t^{-1}(x)$ for $x \in [\tfrac{\beta}{\alpha},1]$, yielding 
\begin{align*} \mu_0\left(\left[ \phi_t^{-1}(x),1 \right] \right) \leq \mu_0\left(\left[\Xi_t^{-1} \left(\tfrac{\beta}{|\alpha|};x\right),1\right] \right) &= \mu_0\left(\left[\frac{\left(1 - x\right) \beta + \left(|\alpha| x - \beta \right) e^{-\frac{\beta}{|\alpha|}\left( \beta - |\alpha| \right)t}}{\left(1-x \right) |\alpha| + \left( |\alpha| x - \beta \right) e^{-\frac{\beta}{|\alpha|}\left( \beta - |\alpha| \right) t} },1 \right]\right)  \\ &= \mu_0\left( \left[1 - \frac{\left(1-x \right)\left(|\alpha| - \beta\right)  }{\left(1-x \right) |\alpha| + \left( |\alpha| x - \beta \right)  e^{-\frac{\beta}{|\alpha|}\left( \beta - |\alpha| \right) t}} ,1\right] \right) \\ &\leq \mu_0 \left( \left[ 1 - \frac{ \left(|\alpha| - \beta\right)}{|\alpha| \delta} \, e^{-\frac{\beta}{|\alpha|}\left( |\alpha| - \beta \right) t} , 1\right]\right)  \end{align*} 
Denoting $K := \frac{\left(1 - \delta\right) \left(|\alpha| - \beta\right)}{|\alpha| \delta}$, we can use our assumption about the \holder exponent of $\mu_0(dx)$ near $x=1$ to show, for sufficiently large time $t$, that
\begin{equation} \label{eq:HDsimpleineq} \mu_0\left(\left[ \phi_t^{-1}(x),1 \right] \right) \leq \mu_0( [1 - K e^{-\frac{\beta}{|\alpha|}\left( |\alpha| - \beta \right) t} , 1  ] ) \approx C K^{\theta}  e^{-\tfrac{\beta}{|\alpha|}\left( |\alpha| - \beta \right) \theta t}    \end{equation}
Because this holds for all $x \in [\tfrac{\beta}{|\alpha},1]$, in the case that $\lambda H(a^{**}) < \frac{\beta}{|\alpha|} \left( |\alpha| - \beta\right) \theta$, we see that 
\begin{dmath} \label{eq:HDsimpleconverge} 2||\psi||_{\infty} M_1(\delta) e^{\lambda H(a^{**})t} e^{-\lambda A(a)} \mu_0\left( \left[\phi_t^{-1}\left( \tfrac{\beta}{|\alpha|} + \delta \right),1 \right] \right) \\ \leq 2 ||\psi||_{\infty} C M_1(\delta)   K^{\theta}  e^{-\lambda A(a)}    \exp\left(\left[\lambda H(a^{**}) -\tfrac{\beta}{|\alpha|}\left( |\alpha| - \beta \right) \theta \right] t\right) \condition{as $t \to \infty$.}    \end{dmath}
For the alternate case in which $\frac{\beta}{|\alpha|} \left( |\alpha| - \beta\right) \theta < \lambda H(a^{**}) < \left( |\alpha| - \beta\right) \theta$, we must use a more refined comparison principle. From this condition, we know that there exists $k \in (\tfrac{\beta}{|\alpha|},1)$ such that $\lambda H(a^{**}) = k \left( |\alpha| - \beta \right) \theta$. Further, this means that choosing $\Gamma := \frac{k+1}{2}$ yields $\lambda H(a^{**}) < \Gamma \left(|\alpha| - \beta\right) \theta$. Because $\phi_t^{-1}\left(\tfrac{\beta}{|\alpha|} + \delta\right)$ travels from $\tfrac{\beta}{|\alpha|} + \delta$ to $1$ as $t$ increases, we know that there exists a time $\tGamma$ such that $\phi^{-1}_{\tGamma}(\tfrac{\beta}{|\alpha|} +\delta) = \Gamma$ and $\phi_t^{-1}(\tfrac{\beta}{|\alpha|} +\delta) > \Gamma$ for $t >  \tGamma$. Further, we see, for $x \geq \Gamma$, that 
\[ \Gamma \left( 1 -x \right) \left(|\alpha| x - \beta \right) \leq x \left(1-x\right)\left( |\alpha| x - \beta \right) \: \: \mathrm{and} \: \: \dsddx{}{t} \Xi_t^{-1}\left(\Gamma;x\right) \leq \dsddx{}{t} \phi_t^{-1}(x). \]
For $t > T^{\Gamma}_{\delta}$ and $x \geq \tfrac{\beta}{|\alpha|} + \delta$, this tells us that the characteristic curves travel faster than solutions to $\Xi^{-1}_t\left(\Gamma;x\right)$, so we know that $\Xi_t^{-1}\left(\Gamma;x\right) \leq \phi_t^{-1}(x)$. We can track solutions of our slower curves $\ddx{}{t} \Xi_t^{-1}\left(\Gamma;\cdot \right)$ located at $\Gamma$ at time $\tGamma$ (coinciding with the backwards characterstic curve $\phi_t^{-1}(\delta)$ at time $\tGamma$), seeing from Equation \ref{eq:Xiinvt} that its subsequent tracjectory for $t > \tGamma$ is given by
\begin{align*} \Xi_{t-\tGamma}^{-1}\left(\Gamma;\Gamma\right) = \frac{\left(1 - \Gamma\right) \, \beta \: + \left(|\alpha| \Gamma - \beta \right) e^{\Gamma \left(|\alpha| - \beta\right)\left(t - \tGamma\right)}}{\left(1 - \Gamma\right) |\alpha| + \left(|\alpha| \Gamma - \beta \right) e^{\Gamma \left(|\alpha| - \beta \right) \left(t-\tGamma\right)}} &= 1 - \frac{\left(1 - \Gamma\right) \left(|\alpha| - \beta\right)}{\left(1 - \Gamma\right) |\alpha| + \left(|\alpha| \Gamma - \beta \right) e^{\Gamma \left(|\alpha| - \beta \right) \left(t-\tGamma\right)}} \\ &\leq 1 - \left(\frac{(1 - \Gamma) \left(|\alpha| - \beta\right) e^{\Gamma \left( |\alpha| - \beta\right) \tGamma}}{|\alpha| \Gamma - \beta} \right) e^{-\Gamma \left( |\alpha| - \beta\right) t}   \end{align*}
Denoting $K_{\Gamma} := \frac{(1 - \Gamma) \left(|\alpha| - \beta\right) e^{\Gamma \left( |\alpha| - \beta\right) \tGamma}}{|\alpha| \Gamma - \beta}$, and using our assumption about the \holder exponent of $\mu_0(dx)$ near $x=1$, we see, for sufficiently large $t$ that
\[\mu_0\left( \left[\phi_t^{-1}(x),1 \right] \right) \leq \mu_0\left( \left[ \Xi_{t-\tGamma}^{-1}\left(\Gamma;\Gamma\right),1 \right]\right) \leq \mu_0\left( \left[1 - K_{\Gamma}e^{-\Gamma \left( |\alpha| - \beta\right) t} , 1  \right] \right) \approx C K_{\Gamma}^{\theta} e^{-\Gamma \left( |\alpha| - \beta\right) t} \]
Because this holds for $x \in (\tfrac{\beta}{|\alpha|} +\delta, 1]$ and we chose $\Gamma$ such that $\lambda H(a^{**}) < \Gamma \left(|\alpha| - \beta\right) \theta$, we can conclude that 
\begin{dmath} \label{eq:HDcarefulconverge} 2||\psi||_{\infty} M_1(\delta) e^{\lambda H(a^{**})} e^{-\lambda A(a)} \mu_0\left( \left[\phi_t^{-1}\left( \tfrac{\beta}{|\alpha|} + \delta \right),1 \right] \right) \\ \leq 2 ||\psi||_{\infty} C M_1(\delta)   K_{\Gamma}^{\theta}  e^{-\lambda A(a)}    \exp\left(\left[\lambda H(a^{**}) -\Gamma \left( |\alpha| - \beta \right) \theta \right] t\right) \condition{as $t \to \infty$.}    \end{dmath}
From Equations \ref{eq:HDsimpleconverge} and \ref{eq:HDcarefulconverge}, we can see that $\mu_t\left(\left[\tfrac{\beta}{|\alpha|} + \delta,1 \right] \right) \to 0$ as $t \to \infty$ whenever $\lambda H(a^{**}) < \left(|\alpha| - \beta\right) \theta$. %
Recalling that we chose $a^{**}$ in alignment with the equivalent conditions of Equation \ref{eq:HDcondition}, we can see that when the inequality $\lambda H(\tfrac{\beta}{\alpha}) < \left(|\alpha| - \beta\right) \theta$ holds, there exists a time $T_2$ such that for $t > T_2$, $\mu_t\left(\left[\tfrac{\beta}{|\alpha|} + \delta,1 \right] \right) < \epsilon$. Then, returning to Equation \ref{eq:HD2eps}, we see for $t > \max\left(T_1,T_2 \right)$, that 
\[ \bigg| \int_0^1 \psi(x) \mu_t(dx) - \psi\left(\tfrac{\beta}{|\alpha|}\right) \bigg| \leq 3 \epsilon \]
confirming that $\int_0^1 \psi(x) \mu_t(dx) \to \int_0^1 \psi(x) \delta\left(x - \tfrac{\beta}{|\alpha|}\right) dx$, and therefore $\mu_t(dx) \rightharpoonup \delta\left(x - \tfrac{\beta}{|\alpha|}\right)$ when the condition of Equation \ref{eq:HDcondition} holds.
  \end{proof}

\subsection{Steady States of General Multilevel HD}\label{sec:HDsteady}

We look for density steady state solutions $f(x)$ for the multilevel HD dynamics, which satisfy
\begin{equation} \label{eq:HDsteady}0 = \dsdel{}{x} \left(x(1-x)(\beta - |\alpha| x) f(x) \right) +  \lambda f(x) \left[ \gamma x - |\alpha| x^2 - \left( \gamma M_1^f  - |\alpha| M_2^f \right) \right]  \end{equation}
 From Lemma \ref{lem:HDleftprob}, we know that the probability of having groups below the within-group equilibrium $x^{eq} = \frac{\beta}{|\alpha|}$ vanishes as $t \to \infty$. Therefore we explore steady state solutions with zero density below the within-group equilibrium, which take the piecewise form
 $$f(x) = \left\{
     \begin{array}{cr}
       0 & :  0 \leq x \leq \frac{\beta}{|\alpha|} \\
       p(x) & : \frac{\beta}{|\alpha|} < x \leq 1
     \end{array}
   \right.$$
Because the moments $M_j^f$ depend on $f(x)$, we look for implicit solutions of Equation \ref{eq:HDsteady} taking the moments as given. An implicit expression for $p(x)$ can be found using separation of variables and integration by partial fractions, yielding 
\begin{dmath} \label{eq:HDpofxmoment} p(x) = Z_f^{-1} x^{\beta^{-1} \lambda \left( \gamma M_1^f - |\alpha| M_2^f \right) - 1}  \left(1 - x \right)^{\left(|\alpha| - \beta\right)^{-1} \left[ \lambda \left(\gamma - |\alpha| \right) - \lambda \left( \gamma M_1^f - |\alpha| M_2^f \right) \right] - 1 } \\ \times  \left(|\alpha| x - \beta \right)^{\left(|\alpha| - \beta\right)^{-1} \left[\lambda \left( \beta - \gamma\right) + \lambda \left(\gamma M_1^f - |\alpha| M_2^f \right) \right] - 1}  \end{dmath}
From Proposition \ref{prop:HDholder}, we know that the \holder exponent near $x=1$ is preserved in time for the multilevel HD dynamics. Therefore it is sensible to parametrize steady states by the Hölder exponent $\theta$ of $p(x)$ near $x=1$. In the same manner as was computed for the PD dynamics in section \ref{sec:PDsteady}, we find that $\theta$ is related to the exponent of the $(1-x)$ term by 
\[\theta = \frac{ \lambda (\gamma - |\alpha|)- \lambda (\gamma M_1^f - |\alpha|  M_2^f)}{|\alpha| -\beta} \] %
This equation can also be rewritten to characterize the average payoff of the population at steady state, $\gamma M_1^f - |\alpha| M_2^f$, as the following
\begin{equation} \label{eq:HDsteadyfitness} \langle G(\cdot) \rangle_{f} = \int_0^1 G(x) f(x) dx =  \gamma M_1^f - 2 M_2^f  = (\gamma - |\alpha|)  - \frac{(|\alpha|-\beta) \theta}{\lambda}. \end{equation}  
Substituting this expression into Equation \ref{eq:HDpofxmoment} allows us to obtain an explicit expression for $p(x)$ in terms of the \holder exponent $\theta$, taking the form
\begin{dmath} \label{eq:HDpofxtheta} p^{\lambda}_{\theta}(x) = Z_f^{-1} x^{\beta^{-1} \left[ \lambda \left(|\alpha| - \gamma\right) + \left(|\alpha| - \beta \right) \theta \right] - 1} \left( 1 - x \right)^{\theta - 1} \left(|\alpha| x - \beta\right)^{\beta^{-1} \left[ \lambda \left(\gamma - |\alpha| - \beta \right) - |\alpha| \theta \right] - 1} \end{dmath}
We note that there is a threshold level of between-group selection strength $\lambda^*$ such these steady states are integrable if and only if 
\begin{equation} \label{eq:lambdastar} \lambda > \lambda^* := \left( \frac{|\alpha|}{ \gamma - (\beta + |\alpha|)} \right) \theta, \end{equation} and that if $\gamma = (\beta + |\alpha|)$, it is not possible to have an integrable steady state of this form, regardless of the value of $\lambda$. We can further rearrange our expression for the threshold $\lambda^*$ to obtain
\begin{equation*}  \lambda^* =  \frac{\left(|\alpha| - \beta \right) \theta}{\left( \frac{|\alpha| - \beta}{|\alpha|} \right) \left( \gamma - \left( \beta + |\alpha| \right) \right) } =  \frac{\left(|\alpha| - \beta \right) \theta}{\left(\gamma - |\alpha| \right) - \left( \gamma \left(\tfrac{\beta}{|\alpha|}\right) - |\alpha| \left( \frac{\beta}{|\alpha|} \right)^2 \right) }   \end{equation*}
Using the expressions for $G(x)$ and $\pi_C(x) - \pi_D(x)$ from Equations \ref{eq:grouppayoffparamsimplified} and \ref{eq:picminuspid}, we can then see that the threshold corresponds to 
\begin{equation} \label{eq:lambdaHDpayoff} \lambda^* = \left( \frac{\pi_D(1) - \pi_C(1)}{G(1) - G( \tfrac{\beta}{|\alpha|})}  \right) \theta. \end{equation}
From this characterization, we see that the threshold $\lambda^*$ is increasing in the defector's payoff advantage over a cooperator in an otherwise full-cooperator group $\pi_D(1) - \pi_C(1)$, and that $\lambda^*$ is decreasing in the collective advantage of a full-cooperator group over a group at the Hawk-Dove equilibrium composition $G(\tfrac{\beta}{|\alpha|})$. As a result, we can understand the struggle to promote cooperation above the within-group equilibrium as a tug-of-war between the individual incentive to be a defector in a group with many cooperators and the group incentive to have many cooperators rather than the level of cooperation that would result from within-group dynamics alone. When the full-cooperator group confers the same average payoff as the HD-equilibrium group, then the individual incentive $\pi_D(1) - \pi_C(1)$ wins out for any strength of between-group selection, corresponding to the threshold satisfying $\lambda^* \to \infty$ as $G(1) \to G(\tfrac{\beta}{|\alpha|})$.
 We can rewrite the threshold selection strength $\lambda^*$ in terms of our original payoff matrix as 
\begin{equation} \label{eq:HDlambdastarpayoff} \lambda^* = \left(\frac{T-P}{R - S} - 1 \right) \theta  \end{equation}
This allows us to make the following observations \begin{itemize} \item $\lambda^*$ is increasing in $T - P$, the relative advantage to a defector of playing against  a cooperator rather than a defector  \item $\lambda^*$ is decreasing in $R - S$, the relative advantage to a cooperator of playing against a cooperator rather than a defector  \item  $\lambda^* \to \infty$ as $R \to S$, so achieving cooperation above the within-group equilibrium by multilevel selection becomes increasingly difficult as cooperators lose the benefit of playing with cooperators rather than defectors.  \end{itemize}

For the multilevel HD dynamics, we have shown in Section \ref{sec:HDlongtime} that a population with initial \holder exponent $\theta$ near $x=1$ will converge to a delta-concentration at the full-defector group when $\lambda (\gamma - \left(\beta +  |\alpha|\right)) <  |\alpha| \theta$. Under the alternate condition $\lambda (\gamma - \left(\beta +  |\alpha|\right)) > |\alpha| \theta$, we have shown in this section that there is a unique and integrable steady state with \holder exponent $\theta$, and we showed in Section \ref{sec:Holderpreserve} that the \holder exponent near $x=1$ is preserved in time for solutions of Equation \ref{eq:replicatormeasurepde} for the HD game. It is then natural to suspect that populations should converge to the steady state with the same \holder exponent near $x=1$ as the initial distribution, so we make the following conjecture about the long-time behavior of solutions to Equation \ref{eq:replicatormeasurepde}.
\begin{conjecture} \label{con:hddensity} Suppose we have an initial distribution $\mu_0(dx)$ has a \holder exponents $\zeta$ near $x=0$ and  $\theta$ near $x=1$. If $\lambda \left( \gamma -\left(\beta + |\alpha| \right) \right) >  |\alpha| \theta$, then \[\mu_t(dx) \rightharpoonup \mu_{\infty}(dx) = \left\{
     \begin{array}{lr}
       0 &: x < \frac{\beta}{|\alpha|} \\
    Z_f^{-1} x^{\beta^{-1} \left[ \lambda \left(|\alpha| - \gamma\right) + \left(|\alpha| - \beta \right) \theta \right] - 1} \left( 1 - x \right)^{\theta - 1} \left(|\alpha| x - \beta\right)^{\beta^{-1} \left[ \lambda \left(\gamma - |\alpha| - \beta \right) - |\alpha| \theta \right] - 1} &: x \geq \frac{\beta}{|\alpha|}   \end{array}
   \right.
    \] 
where $Z_f^{-1}$ is a normalizing constant such that $\int_0^1 \mu_{\infty}(dx) = 1$. \end{conjecture}
This long-term behavior has already been shown to hold for special families of HD games with $\alpha = -2$ and $\beta = 1$ \cite{cooney2019replicator}. The main impediment to the general proof of the conjecture is that we have not yet shown that solutions $\mu_t(dx)$ of Equation \ref{eq:replicatormeasurepde} necessarily converge to a steady state outside of the aforementioned special case with exactly solvable within-group dynamics. However, for the remainder of this section, we will study the properties of the density steady states given by Equation \ref{eq:HDpofxtheta}, knowing that these are time-independent solutions of Equation \ref{eq:replicatormeasurepde}, with the potential additional relevance that these steady states will be the long-time outcome for initial data with \holder exponent $\theta$ near $x=1$ when $\lambda \left( \gamma - \left(\beta + | \alpha| \right) \right) > |\alpha| \theta$. In Figure \ref{fig:hdsteadydensities}, we illustrate sample steady state densities for HDs in which average group payoff is maximized by full-cooperator groups (left) and by groups with an intermediate level of cooperation (right).
\begin{figure}[H]
    \centering
   \hspace{-5mm} \includegraphics[width=0.505\textwidth]{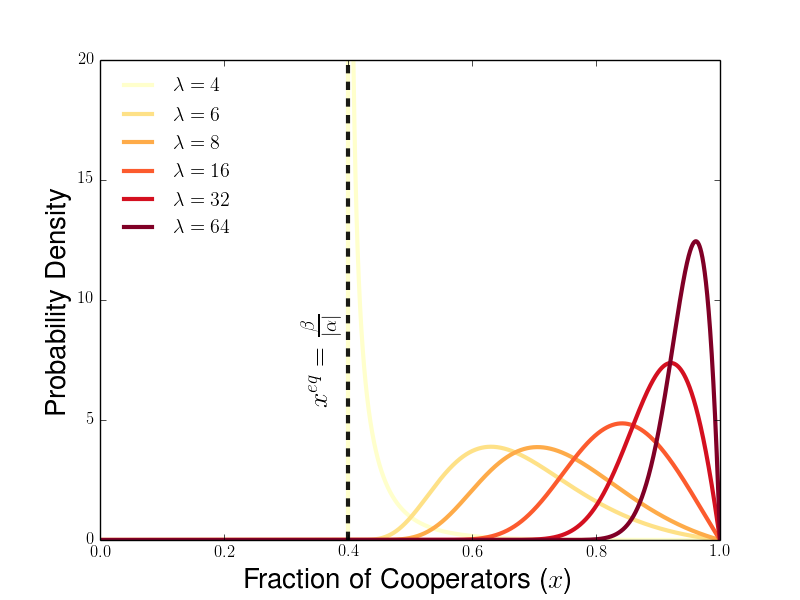} 
    \hspace{-5mm} \includegraphics[width=0.505\textwidth]{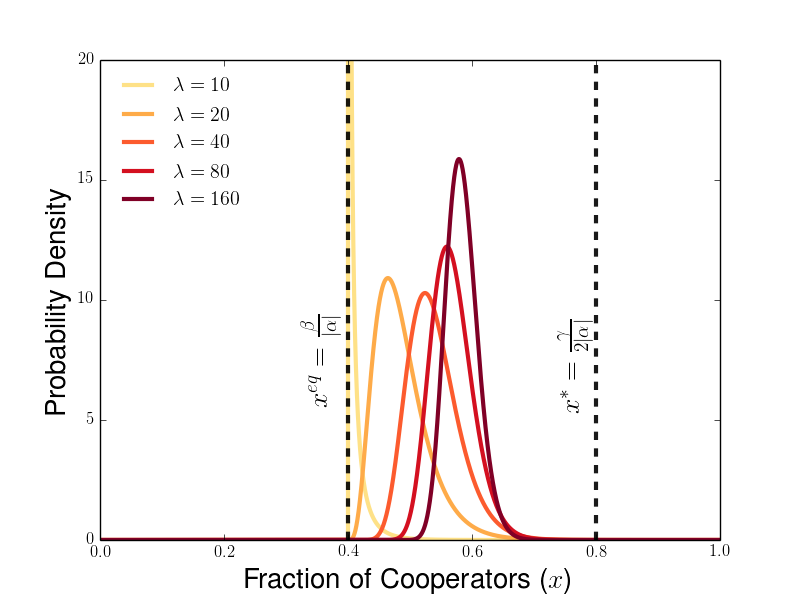}
    \caption{Steady state densities for the HD for various values of $\lambda$. Parameters shown are $\gamma = 4.5$ (Left) and $\gamma = 3.2$ (Right), with $\alpha = -2$, $\beta = \frac{4}{5}$ and $\theta = 2$ for both panels. Left dotted lines in both panels correspond to within-group HD equilibrium $x^{eq} = \tfrac{\beta}{|\alpha|}$ and dotted line on right in right panel corresponds to the group type $x^* = \frac{\gamma}{2|\alpha|}$ with maximal average payoff.}
    \label{fig:hdsteadydensities}
\end{figure}

Now we will examine peak abundance at steady state $\hat{x}_{\lambda}$, with an emphasis on the limit at $\lambda \to \infty$.  As was the case in the PD, we consider only the case of $\theta \geq 1$.  In Proposition \ref{prop:HDpeak}, as in the corresponding Proposition \ref{prop:PDpeak}, we show that if $\gamma \geq 2 |\alpha|$ (full-cooperator groups achieve maximal average payoff), the most abundant group type at steady state $\hat{x}_{\lambda}$ becomes full cooperator groups as the relative strength group selection $\lambda \to \infty$. If $\gamma < 2 |\alpha|$ (the type of group maximizing average payoff has both cooperators and defectors), we see that the most abundant group type at steady states has more defectors than the type of group maximizing average payoff, even in the limit at $\lambda \to \infty$.

\begin{proposition} \label{prop:HDpeak} Suppose $\theta \geq 1$. If  $\gamma  <  2 |\alpha|$, $\ds\lim_{\lambda \to \infty} \hat{x}_{\lambda} (f^{\lambda}_{\theta}) = \tfrac{\gamma}{|\alpha|} -  1 \in (0,1)$ and $\ds\lim_{\lambda \to \infty} \hat{x}_{\lambda} (f^{\lambda}_{\theta}) < x^*$. If  $\gamma > 2 |\alpha|$, $ \ds\lim_{\lambda \to \infty} \hat{x}_{\lambda} (f^{\lambda}_{\theta} ) = 1$.   \end{proposition} 
\begin{proof}
 To determine the fraction of cooperators $x$ at which the density $p_{\theta}(x)$ is maximized, we compute $$p'_{\theta}(x) = Z_f^{-1} g(x) x^{\beta^{-1} \left((2 - \gamma) \lambda + (|\alpha| - \beta) \theta \right) - 2} \left(1-x\right)^{\theta -|\alpha|} \left( |\alpha| x - \beta \right)^{\beta^{-1} \left(\lambda (\gamma - \beta - |\alpha|) - |\alpha| \theta \right) - 2} $$ where $$g(x) = (\gamma -|\alpha|) \lambda - (|\alpha|-\beta)\theta + \beta - \left[\lambda \gamma + 2 (\beta + |\alpha|) \right] x+ |\alpha| \left( \lambda + 3 \right) x^2$$ 
 If $\lambda \left[\gamma - \left(\beta + |\alpha| \right) \right] > |\alpha| \theta$, we have from Equation \ref{eq:HDpofxtheta} that $p^{\lambda}_{\theta} (\tfrac{\beta}{|\alpha|}) = p^{\lambda}_{\theta}(1) = 0$ when $\theta \geq 1$ (when $\theta < 1$, the steady state density blows up near $x=1$). Because $\tfrac{\beta}{|\alpha|}$ and $1$ are the only possible critical points of $f^{\lambda}_{\theta}(x)$ other than the roots of $g(x)$, %
we see that $f^{\lambda}_{\theta}(x)$ is maximized at a root of $g(x)$. These roots are given by $$x^{\lambda}_{\pm} = \frac{\lambda \gamma + 2 (\beta + |\alpha|)}{2 |\alpha| \left(\lambda + 3\right)} \pm \sqrt{\frac{\left( \lambda \gamma + 2 (\beta + |\alpha|)\right)^2}{4 |\alpha|^2 ( \lambda + 3)^2} - \frac{(\gamma -|\alpha|) \lambda - (|\alpha|-\beta)\theta + \beta}{|\alpha| ( \lambda + 3)}}$$
For large $\lambda$, we see that $$x^{\infty}_{\pm} := \ds\lim_{\lambda \to \infty} x^{\lambda}_{\pm} = \frac{\gamma}{2|\alpha|} \pm \sqrt{\frac{\gamma^2}{4 |\alpha|^2} - \frac{\left( \gamma - |\alpha| \right) }{|\alpha|}}  = \frac{\gamma}{2 |\alpha|} \pm \sqrt{\left( \frac{\gamma - 2 |\alpha|}{2 |\alpha|} \right)^2}$$ We now break our analysis into two cases: 
\begin{itemize}
\item If $\gamma \geq 2 |\alpha|$ (and the full-cooperator group maximizes average payoff), then $x^{\infty}_{\pm} = \frac{\gamma}{2 |\alpha|} \pm \left( \frac{\gamma}{2 |\alpha|} - 1\right)$, so $x^{\infty}_+ = \frac{\gamma}{|\alpha|} - 1 > 1$ and $x^{\infty}_{-} = 1$, so the most abundant group type in the steady state $f_{\theta}(x)$ has cooperator fraction $x$ approach $1$ as $\lambda \to \infty$. 
\item If $\gamma < 2 |\alpha|$ (and $x^* = \frac{\gamma}{2|\alpha|}$ maximizes average payoff) , then $x^{\infty}_{\pm} = \frac{\gamma}{2 |\alpha|} \pm \left( 1 - \frac{\gamma}{2 |\alpha|} \right)$, so $x^{\infty}_{+} = 1$ and $x^{\infty}_{-} = \frac{\gamma}{|\alpha|} - 1 \in (0,1)$. In this case, the most abundant group type in steady state approaches an intermediate fraction $\frac{\gamma}{|\alpha|} - 1 $ as $\lambda \to \infty$. We further see that \[x^{\infty}_{-} =  \frac{\gamma}{2|\alpha|} + \underbrace{ \frac{\gamma}{2|\alpha|} - 1}_{< 0} < \frac{\gamma}{2 |\alpha|} = x^*,\] and therefore there are most abundant group type at steady states has fewer cooperators than is optimal even in the limit at $\lambda \to \infty$. 
\end{itemize}
\end{proof}
In particular, we notice in the limit in which $\gamma  \to \beta + |\alpha|$ that
$ \lim_{\gamma \to \left(\beta + |\alpha|  \right)}x^{\infty}_{-}  = \tfrac{\beta + |\alpha|}{|\alpha|} - 1 = \tfrac{\beta}{|\alpha|}$. This tells us that as $\gamma$ approaches a value at which the density $p(x)$ cannot be integrable, %
then the most abundant group type has cooperator composition approaching $x^{\infty}_{-} |_{\gamma = \beta + |\alpha|}  = \frac{\beta}{|\alpha|}$, consistent with the prediction from Proposition \ref{prop:HDdelta} the long-time behavior of Equation \ref{eq:replicatormeasurepde} produces a steady-state of $\delta\left(x - \tfrac{\beta}{|\alpha|}\right)$ when $\gamma = \beta + |\alpha|$, regardless of the value of $\lambda$.

In the process of proving Proposition \ref{prop:HDpeak}, we also see that the most abundant group compostion at steady state is given by  
\begin{equation} \label{eq:xhatlambdahd} \hat{x}^{\lambda} = \left\{
     \begin{array}{cl}
       0 & :  0 \leq \lambda < \lambda^* + \beta \\
       \hat{x}^{\lambda}_{-} & :\lambda >  \lambda^* + \beta  
     \end{array}
   \right. \end{equation}
 We can also show a similar piecewise characterization of the average payoff in the population. From Equation \ref{eq:HDsteadyfitness} and the expressions for $G(x)$, $\pi_C(x) - \pi_D(x)$ from Equations \ref{eq:grouppayoffparamsimplified} and \ref{eq:picminuspid}, we see that 
 \[ \langle G (\cdot) \rangle_{f^{\lambda}_{\theta}} = \left(\gamma - |\alpha|\right) - \frac{\left(|\alpha| - \beta \right) \theta}{\lambda} = G(1) - \frac{\left(\pi_D(1) - \pi_C(1)\right)\theta}{\lambda} \]
 Using the expression for $\lambda^*$ from Equation \ref{eq:lambdaHDpayoff}, we further see that 
 \begin{equation} \label{eq:HDsteadyfitnesspayoff} \langle G(\cdot) \rangle_{f^{\lambda}_{\theta}} = \left( \frac{\lambda^*}{\lambda} \right) G\left( \tfrac{\beta}{|\alpha|} \right) + \left(1 - \frac{\lambda^*}{\lambda} \right) G(1) \end{equation}
 Therefore we see that $ G \langle (\cdot) \rangle_{f^{\lambda}_{\theta}} = G(\tfrac{\beta}{|\alpha|})$ when $\lambda = \lambda^*$ and that $\langle G \rangle_{f^{\lambda}_{\theta}}  \to G(1)$ as $\lambda \to \infty$. Notably, the average payoff at steady-state $\langle G(\cdot) \rangle_{f^{\lambda}_{\theta}}$ is limited by the average payoff of a full-cooperator group $G(1)$, and therefore an HD game with an intermediate average payoff optimum will always see suboptimal outcomes, even in the limit in which between-group selection is infinitely stronger than within-group selection. Because we know from Proposition \ref{prop:HDdelta} that $\mu_t(dx) \rightharpoonup \delta\left(x - \tfrac{\beta}{|\alpha|}\right)$ when $\lambda < \lambda^*$, we can characterize the average payoff at steady state for all values of $\lambda$ with the piecewise description
 \begin{equation} \label{eq:Glambdahd} \langle G \rangle_{f^{\lambda}_{\theta}} = \left\{
     \begin{array}{cl}
       G\left( \tfrac{\beta}{|\alpha|} \right) & :  0 \leq \lambda < \lambda^* \\
   \left( \frac{\lambda^*}{\lambda} \right) G\left( \tfrac{\beta}{|\alpha|} \right) + \left(1 - \frac{\lambda^*}{\lambda} \right) G(1)& :\lambda >  \lambda^* 
     \end{array}
   \right. \end{equation}
 
 In Figure \ref{fig:HDghostfigure}(left), we plot the average payoff at steady state $\langle G (\cdot) \rangle_{f^{\lambda}_{\theta}}$ and the average payoff of the most abundant group type at steady state $G(\hat{x}_{\lambda})$ as functions of $\lambda$, showing that both tend to $G(1)$ as $\lambda \to \infty$ (lower-dashed line) rather than the maximal possible group payoff (upper-dotted line). In Figure \ref{fig:PDghostfigure}(right), we plot the maximal group payoff $G(x^*)$ and the average payoff at steady state in the limit as $\lambda \to \infty$ ($\ds\lim_{\lambda \to \infty} \langle G (\cdot) \rangle_{f^{\lambda}_{\theta}} = G(1)$). In Figure \ref{fig:HDfitnesspeakalphagamma}, we present heat maps of the most abundant group type at steady state (left) and the average payoff at steady state (right) in the large $\lambda$ limit as a function of $\alpha$ and $\gamma$. In Figure \ref{fig:hdpeakvsmax}, we plot the group payoff function $G(x)$ for two sets of parameter values, showing how the group payoff achieved by the most abundant group composition $G\left(\lim_{\lambda \to \infty} \hat{x}_{\lambda}\right) = G(1)$, rather than the group type which maximizes collective payoff.
 
\begin{figure}[H]
    \centering
    \hspace{-5mm} \includegraphics[width=0.505\textwidth]{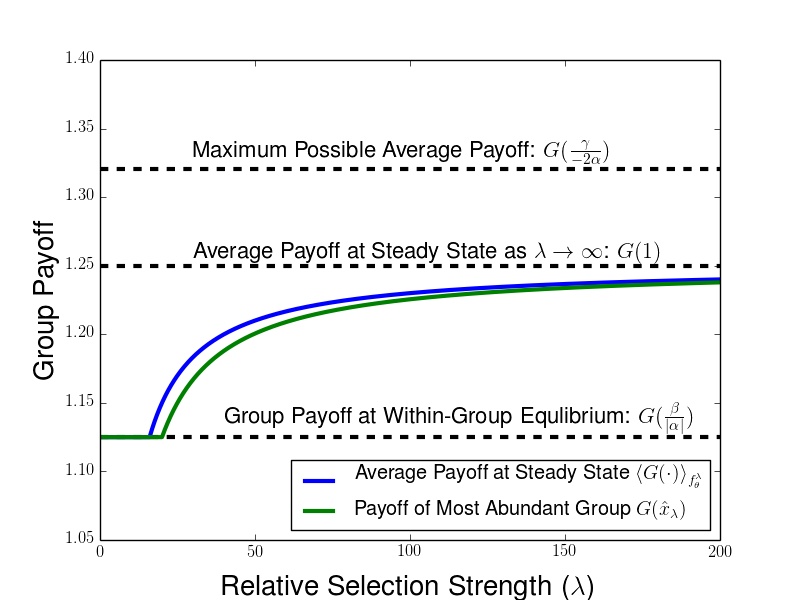} 
   \hspace{-5mm} \includegraphics[width=0.505\textwidth]{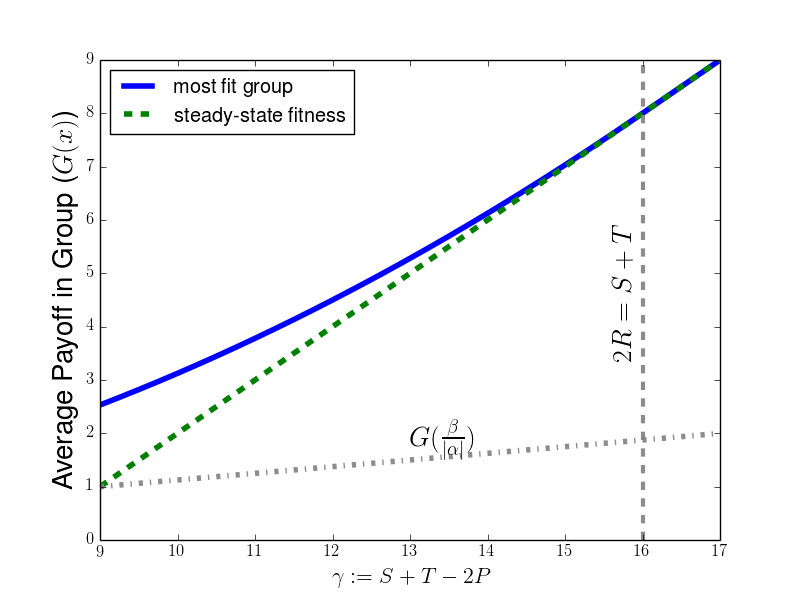} 
    \caption{Comparison between average payoff at steady state and optimal average payoff for a group in the HD game. (Left) Plots of average payoff at steady state $\langle G (\cdot) \rangle_{f^{\lambda}_{\theta}}$ and payoff of most abundant group type in steady state distribution $G(\hat{x}_{\lambda})$ as a function of $\lambda$ for an HD game with optimal group composition $x^* = 0.8125$. Lower dashed line corresponds to $G(1)$, the limit of average steady state payoff as $\lambda \to \infty)$, and the upper dashed line corresponds to maximal possible group payoff $G(x^*)$. (Right) Plots of the group type with maximal average payoff $G(x^*)$ (blue solid line) and the average payoff at the population at steady state $\lim_{\lambda \to \infty} \langle G (\cdot) \rangle_{f^{\lambda}_{\theta}}$ in the large $\lambda$ limit (green dashed line), each described as a function of $\gamma$ with a fixed choice of $\alpha = -8$. The gray dotted verticle line corresponds to the value $\gamma =16$ at which the group compostion maximizing payoff $x^* = -\frac{\gamma}{-2\alpha} \big|_{\alpha = -8} = 1$. We notice that the two lines coincide for $\gamma > 16$, when the full-cooperator groups are optimal for between-group competition, while the green dotted line falls below the blue solid line when $\gamma < 16$, as the average payoff of the population falls belows the interior optimal group payoff. In particular, when $\gamma \to 9$, we see that $\langle G (\cdot) \rangle_{f^{\lambda}_{\theta}} \to 1 = G(\frac{1}{8}) =  G(\tfrac{\beta}{|\alpha|})$, the payoff of the within-group equilibrium for the HD game. Therefore the group achieves  even though the for this game the group average payoff is $G(x) = 9x - 8x^2$, for which groups composed of $x^* = \frac{9}{16}$ cooperators  are most favored by between-group competition.  }
    \label{fig:HDghostfigure}
\end{figure}

\begin{figure}[H]
    \centering
   \hspace{-5mm} \includegraphics[width=0.535\textwidth]{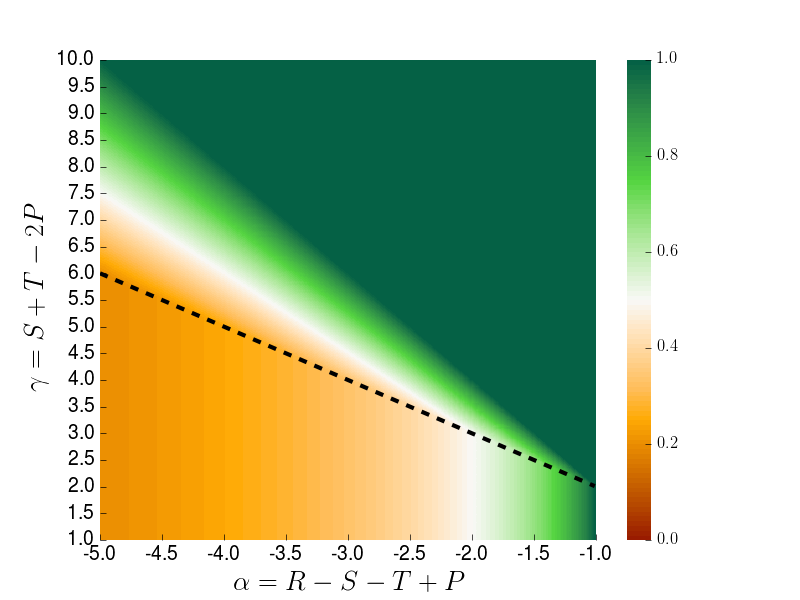} 
    \hspace{-7mm} \includegraphics[width=0.505\textwidth]{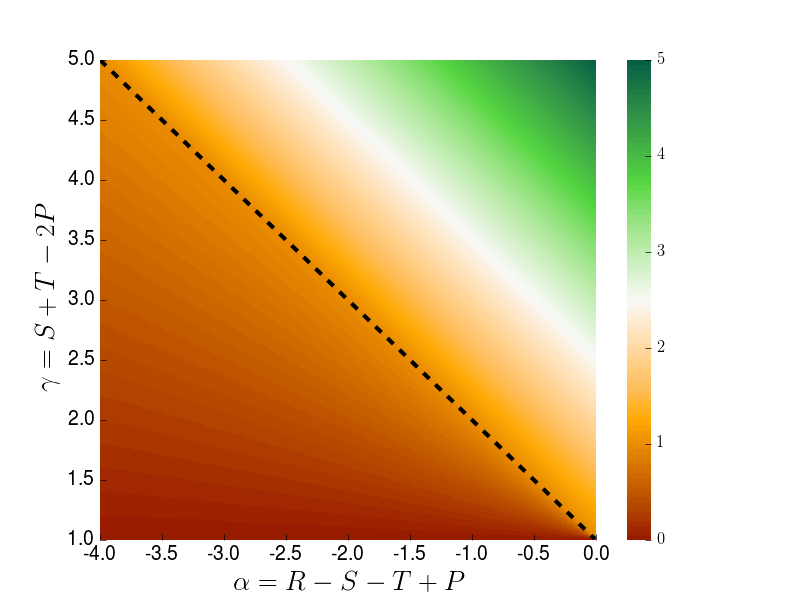}
    \caption{Most abundant group at steady state (left) and average payoff of the population (right) in the limit as between-group selection $\lambda \to \infty$. The region above the dotted black lines indicates values of $\alpha$ and $\gamma$ which constitute a HD game when $\beta = 1$.}
    \label{fig:HDfitnesspeakalphagamma}
\end{figure}

\begin{figure}[H]
    \centering
   \hspace{-5mm} \includegraphics[width=0.455\textwidth]{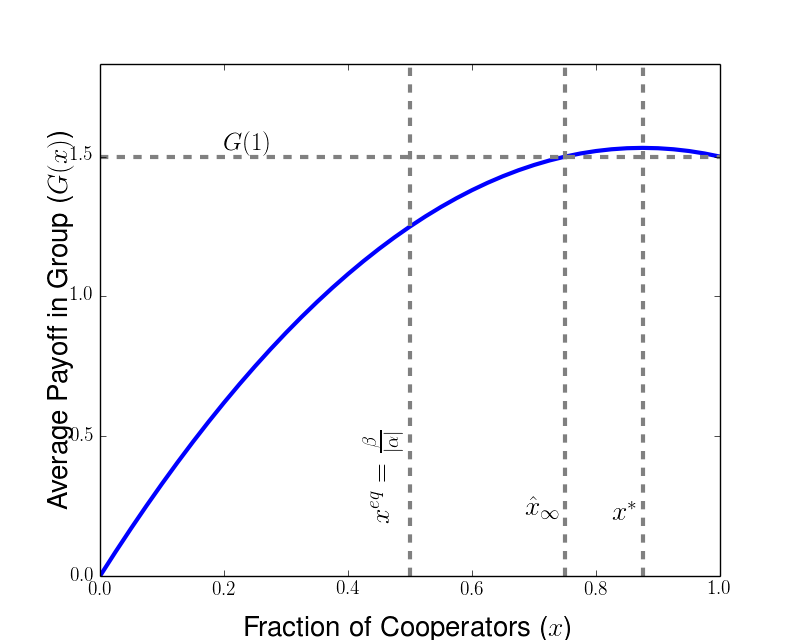} 
    \hspace{-5mm} \includegraphics[width=0.455\textwidth]{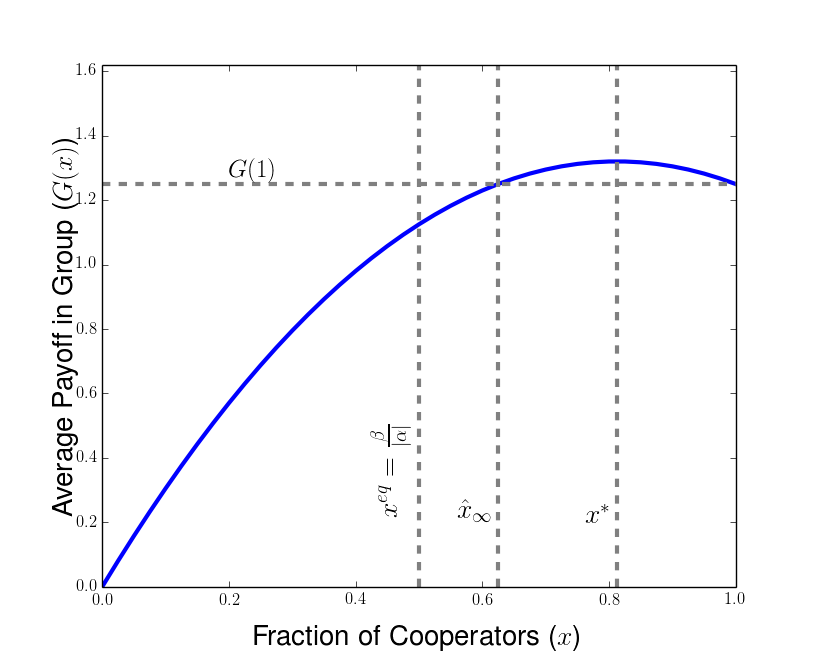}

    \caption{Illustration of peak abudance at steady state as $\lambda \to \infty$ and group type with maximal payoff $x^*$ for the HD game with $\alpha = -2$, $\beta = 1$, and $\gamma = 3.5$ (left) and $\gamma = 3.25$ (right). In both cases, we see how the modal level of cooperation at steady state $\hat{x}_{\lambda}$ in the large $\lambda$ limit has the same collective payoff as the full cooperator group, $G(\hat{x}_{\infty}) = G(1)$, although the gap between the optimal and modal group compositions widens as the optimum $x^*$ decreases.
     }
    \label{fig:hdpeakvsmax}
\end{figure}

\section{SH Games and Other Social Dilemmas} \label{sec:SHandother}

In this section, we discuss the multilevel dynamics for the other two-player, two-strategy social dilemmas. Considering a game with the payoff matrix of Equation \ref{eq:generalpayoffmatrix}, we characterize social dilemmas in terms of the following four conditions on payoffs
\begin{enumerate}[(i)]
\item $R > P$: mutual cooperation yields greater payoff than mutual defection
\item $T > R$: the temptation to defect exceeds the reward of mutual cooperation
\item $P > S$: punishment for mutual defection yields greater payoff than cooperation with a defector
\item $T > S$: the defector gets the best outcome of an interaction between a cooperator and defector 
 \end{enumerate}
A relatively expansive definition of a social dilemma is any game for which condition (i) holds and at least one of conditions (ii)-(iv) holds as well \cite{cooney2016assortment,allen2015games}. Other definitions of social dilemmas have been proposed for more general classes of games \cite{broom2018generalized} or for more a more restrictive class of matrix games \cite{kerr2004altruism,hauert2006synergy,pena2016ordering}, but we will consider the definition above to illustrate the variety of qualitative behaviors of multilevel dynamics that are possible for different games. In addition to the PD, HD, and SH games, there are five other games which satisfy our definition of a social dilemma, which rank payoffs as follows 
\begin{align*}
 T > S > R > P \: &: \: \: \textnormal{anti-coordination game 1 (AC1)} \\ 
 S > T > R > P \: &: \: \: \textnormal{anti-coordination game 2 (AC2)} \\
 R > P > S > T \: &: \: \: \textnormal{coordination game 1 (CG1)} \\ 
 R > P > T > S \: &: \: \: \textnormal{coordination game 2 (CG2)} \\
 R > T > S > P \: &: \: \: \textnormal{Prisoners' Delight (PDel)} 
\end{align*}
Within-group dynamics promote coexistence of a cooperators and defectors in anti-coordination games at $x^{eq} = \frac{\beta}{|\alpha|}$, bistability of all-cooperator and all-defector groups in coordination games (with unstable equilibrium $x^{eq} = \frac{|\beta|}{\alpha}$), and dominance of cooperators over defectors in the Prisoners' Delight. 

Using the defining payoff rankings, we find that group payoff is maximized by the full-cooperator group for the Prisoners' Delight and both of the coordination games. Furthermore, we can use calculations like we made for the SH game in Section \ref{sec:grouppayoff} to see that $G(x)$ is always increasing for the Prisoners' Delight and that the group level payoff for the coordination games is analogous to the behavior of the SH game. For the anti-coordination games, we see that $S + T > 2R$, so we know that group payoff is maximized by an intermediate level of cooperators $x^* = \frac{\gamma}{2 |\alpha|}$. Comparing the optimum group outcome with the stable within-group equilibrium $x^{eq} = \frac{\beta}{|\alpha|}$, we see that $x^* > x^{eq}$ if and only if $\gamma > 2 \beta$. This is true when $T > S$, and therefore we see that $x^* > x^{eq}$ for AC1 and that $x^* < x^{eq}$ for AC2.

Now we would like to characterize the long-time behavior of the remaining social dilemmas. For stag-hunt games, the other coordination games, and the Prisoners' delight, the all-cooperator equilibrium is locally stable under the within-group replicator dynamics. Because average group payoff is maximized by full-cooperator groups as well for these games, we can show that full-cooperation will be promoted by multilevel selection. In Proposition \ref{prop:staghunt}, we show, in the presence of any between-group competition (namely when $\lambda > 0$), that the population will eventually concentrate to full-cooperation in all groups if there is initially some density of groups arbitrarily close to the full-cooperator equilibrium. For convenience, we will analyze the dynamics using the density-valued formulation of Equation \ref{eq:replicatorpde}, though it would be nice to generalize this result to the measure-valued formulation.

\begin{proposition} \label{prop:staghunt} Suppose we have an initial density $f_0(x)$ such that for some $z^- < 1$, $\int_z^1 f_0(y) dy > 0$ for each $z > z^-$. If $\lambda > 0$, then the solution $f(t,x)$ solving Equation \ref{eq:replicatorpde} satisfies $f(t,x) \rightharpoonup \delta(1-x)$ as $t \to \infty$.  \end{proposition}

\begin{proof} %
We denote the probability $P_{\mc{I}_z}(t)$ of being in the interval $\mc{I}_{z} = [z,1]$ at time $t$, given by $P_{\mc{I}_z}(t) = \int_z^1 f(t,y) dy$, and integrate Equation \ref{eq:replicatorpde} from $z$ to $1$ to obtain 
\begin{align*} \ddt{P_{\mc{I}_z}} = \int_z^1 \dsdel{f(t,y)}{t} dy &= -\left[x(1-x) \left(\pi_C(x) - \pi_D(x) \right) f(t,x) \right] \bigg|_z^1  \\  &+ \lambda \left[ \int_z^1 G(x) f(t,x) dx - P_{\mc{I}_z} \left( \int_0^1 G(y) f(t,y) dy \right) \right]
\end{align*} 
First we will address the SH and the coordination games. Because $\pi_C(x) > \pi_D(x)$ for $x > \frac{|\beta|}{\alpha}$ for the SH and coordination games, choosing $z > \frac{|\beta|}{\alpha}$ for these games implies that $z(1-z) \left( \pi_C(z) - \pi_D(z) \right) \geq 0$. Further noting that $x=1$ is an equilibrium of the within-group dynamics, we have that
\begin{align*}  \ddt{P_{\mc{I}_z}}  &\geq \lambda \left[ \int_z^1 G(x) f(t,x) dx - P_{\mc{I}_z} \left( \int_0^1 G(y) f(t,y) dy \right) \right] \\ &=  \lambda P_{\mc{I}_z} \left[ \frac{1}{P_{\mc{I}_z}} \int_z^1 G(x) f(t,x) dx -  \left( \int_0^1 G(y) f(t,y) dy \right) \right]   \end{align*}
Noting that $ \frac{1}{\mc{I}_z} \int_z^1 G(x) f(t,x) dx$ is the conditional mean of $G(x)$ for $x \in \mc{I}_z$ and that $G(x)$ is increasing in $x$ and greater than $\max_{y \in [0,x^{eq}]} G(y)$ when $x > x^{eq}$ for both the SH and coordination games, we know that $ \frac{1}{P_{\mc{I}_z}} \int_z^1 G(x) f(t,x) dx \geq \langle G(\cdot)\rangle_f(x)$, 
for $x \geq z$. Therefore we have that $\ddt{P_{\mc{I}_z}}  > 0$ whenever $0 < P_{\mc{I}_z}(t) < 1$, and, then, if $z > z^{-}$, we further know that $0 < P_{\mc{I}_z}(0) \leq  P_{\mc{I}_z}(t) \leq 1$. Therefore we know for $z > \max\left(\tfrac{|\beta|}{\alpha}, z^{-} \right)$ that there exists a limit $P_{\mc{I}_z}^* \in \left[ P_{\mc{I}_z}(0) , 1 \right]$ such that $P_{\mc{I}_z} \to P_{\mc{I}_z}^*$ as $t \to \infty$. Now we will show that $P_{\mc{I}_z}^* = 1$. Suppose instead that $P_{\mc{I}_z}^* < 1$. Then consider $z$ and $z'$ such that $z' > z > \max\left(\tfrac{|\beta|}{\alpha},z^-\right)$, we know from the inequality above that 
\begin{dmath*} \ddt{P_{\mc{I}_{z'}}(t)} \geq \lambda P_{\mc{I}_{z'}(t)}\left[ \frac{1}{P_{\mc{I}_{z'}}} \int_{z'}^1 G(x) f(t,x) dx -  \left( \int_0^1 G(y) f(t,y) dy \right) \right] \\ \geq  \lambda P_{\mc{I}_{z'}(t)}\left[ \left( \frac{1}{P_{\mc{I}_{z'}}(t)} - 1\right) \int_{z'}^1 G(x) f(t,x) dx -   \int_0^z  G(y) f(t,y) dy  - \int_z^{z'} G(y) f(t,y) dy  \right]  \end{dmath*}
We see that the first integral has the lower bound 
\[  \int_{z'}^1 G(y) f(t,y) dy \geq G(z') P_{\mc{I}_{z'}}(t), \]
while the second and third integrals have the upper bounds \[ \int_0^z G(y) f(t,y) dy \leq G(z) \left( 1 - P_{\mc{I}_z}(t) \right)   \: \: \mathrm{and} \: \:  \int_{z}^{z'} G(y) f(t,y) dy \leq G(z') \left( P_{\mc{I}_z}(t) - P_{\mc{I}_{z'}}(t) \right) \]
Incorporating these estimates into our above inequality  and noting that  $\frac{1}{P_{\mc{I}_{z'}}(t)} - 1 > 0$ for $P_{\mc{I}_{z'}}(t) \leq P_{\mc{I}_{z'}}^* < 1$ yields 
\begin{dmath}  \label{eq:SHproofsecondlastinequality} \ddt{P_{\mc{I}_{z'}}(t)} \geq \lambda P_{\mc{I}_{z'}}(t)\left[ \left( 1 - P_{\mc{I}_{z'}}(t) \right) G(z') - \left( P_{\mc{I}_z}(t) - P_{\mc{I}_{z'}}(t) \right) G(z') - \left( 1 - P_{\mc{I}_z}(t) \right) G(z) \right] 
\\ =  \lambda P_{\mc{I}_{z'}}(t) \left( 1 - P_{\mc{I}_{z}}(t) \right) \left( G(z') - G(z) \right) \end{dmath}
Because $z' > z$, we know for our relevant games that $G(z') > G(z)$ when $z > \max\left(\frac{|\beta|}{\alpha}, z^- \right)$. Furthermore, we know that $1 -  P_{\mc{I}_{z}}(t)$ is a decreasing function of time, and therefore we know that $1 - P_{\mc{I}_{z}}(t) \geq 1 - P_{\mc{I}_{z}}^* > 0$. This let's us know that
\begin{dmath}  \label{eq:SHprooflastinequality} \ddt{P_{\mc{I}_{z'}}(t)} \geq \left[ \lambda \left(1 - P_{\mc{I}_{z}}^* \right) \left( G(z') - G(z) \right) \right] P_{\mc{I}_{z'}}(t)  \end{dmath}
We can solve this differential inequality to find that 
\[  P_{\mc{I}_{z'}}(t) \geq P_{\mc{I}_{z'}}(0) e^{ \left[ \lambda \left(1 - P_{\mc{I}_{z}}^* \right) \left( G(z') - G(z) \right) \right] t}. \] 
Knowing that $\lambda \left(1 - P_{\mc{I}_{z}}^*\right) \left( G(z') - G(z) \right) > 0$, we deduce that $P_{\mc{I}_{z'}}(t) \to \infty$ as $t \to \infty$, contradicting the fact that $P_{\mc{I}_{z'}}(t)$ is a probabilit. This allows us to conclude that $P_{\mc{I}_{z}}^* = 1$ for $z  > \max\left(\frac{|\beta|}{\alpha}, z^- \right)$. 
Because this condition holds for all $z$ sufficiently close to 1, we know that probability of being in intervals of the form $[z,1]$ gets arbitrarily close to $1$ in the long-run for all such $z$. This means that the probability density concentrates arbitrarily closely to $1$ in the long-time limit, and we can conclude that $f(t,x) \rightharpoonup \delta(1-x)$ as $t \to \infty$ for the SH and coordination games.

For the Prisoners' Delight, we know both that $\pi_C(x) - \pi_D(x) \geq 0$ and that $G(x)$ is increasing for all $x \in [0,1]$. Therefore we can use the same argument as above, but don't have to worry about an interior equilbrium at $\frac{|\beta|}{\alpha}$, so intead we simply require choosing $z' > z > z^-$ and repeat the steps above.   \end{proof}

\begin{remark} Our assumption on the initial density requires that there is some density of groups arbitrarily close to a full-cooperator composition. In principle, we could relax this assumption to only requiring that there is some initial positive fraction of groups within the basin of attraction for the full-cooperator groups under the within-group dynamics. This would necessitate waiting for the within-group dynamics to produce sufficiently high levels of cooperation, and then invoking the arguments made above. \end{remark}

The remaining social dilemmas to analyze are the anti-coordination games. Like the Hawk-Dove game, within-group dynamics push groups towards a mixed composition of cooperators and defectors at $x^{eq} = \tfrac{\beta}{|\alpha|}$. 

\begin{proposition} \label{prop:anti-coordination}
Consider either of the anti-coordination games and an initial distribution with \holder exponents $\zeta > 0$ and $\theta > 0$ near $x=0$ and $x=1$, respectively. For any level of selection strength $\lambda \geq 0$, $\mu_t(dx) \rightharpoonup \delta(x - \frac{\beta}{|\alpha|})$ as $ t \to \infty$.
\end{proposition}

\begin{remark} The proof strategy is identitical to that of Proposition \ref{prop:HDdelta}, except that we can omit the condition $\lambda \left( \gamma - \left( \beta + |\alpha| \right) \right) < |\alpha| \theta$ because \[ \gamma - \left( \beta + |\alpha| \right) = \left( S + T - 2 P \right) - \left( S - P - (R - S - T + P) \right) = R -S < 0  \]
for the anti-coordination games. Then the lefthand side of the inequality is negative, and therefore is satisfied for any possible non-negative $\lambda$. The results of Lemmas \ref{lem:HDleftprob} and \ref{lem:HDleftint} carry through for the anti-coordination games as well. 
\end{remark}

One manner in which we can explore the most extreme effects of the shadow of lower-level selection in the HD game is to increase $S$ towards $R$ until the threshold for achieving additional cooperation $\lambda^*_{HD} \to \infty$ as $S \to R$ (or $\gamma \to \beta + |\alpha|)$. Taking an HD game and increasing $S$ past $R$ sees the transition in payoff rankings
\[ T > R > S > P \Longrightarrow  T > S > R > P, \] resulting in an anti-coordination game (AD1). Therefore one way to understand the result of Proposition \ref{prop:anti-coordination} is as the behavior HD dynamics in the limit of $S \uparrow R$ extended into the region in which $R < S$. Further increasing $S$ past $T$ then results in the transition of payoff rankings
\[ T > S > R > P \Longrightarrow  S > T > R > P, \] correpsonding to a transition from AD1 to AD2. 

\myindent Because the long-time dynamics for the anti-coordination games results in achieving an $x^{eq}$ fraction of cooperators, this means that there are fewer cooperators than what is optimal for the group when $x^* > x^{eq}$ (AD1), while there are more cooperators than what is optimal for the group when $x^* < x^{eq}$ (AD2). Then the behavior of the AD2 game is somewhat unique amongst the other cooperative dilemmas, as the multilevel dynamics actually produces too much cooperation. This result is not in and of itself inconsistent with the idea of the shadow of lower-level selection, and this outcome with more cooperators than is optimal still produces a lower collective payoff the desired state with fewer cooperators. Perhaps this calls into question the meaning of the word ``cooperation'' in the context in which groups are best off with a level of cooperation less than what is achieved by individual-level selection alone, as is the case with the AD2 game. However, together with the results of the PD and HD games with interior group payoff optima, the results for the anti-coordination games do show us that the intermediate optimum is never achievable for any level of relative selection strength $\lambda$, even in the limit of infinitely strong between-group selection. Therefore, we have seen that all two-player, two-strategy games with intermediate group payoff optima display an indelible shadown of individual-level competition.

\section{Discussion} \label{sec:Discussion}

In this paper, we have shown that comparison principles can be used to characterize when multilevel replicator dynamics converge to delta-functions at within-group equilibria, and to show that the \holder exponents near $x=1$ are preserved in time, demonstrating how this invariant quantity can select a single steady state for a given initial population out of an inifinity of possibilities. Notably, we have characterized the cooperative dilemmas for which the there is a shadow of lower-level selection: for games in which the average payoff of a group is maximized by a composition with less than full cooperation, the most abundant group composition at steady state always has fewer cooperators than in the optimal composition, even in the limit of infinitely-strong between group selection.  

\myindent With the methods used in this paper, we have also found a potential strategy for analyzing a broader class of multilevel selection problems. By proving that the \holder exponent near $x=1$ is preserved by the multilevel dynamics, we have also gained more insight into the long-time behavior found in previous analyses of special cases of these models \cite{luo2017scaling,cooney2019replicator}. These techniques are employed in a companion paper exploring the role of assortment and reciprocity mechanisms for individual interactions in altering the within-group and between-group competition \cite{cooney2019assortment}. In those cases, varying a parameter characterizing the probability of assortment or reciprocity will change the within-group replicator dynamics, and therefore one naturally would like to compare multilevel selection models beyond a cherry-picked family of payoff matrices with solvable withn-group dynamics. 

\myindent In addition, the comparison principle techniques from this paper may find application in other nonlocal PDEs for models in game theory and collective behavior. Advection terms also arise in models for continuous-strategy games in which individuals gradually adjust their strategies by climing their local payoff gradient \cite{friedman2010gradient,friedman2013evolutionary,friedman2008conspicuous}. In such models, the individual dynamics can depend nonlocally on the strategy distribution of the population, leading to nonlocal characteristic ODEs that may require comparison principle approaches to study long-time behavior. These models and methods have also found use in continuum models of opinion dynamics \cite{aletti2007first}, and similar approaches could find use with proving convergence to consensus of opinions or to Nash equilibria in continuous-strategy games.

\myindent The work in this paper also lays out challenges for future analytical and numerical work. We have made Conjectures \ref{con:pddensity} and \ref{con:hddensity} about the long-time behavior of multilevel PD and HD dynamics, and we would still like to show that solutions of Equation \ref{eq:replicatormeasurepde} converge to a density steady-state in the cases in which we know that the steady-states exists and that the dynamics does not converge to a delta-function. 
Dawson suggests a strategy for characterizing long-time behavior using an invariant quantity for the moments of the distribution \cite{dawson2013multilevel}. The result in this paper on \holder exponents serves a similar purpose to the idea suggested by Dawson, but characterizing the equivalent approach with moments may provide an alternate strategy for proving convergence to steady state and bears resemblence to the strategy used to analyze the long-time behavior for Becker-D{\''o}ring models of aggregation-fragmentation processes \cite{ball1986becker,bressloff2016aggregation,conlon2019non}. In related models possesing individual level selection and replicator-mutator or replicator-diffusion dynamics, formulation of these systems as gradient flows has provided a strategy for proving convergence of the dynamics to a steady-state solution \cite{chalub2019gradient,jabin2017non}, and others models with diffusion have used arguments regarding the principle eigenvalue of the diffusion operator \cite{burger1996stationary,alfaro2018evolutionary} or decay of an energy-like function \cite{ogura1987stationary,ogura1987stationary2} to prove convergence to steady-state. In addition to analytical attempts to charactertize the long-time behavior, effort should be placed into developing numerical methods for solving Equations \ref{eq:replicatorpde} and \ref{eq:replicatormeasurepde} and comparing the numerical solutions to analytical predictions.

\myindent So far, we have only considered between-group competition in which group-level birth and death events depend on the average payoff of group members, and don't consider any interaction between the groups. One could attempt to extend the multilevel replicator dynamics PDE to describe strategic or frequency-dependent interactions between groups. Mathematically, the term describing between-group competition would be reminiscent of the continuous trait competition models explored by in replicator diffusion equations \cite{jabin2017non}, and could address game-theoretic problems such as the frameworks of team games \cite{bornstein2003intergroup} and hierachical games \cite{fujimoto2017hierarchical}, which model games played between groups in which group strategy is determined by the choices of its constituent indviduals. Another interesting direction for future research is to explore the competing effects of three or more levels of selection. Three competing levels of selection arise naturally in virus dynamics, as defecting viral genomes \cite{manzoni2018defective} and collective infectious units \cite{leeks2019evolution} make possible the misalignment of evolutionary incentives at the within-cell, with-host between-cell, and between-host levels of pathogen transmission. To describe simultaneous selection at three levels of organization, one must describe the population by a probability distribution of probability distributions, which was the approach taken by Ambrosio et al to describe a spatial version of the replicator dynamics with mixed strategies and spatial movement \cite{ambrosio2018spatially}.

\myindent From both a biological and mathematical perspective, an important direction for future research is to better understand the mechanisms and significance of the shadow of lower-level selection. In particular, do signatures of this outsized effect of individual selection continue to arise in models which take into account heterogeneous group size and group level fission/fusion events \cite{simon2010dynamical,simon2012numerical,simon2013towards,simon2016group,puhalskii2017large,markvoort2014computer} or in models in which group structure is emergent from assortative interactions \cite{jensen2018evolutionary,bergstrom2002evolution} or spatial clustering \cite{krieger2019turbulent,boza2010beneficial}. It is also interesting to ask whether a scenario in which collapse of group benefit in a two-level system, as seen in phenomena from mitochondrial DNA \cite{haig2016intracellular} and cancer \cite{aktipis2015cancer} are better attributed to incentives favoring the dominance of cheaters or a group-level scenario preventing the possibility of cooperation, as would hold in determinsitic analogues of stochastic corrector models \cite{szathmary1987group,fontanari2013solvable}. Overall, this phenomenon displayed by two-level replicator dynamics provides an interesting motivation to further explore related mathematical frameworks and to connect the behavior of two-level models with empirical work in biology.

\renewcommand{\abstractname}{Acknowledgments}
\begin{abstract} 
 This research was supported by NSF grants DMS-1514606 and GEO-1211972 and by ARO grant W911NF-18-1-0325. I am thankful to Simon Levin, Joshua Plotkin, Carl Veller, and Fernando Rossine for helpful discussions. 

\end{abstract}

\addtocontents{toc}{\protect\setcounter{tocdepth}{1}}
\appendix

\section{Well-Posedness for Measure-Valued Formulation} \label{sec:measureexistence}

In this section, we will demonstrate that our representation of $\mu_t(dx)$ in terms of the push-forward measure of $\mu_0(dx)$ is well-posed, justifying our use of the push-forward representation in proving preservation of the \holder exponents and in characterizing convergence to delta-functions at within-group equlibrium below the threshold for existence of steady-state. Our strategy for proving existence of solutions to our measure-valued equations involves a contraction mapping approach often used for hyperbolic scalar or systems of equations with nonlocal terms, which often arise in models of populations structured by age or size \cite{dawidowicz1986existence},  as well as models of collective motion of animal groups \cite{eftimie2009weakly,canizo2011well} or bacterial chemotaxis \cite{hillen2000hyperbolic}. This approach has also been considered in the context of measure-valued solutions for transport equations \cite{evers2015mild,evers2016mild}, models of collective motion, and models with genetic or age structure \cite{canizo2011well,canizo2013measure}.

\myindent To discuss existence of solutions, we can consider the following slight generalization of the measure-valued dynamics of Equation \ref{eq:replicatormeasurepde} 
\begin{equation} \label{eq:measurepdeexistence} \dsddt \ds\int_0^1 \Psi(x) \mu_t(dx) =  \ds\int_0^1 \left\{j(x) \dsdel{\Psi(x)}{x}   + \lambda  \Psi(x) \left[G(x) - \left( \int_0^1 G(y) \mu_t(dy)  \right) \right] \right\} \mu_t(dx) \end{equation}
where the within-group dynamics $j(x)$, the group payoff $G(x)$, and the test function $\Psi(x)$ are defined on $[0,1]$ and are continuously differentiable in $x$. We also assume that $j(0) = j(1) = 0$ to capture the feature that all-cooperator and all-defector groups are steady states of the within-group dynamics. 
To understand the solutions of Equation \ref{eq:measurepdeexistence}, we consider an associated linear PDE. Given an arbitrary $h(t) \in C([0,T])$, we define the linear PDE  
\begin{subequations} \label{eq:linearmeasureexistencepde}
\begin{align}  \dsddt \ds\int_0^1 \Psi(x) \mu^h_t(dx) &=  \ds\int_0^1 \left\{ j(x) \dsdel{\Psi(x)}{x}   + \lambda  \Psi(x) \left[G(x) - h(t)\right] \right\} \mu^h_t(dx) \\ \label{eq:initialh} \mu_0^{h}(dx) &= \mu_0(dx) \end{align}  \end{subequations}
where $\mu_t^h(dx)$ denotes our solution of Equation \ref{eq:linearmeasureexistencepde} for given $h(t)$. 
Because $j(x)$ is Lipschitz and $j(0) = j(1) = 0$, we see that the characteristic curves $x(t,x_0)$ satisfying $\dsddt{} x(t,x_0) = j(x(t))$ and $x(0,x_0) = x_0$ exist globally in time. %

\myindent Our strategy will be to use solutions of Equations \ref{eq:linearmeasureexistencepde} in order to establish existence a solution to Equation \ref{eq:measurepdeexistence}. We can do this by considering an arbitary forcing function $h^0(t)$, finding the corresponding solution $\mu_t^{h^0}(dx)$, and then finding a new forcing function $h^1(t) := \int_0^1 G(x) \mu_t^{h^0}(dx)$. By iterating this process, we hope to construct a sequence of forcing functions $h^j$ converging to a fixed point.
In other words, if we define our iteration function as \begin{equation} \label{eq:Hdefinition} H(h(t)) := \int_0^1 G(x)  \mu_t^h(dx), \end{equation} then we are looking to find a fixed point $h_{\flat}(t)$ satisfying $H(h_{\flat}(t)) = h_{\flat}(t)$, in which case $\mu_t^{h_{\flat}}(dx)$ satisfies Equation \ref{eq:measurepdeexistence}.
For our subsequent analysis, we denote $G^* := \max_{x \in [0,1]} G(x)$, $(G')^* = \max_{x \in [0,1]} |G'(x)|$, $j^* := \max_{x \in [0,1]} j(x)$, and $||h(t)||_T := \sup_{s \in [0,t]} |h(s)|$. 
Using the norm $||\cdot||_T$, we aim to show that $H(h(t))$ is a contraction on $C\left([0,T] \right)$. To discuss contraction mappings, we will look to estimate $||H(h(t)) - H(\tilde{h}(t))||_T$ for any two given functions $h(t), \tilde{h}(t)$ in our function space. In describing the measure-valued dynamics, it is helpful to use the shorthand notation \begin{equation} \label{eq:innerproduct} \langle \Psi, \mu_t \rangle := \int_0^1 \Psi(x) \mu_t(dx) \end{equation}
For example, this allows us to describe the dynamics of the auxiliary linear problem of Equation \ref{eq:linearmeasureexistencepde} as follows 
\begin{subequations} \label{eq:linearmeasureexistencepdeshorthand}
\begin{align}  \dsddt{} \langle \Psi , \mu^h_t \rangle &=  \left\langle j(x) \frac{\partial \Psi}{\partial x} , \mu^h_t \right\rangle  + \lambda  \langle \Psi \left[G(x) - h(t)\right]  ,  \mu^h_t \rangle \\ \label{eq:initialh} \mu_0^h &= \mu_0 \end{align}  \end{subequations}

\myindent Now we will present the results on well-posedness for the two-level PDE of Equation \ref{eq:measurepdeexistence}. First, we have two lemmas dealing with the linear auxiliary problem from Equation \ref{eq:linearmeasureexistencepde}. In Lemma \ref{lem:variationofconstants}, we show that there is a variation of constants formula which must be satisfied by any solution $\mu^h_t(dx)$ of Equation \ref{eq:linearmeasureexistencepde}, which serves as a useful tool for proving uniqueness and computing contraction-mapping estimates. In Lemma \ref{lem:pdeexistence}, we show that there exists a unique solution to Equation \ref{eq:linearmeasureexistencepde}, which has an explicit representation formula reminiscent of the implicit formula of Equation \ref{eq:mutimplicit} for the nonlinear problem. In Proposition \ref{prop:nonlinearexistence}, we use the results of the two lemmas and a contraction-mapping argument to show that Equation \ref{eq:measurepdeexistence} has a unique solution and that this solution $\mu_t(dx)$ satisfies the implicit representation formula of Equation \ref{eq:mutimplicit}.

\begin{lemma} \label{lem:variationofconstants} If $\mu^h_t(dx)$ is a measure-valued solution to Equation \ref{eq:linearmeasureexistencepde}, then it also satisfies the variation of constants formula 
\begin{equation} \label{eq:vocformula} \langle \Psi, \mu^h_t \rangle = \langle P_t \Psi, \mu^h_0 \rangle +  \lambda  \int_0^t  \langle P_{t-s}  \left(G(x) - h(t) \right), \mu_s^h \rangle ds \end{equation}
where $P_t (f(x)) = f(\phi_t(x))$ denotes evaluating functions along characteristic curves $\phi_t(x)$.  \end{lemma}

\begin{remark} This is a generalization of the variation of constants formula described in Lemma 11 of Luo and Mattingly \cite{luo2017scaling}, which described the special case in which $j(x) = -sx(1-x)$ and $G(x) = x$. The derivation of the formula is analogous for the generalized formula, so we will omit the proof.   \end{remark}

\begin{lemma} \label{lem:pdeexistence} Given $T > 0$, the flow of measures $\mu^h_t(dx)$ given by the formula \begin{subequations} \label{eq:linearmeasuresolution} \begin{align}  \label{eq:measureformula} \mu^h_t(dx)& = w^h_t(x) (\mu_0 \circ \phi_t^{-1})(dx) \\ \label{eq:whformula} w^h_t(\phi_t(x_0)) &= \exp\left( \lambda \int_0^t \left[ G(\phi_s(x_0)) - h(s) \right] ds \right) \end{align} is the unique solution of Equation \ref{eq:linearmeasureexistencepde} for each $t \in [0,T]$. \end{subequations}  \end{lemma}

\begin{proof}
We see that Equation \ref{eq:linearmeasuresolution} solve Equation \ref{eq:linearmeasureexistencepde} by differentiating with respect to time, obtaining 
\begin{align*} \dsddx{}{t} \ds\int_0^1 \Psi(x) \mu^h_t(dx) &= \dsddx{}{t} \ds\int_0^1 \Psi(x) w^h(x) (\mu^h_0 \circ \phi_t^{-1})(dx)  = \dsddx{}{t} \int_0^1 \Psi(\phi_t(x)) w^h_t(\phi_t(x)) \mu^h_0(dx) \\ &= \ds\int_0^1 \del{\Psi(\phi_t(x))}{\phi_t(x)} \left[ \del{\phi_t(x)}{t} \right] w^h_t(\phi_t(x)) \mu^h_0(dx) + \ds\int_0^1 \Psi(\phi_t(x)) \del{w^h_t(\phi_t(x))}{t} \mu^h_0(dx) \\ &=  \ds\int_0^1 \del{\Psi(\phi_t(x))}{\phi_t(x)} j(\phi_t(x)) w^h_t(\phi_t(x)) \mu^h_0(dx) \\ &+ \lambda \ds\int_0^1 \Psi(\phi_t(x)) \left[G(\phi_t(x)) - h(t) \right] w^h_t(\phi_t(x)) \mu^h_0(dx) \\ &= \ds\int_0^1 \left\{ \del{\Psi(x)}{x} j(x) w^h_t(x) + \lambda \Psi(x) \left[G(x) - h(t) \right] w^h_t(x) \right\} \left[\mu^h_0 \circ \phi_t^{-1} \right] (dx)  \\ &=  \ds\int_0^1 \left\{ \del{\Psi(x)}{x} j(x) + \lambda \Psi(x) \left[G(x) - h(t) \right]  \right\} \mu_t(dx) \end{align*} where we used that $\phi_t(x)$ solves $\del{}{t}\phi_t(x_0) = j(\phi_t(x_0))$, that $\del{}{t}w^h_t(\phi_t(x_0)) = \lambda \left( G(\phi_t(x_0)) - h(t) \right) w^h_t(\phi_t(x_0))$. Putting together the first line and last line, we can conclude that $\mu^h_t(dx)$ given by Equation \ref{eq:measureformula} is a solution to Equation \ref{eq:linearmeasureexistencepde}.

Next we see that $\mu_t^h(dx)$ is unique. Suppose there were two solutions $\mu_t^h(dx)$ and $\nu_t^h(dx)$ to Equation \ref{eq:linearmeasureexistencepde} (where notably the initial conditions agree upon $\mu^h_0(dx) = \nu_0^h(dx) = \mu_0(dx)$. Using the variation of constants formula, we have that \begin{align*} \langle \Psi,  \mu_t^h \rangle \: = \:  \langle P_t \Psi, \mu_0 \rangle + \lambda \int_0^t \langle \left(P_{t-s} \Psi \right) \left( G(x) - h(t) \right) , \mu_s^h \rangle ds  \\  \langle \Psi,  \nu_t^h \rangle \: = \:  \langle P_t \Psi, \mu_0 \rangle + \lambda \int_0^t \langle \left(P_{t-s} \Psi \right) \left( G(x) - h(t) \right) , \nu_s ^h\rangle ds \end{align*}
We see that $$ | \langle \Psi, \mu_t^h - \nu_t^h \rangle  | \leq  \lambda \left( G^* + ||h||_t \right) \int_0^t \langle P_{t-s} \Psi, \mu_s^h - \nu_s^h \rangle ds$$ 
Using the total variation norm $||\mu_t - \nu_t||_{TV} := \sup_{||\Psi|| \leq 1} \langle \Psi, \mu_t - \nu_t \rangle$ and that $||P_{t-s} \Psi||_{\infty} \leq ||\Psi_{\infty}||$ (because $P_{t-s} \Psi := \Psi(\phi_{t-s}(x_0))$), we can further see, for test functions satisfying $||\Psi||_{\infty} \leq 1$, that $$||\mu_t^h - \nu_t^h||_{TV} \leq \lambda\left(G^* + ||h||_t \right) \int_0^t ||\mu_t^h - \nu_t^h ||_{TV} ds $$ and \gronwalls inequality lets us deduce that $$||\mu_t^h - \nu_t^h||_{TV} \leq ||\mu_0^h - \nu_0^h||_{TV} \exp\left(\lambda \left[G^* + ||h||_t \right] T \right) = 0 \: \: \mathrm{because} \: \: \mu_0^h = \nu_t^h,$$
which allows us to conclude that $\mu_t^h(dx) = \nu_t^h(dx)$ for all $t \in [0,T]$ and solutions of Equation \ref{eq:linearmeasureexistencepde} are unique.
Together with our explicit representation formula from $\mu_t^h(dx)$, we can conclude that there exists a unique solution to equation \ref{eq:linearmeasureexistencepde} for given function $h(t)$. We now need to construct an iteration scheme use solutions from Equation \ref{eq:linearmeasureexistencepde} to produce a solution of Equation \ref{eq:measurepdeexistence}.
\end{proof}

\begin{proposition} \label{prop:nonlinearexistence} Assume that $j(x)$, $G(x)$, and $\psi(x)$ are $C^1$ on $[0,1]$. Given initial probability measure $\mu_0(dx)$, there eixsts a unique solution $\mu_t(dx)$ to Equation \ref{eq:measurepdeexistence} for all time $t \geq 0$. Furthermore, the solution $\mu_t(dx)$ satisfies the implicit representation formula from Equation \ref{eq:mutimplicit}  \[\mu_t(dx) = w_t(x) \left(\mu_0 \circ \phi_t^{-1}\right)(dx). \]

\end{proposition}

\begin{proof}

From the push-foward representation of $\mu_t^h(dx)$ and the fact that $\mu_0(dx)$ is a probability measure, we see that $\mu_t^h$ satisfies the following a priori estimate
\begin{dmath} \label{eq:aprioriestimate} \int_0^1 \Psi(x) \mu_t^h(dx) \leq \ds\int_0^1 \Psi(\phi_t(x)) w_t(\phi_t(x)) \mu_0(dx) \leq ||\Psi||_{\infty} \exp\left( \lambda \left[G^* + ||h||_T \right] T \right) \int_0^1 \mu_0(dx) \\ = ||\Psi||_{\infty} \exp\left( \lambda \left[G^* + ||h||_T \right] T \right) , \end{dmath} so for $T < \infty$ and given $h(t)$, there exists $M^T_{h} < \infty$ such that $\langle \Psi, \mu_t^h \rangle \leq M^T_h$.

To apply the Banach fixed-point theorem, we now need to show that 

\begin{enumerate}[(i)]

\item $H(h(t))$ is a contraction: $\exists \eta \in (0,1)$ such that $\forall h(t), \tilde{h}(t) \in C[0,T]$, $||H(h(t)) - H(\tilde{h}(t))||_T \leq \eta ||h(t) - \tilde{h}(t) ||_T $

\item There exists sufficiently large $R$ such that the closed ball $B(R,0)$, centered at $0$ with radius $R$, is mapped to itself by $H(h(t))$. 

\end{enumerate}

First we show that $H(h(t))$ is a contraction. Because $G(x)$ is an admissible test function, we can use Equation \ref{eq:linearmeasureexistencepde} to compute

\begin{align*} \hspace{-10mm} \dsdel{}{t} \bigg| \int_0^1 G(x) \left( \mu_t^h(dx) - \mu_t^{\tilde{h}}(dx) \right) \bigg| & \leq \bigg| \dsdel{}{t}  \int_0^1 G(x) \left( \mu_t^h(dx) - \mu_t^{\tilde{h}}(dx) \right) \bigg| \\ &= \bigg| - \int_0^1 G'(x) j(x) \left(\mu_t^h(dx) - \mu_t^{\tilde{h}}(dx) \right) \bigg| \\ &+ \lambda \bigg| \int_0^1 G(x) \left[ G(x) \left(\mu_t^h(dx) - \mu_t^{\tilde{h}}(dx) \right) + \left(h(t) \mu_t^h - \tilde{h}(t) \mu_t^{\tilde{h}} \right) \right]
\bigg| \\ & \leq (G')^* j^* \bigg| \int_0^1 \left( \mu_t^h(dx) - \mu_t^{\tilde{h}}(dx) \right) \bigg| + \lambda G^* \bigg| \int_0^1 G(x) \left(\mu_t^h(dx) - \mu_t^{\tilde{h}}(dx) \right) \bigg|  \\ &+ \lambda \bigg| \int_0^1 G(x) h(t) \left( \mu_t^h (dx) - \mu_t^{\tilde{h}}(dx) \right) \bigg| + \lambda \bigg| \int_0^1 G(x) \left( h(t) - \tilde{h}(t) \right) \mu_t^{\tilde{h}}(dx) \bigg| \end{align*} 
Because $h(t) \in C([0,T])$, we know that there exists a bound $B_{h} < \infty$ such that $h(t) \leq ||h(t)||_T \leq B_h$. From our a priori estimate on $\mu_t^{\tilde{h}}(dx)$, we know that there is $M^T_{\tilde{h}}$ such that $\int_0^1 \mu_t^{\tilde{h}} (dx) \leq M^T_{\tilde{h}}$ for $t \in [0,T]$ and a corresponding bound $M^T_h$ for $\mu_t^h(dx)$.  Using these, we can now say that 
\begin{align*} \hspace{-10mm} \dsdel{}{t} \bigg| \int_0^1 G(x) \left( \mu_t^h(dx) - \mu_t^{\tilde{h}}(dx) \right) \bigg| & \leq  (G')^* j^* \left(M^T_{h} + M^T_{\tilde{h}} \right)  + \lambda \left(G^* + B_{h} \right) \bigg| \int_0^1 G(x) \left( \mu_t^h(dx) - \mu_t^{\tilde{h}}(dx) \right) \bigg| \\ &+ \lambda G^* M^T_{\tilde{h}} ||h(t) - \tilde{h}(t)||_T  \end{align*} 

If $||h(t) - \tilde{h}(t)||_T = 0$, we know that $h(t) = \tilde{h(t)}$ and $\mu_t^h(dx) = \mu_t^{\tilde{h}}(dx)$ by the uniqueness of solutions to Equation \ref{eq:linearmeasureexistencepde}, and we can correspondingly conclude $\int_0^1 G(x) \left( \mu_t^h(dx) - \mu_t^{\tilde{h}}(dx) \right) = 0$ for $t \in [0,T]$.
In the alternate case that $||h(t) - \tilde{h}(t)||_T > 0$, we know that there exists $W < \infty$ such that $(G')^* j^* \left(M^T_{h} + M^T_{\tilde{h}} \right) \leq W ||h(t) - \tilde{h}(t)||_T$, and we can write that  \begin{align*} \hspace{-10mm} \dsdel{}{t} \bigg| \int_0^1 G(x) \left( \mu_t^h(dx) - \mu_t^{\tilde{h}}(dx) \right) \bigg| & \leq  \lambda \left(G^* + B_{h} \right) \bigg| \int_0^1 G(x) \left( \mu_t^h(dx) - \mu_t^{\tilde{h}}(dx) \right) \bigg| + \left( W + \lambda G^* M^T_{\tilde{h}} \right) ||h(t) - \tilde{h}(t)||_T \end{align*}

By \gronwalls inequality, we have that \[ \bigg| \int_0^1 G(x) \left( \mu_t^{h}(dx) - \mu_t^{\tilde{h}}(dx) \right) \bigg| \leq \left( \frac{W + \lambda G^* M_{\tilde{h}}^T}{\lambda \left(G^* + B_h \right)} \right) \left(e^{\lambda \left(G^* + ||h||_T \right)T} - 1\right) ||h(t) - \tilde{h}(t)||_T. \] 
For an $\eta \in (0,1)$, we can choose $T(\eta)$ close enough to 0 guarantees that \[||H(h(t)) - H(\tilde{h}(t))||_T(\eta) := \sup_{t \in [0,T(\eta)]} \bigg|\int_0^1 G(x) \left( \mu_t^h(dx) - \mu_t^{\tilde{h}}(dx) \right) \bigg| \leq \eta ||h(t) - \tilde{h}(t)||_{T(\eta)},  \]
which tells us that $H(h(t)) : C[0,T] \to C[0,T]$ is a contraction. 
Now we show that $H(h(t))$ is a maps closed balls $B_R$ of radius $R$ to itself. 
We compute
\[ |H(h(t))| =\bigg| \int_0^1 G(x) \mu_t^h(dx) \bigg| \leq G^* | \langle 1, \mu_t \rangle \leq \sup_{||\Psi|_{\infty} \leq 1} \langle \Psi, \mu_t^h \rangle := ||\mu_t^h||_{TV} \] Therefore, to estimate $|H(h(t))|$, it suffices to estimate $||\mu_t^h||_{TV}$. We use the variation of constants formula to compute that
\begin{align*} |\langle \Psi, \mu_t^h \rangle | = \bigg|\langle \Psi, \mu_0^h \rangle + \lambda \int_0^t \langle \left(P_{t-s}\Psi\right) (G(x) - h(t)), \mu_s^h \rangle ds \bigg| \leq ||\mu_0^h||_{TV} + \lambda \left(G^* + ||h||_T \right) \int_0^t \langle P_{t-s} \Psi, \mu_s^h \rangle ds\end{align*} where we used that $P_{t} \Psi(x) = \Psi(\phi_t(x))$ to deduce that $||P_t \Psi||_{\infty} \leq ||\Psi||_{\infty} \leq 1$ and that $\langle P_t \Psi, \mu_0^h \rangle \leq ||\mu_0^h||_{TV}$. Noting further that $ \langle P_{t-s} \Psi, \mu_s \rangle \leq ||P_{t-s} \Psi||_{\infty} \int_0^1 \mu_s^h(dx) \leq ||\Psi||_{\infty} M^T_{h} \leq M^T_{h}$ for $||\Psi||_{\infty} \leq 1$, we find that $$| \langle \Psi, \mu_t \rangle |  ||\mu_0^h||_{TV} + \lambda \left(G^* + ||h||_T \right) M^T_{h} T $$
Beacuse this is true for all test functions $\Psi$ satisfying $||\Psi||_{\infty} \leq 1$, we can combine our inequalities above and the definition of the $||\cdot||_T$ norm to conclude that $$||H(h(t))||_T \leq  ||\mu_0^h||_{TV} + \lambda \left(G^* + ||h||_T \right) M^T_{h} T$$
Choosing $T_{\epsilon} < \left( \lambda \left(G^* + ||h||_T \right) M^T_{h} \right)^{-1} \epsilon$ gives us that \[||H(h(t))||_{T_{\epsilon}} \leq ||\mu_0||_{TV} + \epsilon\], so we see that choosing a ball with radius $R \geq ||\mu_0||_{TV} + \epsilon$, then we see that $H(h(t))$ maps that closed ball $\{h(t) | ||h(t)||_T \leq R\}$ to itself. 

Combining the facts that $H(h(t))$ maps $C[0,T]$ into $C[0,T]$, it is a contraction mapping, and it maps sufficiently large closed balls $B_R$ to themselves, and that $C[0,T]$ is complete with respect to the $||\cdot||_T$ norm, we can apply the Banach fixed-point theorem to show that there exists a unique fixed point $h_{\flat}$ such that  $H(h_{\flat}(t)) = \int_0^1 G(x) \mu_t^{h_{\flat}}(dx) = h_{\flat}(t)$, which further means that $\mu_t^{h_{\flat}}(dx)$ solves Equation \ref{eq:measurepdeexistence} and is unique. Because there exists a $T > 0$ such that the solution $\mu_t(dx)$ exists and is unique, we can use the a priori estimate of Equation \ref{eq:aprioriestimate} to applying a similar argument existence argument starting with an initial population at $\mu_{\frac{T}{2}}(dx)$ and function $h(t)$ definied for $t \in [\frac{T}{2}.\frac{3T}{2}]$ to demonstrate existence of solutions on the time interval $[\frac{T}{2},\frac{3T}{2}]$, and we can continue this iteration to establish the existence of a unique solution $\mu_t(dx)$ to Equation \ref{eq:measurepdeexistence} for any time $t \geq 0$.

Furthermore, because solutions to Equation \ref{eq:linearmeasureexistencepde} satisfy the representation formula \[ \mu^h_t(dx) = w_t^h(x)( \mu_0 \circ \phi_t^{-1})(dx) \: \: \mathrm{or} \: \: \int_0^1 \Psi(x) \mu^h_t(dx) = \int_0^1 \Psi(\phi_t(x)) \exp\left( \lambda \int_0^t \left[ G(\phi_s(x_0)) - h(s) \right] ds \right) \mu_0(dx),    \] 
we can choose $h(t) = h_{\flat}(t)$ and use the fixed point relation $h_{\flat}(t) = H(h_{\flat}(t)) = \langle G(\cdot) \rangle_{\mu^{h_{\flat}}_t}$ to find that
\[ \int_0^1 \Psi(x) \mu_t(dx) = \int_0^1 \Psi(\phi_t(x)) \exp\left( \lambda \int_0^t \left[ G(\phi_s(x_0)) -  \langle G(\cdot) \rangle_{\mu^{h_{\flat}}_s}  \right] ds \right) \mu_0(dx).  \]
Because $\mu_t^{h_{\flat}}(dx)$ is the unique solution to Equation \ref{eq:measurepdeexistence}, we now see that the solution of Equation \ref{eq:measurepdeexistence} satisfies the implicit representation formula $\mu_t^{h_{\flat}}(dx) = w_t(x)( \mu_0 \circ \phi_t^{-1})(dx)$ of Equation \ref{eq:mutimplicit}.
\end{proof}

\section{Integrals Along Simplified Characteristics} \label{sec:integrals}

In this section, we show how to compute the integrals of solutions along the simplified characteristic curves $\Psi_t(k;x_0)$ from the PD and $\Xi_t(k;x_0)$ and $\Pi_t(k;x_0)$ from the HD game. 

\subsection{PD Integrals}

We start with $\Psi_t(k;x)$, the solution for the logistic family serving as the faster and slower characteristics for the within-group replicator dynamics. 

\begin{align*} \int_0^t \Psi_s(k;x_0) ds &=  \int_0^t \left[ \frac{x_0}{x_0 + \left( 1 - x_0 \right) e^{k s}} \right] ds \\ &= x_0 \int_1^{x_0 + \left(1 - x_0\right) e^{k t}} \frac{du}{ k u\left( u - x_0\right)} \: \: \textnormal{(where } u = x_0 + \left( 1 - x_0\right) e^{k s }) \\ &= \frac{1}{k} \int_1^{x_0 + \left(1-x_0\right)e^{k t}} \left(- \frac{1}{u} + \frac{1}{u - x_0} \right) du  \\ &=  \frac{1}{ k} \left[ -\log\left(u \right) + \log\left(u - x_0\right) \right] \bigg|_1^{x_0 + \left(1 - x_0\right) e^{k t}} \\ &=\left[ t  -  \frac{1}{k}  \log\left( x_0 + \left(1 - x_0\right) e^{ k t} \right) \right]  \end{align*} 
Knowing that $x_0$ can be written as $x_0 = \frac{x}{x + (1-x) e^{-kt}}$, we can plug in for $x_0$ above and conclude that
\[ \int_0^t \Psi_s(k;x_0) ds = t - \frac{1}{k} \log\left( \frac{x}{x + (1-x)e^{-kt}} + \left[\frac{(1-x)e^{-kt}}{x + (1-x)e^{-kt}}\right] e^{kt} \right) = t + \frac{1}{k} \log\left(x + (1-x)e^{-kt} \right) \]
For the logistic ODE, we can also compute the integral of $\Psi_t(k.x)$ as 
\begin{align*}  \int_0^t \Psi_s(k;x_0)^2 ds &=  x_0^2 \int_0^t \frac{ds}{\left(x_0 + \left( 1 - x_0 \right) e^{k s}\right)^2} = \frac{x_0^2}{k} \int_1^{x_0 + \left(1 - x_0\right) e^{k t}} \frac{du}{ u^2 \left( u - x_0\right)} \\ &= \frac{1}{k} \int_1^{x_0 + \left(1 - x_0\right) e^{k t}}  \left(\frac{1}{u-x_0} - \frac{1}{u}  - \frac{x_0}{u^2} \right) du \\ &= \frac{1}{k} \left[ \log(u - x_0) - \log(u) + \frac{x_0}{u} \right] \bigg|_{1}^{x_0 + (1 - x_0) e^{k t}} \\ &= t -  \frac{1}{k} \log(x_0 + \left(1-x_0\right)e^{kt}) + \frac{x_0}{k} \left[\frac{1}{x_0 + (1-x_0)e^{kt}} - 1\right]
  \end{align*}
  Using the formula $x_0 = \frac{x}{x + (1-x) e^{-kt}}$ and our expression for $\int_0^t \Psi_t(k;x_0)dt$, we can further see that 
  \begin{align*}  \int_0^t \Psi_s(k;x_0)^2 ds &= t + \frac{1}{k} \log\left(x + (1-x)e^{-kt} \right) + \frac{1}{k} \left[\frac{x}{x + (1-x) e^{-kt}} \right] \left[\frac{x + (1-x)e^{-kt}}{x + \left((1-x) e^{-kt}\right) e^{kt}}  -  1 \right] \\ &=  t + \frac{1}{k} \log\left(x + (1-x)e^{-kt} \right) + \frac{x}{k} \left[ 1 - \frac{1}{x + (1-x)e^{-kt}} \right]  \end{align*}

  \subsection{HD Integrals}
  
  \subsubsection{Dynamics Above Within-Group Equilibrium at $\frac{\beta}{|\alpha|}$}

We first consider $\Xi_t(k,x)$, the faster and slower characteristic curves for the within-group replicator dynamics for the HD games when the level of cooperation exceeds the within-group equilibrium. 
Because we are interested in dynamics in the interval $[\frac{\beta}{|\alpha|},1]$, we can choose a rescaled state variable $X := \frac{|\alpha| x - \beta}{|\alpha| - \beta}$. In terms of our new variable, the ODE from Equation \ref{eq:HDrightequation} takes the following form
\begin{equation} \label{eq:HDrightrescaled} \dsddx{X(t)}{t} = - \left(|\alpha| - \beta\right) k X \left( 1 - X \right) \: \:, \: \: X(0) = X_0 := \frac{|\alpha| x_0 - \beta}{|\alpha| - \beta} \end{equation}
whose solution is given by $\Psi_t\left(\left(|\alpha| - \beta\right) k;X_0\right)$. Because $\Xi_t(k;x_0)$ describes the evolution of $x(t)$ and $\Psi_t\left(\left(|\alpha| - \beta\right) k;X_0\right)$ describes the evolution of $X(t) = \frac{|\alpha| x(t) - \beta}{|\alpha| - \beta}$, we see that we can relate the two named solutions by
\begin{equation} \label{eq:XiintermsofPsi} \Xi_t(k;x_0) = \frac{\beta}{|\alpha|} + \left( \frac{|\alpha| - \beta}{|\alpha|}\right) \Psi_t\left(\left(|\alpha| - \beta\right) k;X_0 \right) %
 \end{equation}
 and we also have that 
 \begin{equation} \label{eq:XisquaredintermsofPsi} \Xi_t(k;x_0)^2 = \frac{\beta^2}{|\alpha|^2} + 2 \frac{\beta}{|\alpha|} \left( \frac{|\alpha| - \beta}{|\alpha|}\right) \Psi_t\left(\left(|\alpha| - \beta\right) k;X_0 \right) + \left( \frac{|\alpha| - \beta}{|\alpha|}\right)^2 \Psi_t\left(\left(|\alpha| - \beta\right) k;X_0 \right)^2  \end{equation}
 Using Equation \ref{eq:XiintermsofPsi} and our result from Equation \ref{eq:psiintegral}, we can compute
\begin{align*} \int_0^t \Xi_s\left(k;x\right) ds &= \int_0^t \left[ \frac{\beta}{|\alpha|} + \left(\frac{|\alpha| - \beta}{|\alpha|} \right) \Psi_s\left(\left(|\alpha| - \beta\right)k;X_0\right) \right] ds  \\ &= \left(\frac{\beta}{|\alpha|}\right) t +  \left(\frac{|\alpha| - \beta}{|\alpha|} \right) \int_0^t \Psi_s\left(\left(|\alpha| - \beta\right)k;X_0\right) ds \\ &= t + \left(\frac{|\alpha| - \beta}{|\alpha|} \right) \left(\frac{1}{\left(|\alpha| - \beta\right) k} \right) \log\left(X + (1-X) e^{-\left(|\alpha| - \beta \right) t} \right) \end{align*}  
Then, using $X = \frac{|\alpha|x - \beta}{|\alpha| - \beta}$, we can deduce that 
\begin{align*} \int_0^t \Xi_s\left(k;x\right) ds &= t + \frac{1}{|\alpha| k} \log\left(\frac{|\alpha| x - \beta + |\alpha| \left(1-x\right) e^{-kt}}{|\alpha| - \beta} \right) \end{align*}
Using Equation \ref{eq:XisquaredintermsofPsi} and the integrals calculated in Equations \ref{eq:psiintegral} and \ref{eq:psisquaredintegral}, we can also see that 
\begin{align*}  \int_0^t \Xi_s\left(k;x\right)^2 ds &= \left(\frac{\beta^2}{|\alpha|^2}\right) t + \frac{2\beta\left(|\alpha| - \beta\right)}{|\alpha|^2} \int_0^t \Psi_s\left(\left(|\alpha| - \beta\right) k; X_0 \right) ds + \left(\frac{|\alpha| - \beta}{|\alpha|} \right)^2 \int_0^t  \Psi_s\left(\left(|\alpha| - \beta\right) k; X_0 \right)^2 ds \\ &= t + \left(\frac{|\alpha| + \beta}{|\alpha|^2 k}  \right)  \log\left(X + (1-X) e^{-\left(|\alpha| - \beta\right)kt} \right) + \left(\frac{|\alpha| - \beta}{|\alpha|^2} \right) \frac{X}{k} \left[ 1 - \frac{1}{X + (1-X) e^{-\left(|\alpha| - \beta\right) k t}} \right] \end{align*}
Using that $X = \frac{|\alpha| x - \beta}{|\alpha| - \beta}$, we are able to see that 
\begin{dmath*}  \int_0^t \Xi_s\left(k;x\right)^2 ds  = t + \left(\frac{|\alpha| + \beta}{|\alpha|^2 k} \right)  \left[ \log\left(|\alpha| x - \beta + |\alpha| (1-x) e^{-\left(|\alpha| - \beta\right) k t}\right) - \log(|\alpha| - \beta) \right] \\ - \frac{1}{|\alpha| k} \left[ \frac{\left(1-x\right)\left( |\alpha| x - \beta \right) \left( 1 - e^{- \left( |\alpha| - \beta \right) k t} \right)}{|\alpha| x - \beta + |\alpha| \left( 1 - x \right) e^{- \left( |\alpha| - \beta \right) k t}} \right]  
 \end{dmath*}

\subsubsection{Dynamics Below $x^{eq} = \frac{\beta}{|\alpha|}$}
For the HD dynamics below the within-group dynamics, we choose rescaled state variable $X = \frac{|\alpha|}{\beta}x$. In terms of the new variable, the within-group dynamics of Equation \ref{eq:HDleftequation} can be written as
\begin{equation} \label{eq:HDleftrescaled} \dsddx{X(t)}{t} = \beta k X\left(1 - X\right) \: \: , \: \: X_0 = \frac{|\alpha|}{\beta} x_0 \end{equation}
Because Equation \ref{eq:HDleftrescaled} can be obtained from Equation \ref{eq:HDrightrescaled} by reversing time, we see that solutions to Equation \ref{eq:HDleftrescaled} are the backwards-in-time solution to the logistic ODE, $\Psi_t^{-1}(k;X_0)$. Therefore we have that $\Pi_t(k;x_0) = \Psi_T^{-1}(\beta k;X_0)$, so we see the equivalence of the characteristic curves
\begin{equation} \label{eq:Pitequivalence} \Pi_t(k;x_0) = \left(\frac{\beta}{|\alpha|}\right) \Psi_t^{-1}\left( \beta k ; X_0 \right)  \end{equation}
 and, using Equation \ref{eq:Psiinvt}, we can now compute the following integrals along the characteristic curves $\Pi_t(k;x_0)$
\begin{align*} \int_0^t \Pi_s(k;x_0) ds &= \frac{\beta}{|\alpha|} \int_0^t \Psi_s^{-1}\left( \beta k ; X_0 \right)  ds \\ &=  \frac{\beta}{|\alpha|} \int_0^t \left[ \frac{X_0}{X_0 + \left( 1 - X_0 \right) e^{-\beta k s}} \right] ds \\ &= - \frac{\beta X_0}{\alpha} \int_1^{X_0 + \left(1 - X_0\right) e^{\beta k t}} \frac{dU}{\beta k U\left( U - X_0\right)} \: \: \textnormal{(where } U = X_0 + \left( 1 - X_0\right) e^{\beta k s }) \\ &= \frac{1}{|\alpha| k} \int_1^{X_0 + \left(1-X_0\right)e^{\beta k t}} \left( \frac{1}{U} - \frac{1}{U - X_0} \right) dU  \\ &=  \frac{1}{|\alpha| k} \left[ \log\left(U \right) - \log\left(U - X_0\right) \right] \bigg|_1^{X_0 + \left(1 - X_0\right) e^{-\beta k t}} \\ &= \frac{1}{|\alpha|} \left[\beta t + \frac{1}{k} \log\left( X_0 + \left(1 - X_0\right) e^{-\beta k t} \right) \right]  \end{align*} 
Now, we can write $X_0$ in terms of $t$ and $X$ using $X_0 = \Psi_t(\beta k ;X) = \frac{X}{X + \left(1-X\right)e^{\beta k t}}$ and use $X = \frac{|\alpha|}{\beta} x$ to tell us that 
\begin{align*} \int_0^t \Pi_s(k;x_0) ds &= \frac{1}{|\alpha|} \left[\beta t + \frac{1}{k} \log\left(  \frac{X}{X + \left(1-X\right)e^{\beta k t}} + \left(  \frac{\left(1-X\right)e^{\beta k t}}{X + \left(1-X\right)e^{\beta k t}} \right) e^{-\beta k t}  \right) \right]  \\ &= \frac{1}{|\alpha|} \left[ \beta t - \frac{1}{k} \log\left(X + \left(1-X\right) e^{\beta k t} \right) \right] \\ &= \frac{1}{|\alpha| k} \log\left( \left(1-X\right) + Xe^{-\beta k t} \right)  \\ &=  - \frac{1}{|\alpha| k} \left[  \log\left( \left(\beta-|\alpha|\right) + |\alpha|e^{-\beta k t} \right) + \log\left(\beta\right) \right]\end{align*}
We also want to calculate
\begin{align*} \int_0^t \Pi_s(k;x_0)^2 ds &= \frac{\beta^2}{|\alpha|^2} \int_0^t \Psi_s^{-1}\left( \beta k ; X_0 \right)^2  ds \\ &=  \frac{\beta^2}{|\alpha|^2} \int_0^t \left[ \frac{X_0^2}{\left(X_0 + \left( 1 - X_0 \right) e^{-\beta k s}\right)^2} \right] ds \\ &= - \frac{\beta^2 X_0^2}{|\alpha|^2} \int_1^{X_0 + \left(1 - X_0\right) e^{\beta k t}} \frac{dU}{\beta k U^2\left( U - X_0\right)} \: \: \textnormal{(where } U = X_0 + \left( 1 - X_0\right) e^{\beta k s }) \\ &= \frac{\beta}{|\alpha|^2 k}  \int_1^{X_0 + \left(1 - X_0\right) e^{-\beta k t}} \left( \frac{X_0}{U^2} + \frac{1}{U} - \frac{1}{U-X_0} \right) dU \\ &= \frac{\beta}{|\alpha|^2 k} \left[\frac{X_0}{U} + \log(U) - \log(U - X_0) \right] \bigg|_1^{X_0 + \left(1-X_0\right) e^{-\beta k t}} \\ &= \frac{\beta}{|\alpha|^2 k} \left[\beta k t + \log\left(X_0 + \left(1-X_0\right)e^{-\beta k t} \right) +  X_0 \left( 1 - \frac{1}{X_0 + \left(1 - X_0 \right) e^{- \beta k t} } \right)  \right] \end{align*}
From Equation \ref{eq:Pitequivalence}, we can also express $X_0$ in terms of $t$ and $X$ as $X_0 = \frac{X}{X+ (1-X)e^{\beta k t}}$. Plugging this in for $X_0$ above lets us write that
\begin{align*} \int_0^t \Pi_s(k;x_0)^2 ds &= \frac{\beta}{|\alpha|^2 k} \left[ \beta k t - \log\left(X + \left(1-X\right) e^{\beta k t} \right)   +  \frac{X(1-X)\left(1 - e^{\beta k t}\right)}{X + (1-X) e^{\beta k t}} \right]  \\ &=  \frac{\beta}{|\alpha|^2 k} \left[  \log\left((1-X) + Xe^{-\beta k t} \right)  +  \frac{X(1-X)\left(1 - e^{\beta k t}\right)}{X + (1-X) e^{\beta k t}} \right] \end{align*}
Then, using $X = \frac{|\alpha|}{\beta}x$, we can finally see that 
\begin{dmath} \int_0^t \Pi_s(k;x_0)^2 ds = - \frac{\beta}{|\alpha|^2 k} \left[ \log\left( \beta - |\alpha| x + |\alpha| x e^{- \beta k t} \right) - \log(\beta) \right] \\ - \frac{1}{|\alpha| k} \left[ \frac{x \left( \beta - |\alpha| x \right) \left(e^{\beta k t} - 1 \right)}{|\alpha| x + \left( \beta - |\alpha| x \right) e^{\beta k t}} \right] \end{dmath}

\bibliography{multilevelselection}
\bibliographystyle{ieeetr}

\end{document}